\documentclass[USletter,10pt]{article}
\usepackage[margin=1in,dvips]{geometry}
\usepackage{graphicx,cancel}
\usepackage{graphics}
\usepackage{color, soul}
\usepackage{amsmath, amsthm, slashed}
\usepackage{amssymb}
\usepackage{rotating}
\usepackage{mathrsfs}
\usepackage{upgreek}
\usepackage{epsfig}
\usepackage{verbatim}
\usepackage[affil-it]{authblk}
\usepackage{rotating}
\usepackage{graphicx}
\usepackage{accents, enumerate}

\usepackage[hidelinks]{hyperref}
\usepackage{float}
\usepackage{tikz-cd} 
\usepackage{pgfplots}
\usetikzlibrary{decorations.pathmorphing}
\usetikzlibrary{decorations.markings}
\usetikzlibrary{patterns}
\usetikzlibrary{shapes}
\usetikzlibrary{plotmarks}

\newtheorem{theorem}{Theorem}[section]
\newtheorem{definition}[theorem]{Definition}

\newtheorem{corollary}[theorem]{Corollary}
\newtheorem{proposition}[theorem]{Proposition}
\newtheorem{lemma}[theorem]{Lemma}
\newtheorem{remark}[theorem]{Remark}

\newtheorem*{conjecture*}{Conjecture}
\newtheorem*{theorem*}{Theorem}
\newtheorem*{corollary*}{Corollary}
\newcommand{\So}{\mathscr{S}}
\newcommand{\Sop}{\tilde{\mathscr{S}}}
\newcommand{\Sf}{\accentset{\land}{\mathscr{S}}}
\newcommand{\Si}{\underaccent{\lor}{\mathscr{S}}}

\newcommand{\ns}{\slashed{\nabla}}
\newcommand{\divs}{\slashed{\mathrm{div}}}
\newcommand{\ds}{\slashed{\Delta}}
\newcommand{\curls}{\slashed{\mathrm{curl}}}
\newcommand{\Dst}{\slashed{\mathcal{D}}_2^\star}
\newcommand{\Dso}{\slashed{\mathcal{D}}_1^\star}

\newcommand{\Ab}{\underline{A}}

\newcommand{\Olin}{\Omega^{-1}\accentset{\scalebox{.6}{\mbox{\tiny (1)}}}{\Omega}}
\newcommand{\Olino}{\accentset{\scalebox{.6}{\mbox{\tiny (1)}}}{\Omega}}
\newcommand{\glinh}{\accentset{\scalebox{.6}{\mbox{\tiny (1)}}}{\hat{\slashed{g}}}}
\newcommand{\glin}{\accentset{\scalebox{.6}{\mbox{\tiny (1)}}}{\slashed{g}}}
  \newcommand{\glinto}{\accentset{\scalebox{.6}{\mbox{\tiny (1)}}}{\sqrt{\slashed{g}}}}
\newcommand{\bmlin}{\accentset{\scalebox{.6}{\mbox{\tiny (1)}}}{b}}

\newcommand{\xblin}{\accentset{\scalebox{.6}{\mbox{\tiny (1)}}}{\underline{\hat{\chi}}}}
\newcommand{\xlin}{\accentset{\scalebox{.6}{\mbox{\tiny (1)}}}{{\hat{\chi}}}}

\newcommand{\eblin}{\accentset{\scalebox{.6}{\mbox{\tiny (1)}}}{\underline{\eta}}}
\newcommand{\elin}{\accentset{\scalebox{.6}{\mbox{\tiny (1)}}}{{\eta}}}
\newcommand{\otx}{\accentset{\scalebox{.6}{\mbox{\tiny (1)}}}{\left(\Omega \mathrm{tr} \chi\right)}}
\newcommand{\otxb}{\accentset{\scalebox{.6}{\mbox{\tiny (1)}}}{\left(\Omega \mathrm{tr} \underline{\chi}\right)}}
\newcommand{\olin}{\accentset{\scalebox{.6}{\mbox{\tiny (1)}}}{\omega}}
\newcommand{\olinb}{\accentset{\scalebox{.6}{\mbox{\tiny (1)}}}{\underline{\omega}}}

\newcommand{\ablin}{\accentset{\scalebox{.6}{\mbox{\tiny (1)}}}{\underline{\alpha}}}
\newcommand{\alin}{\accentset{\scalebox{.6}{\mbox{\tiny (1)}}}{{\alpha}}}

\newcommand{\bblin}{\accentset{\scalebox{.6}{\mbox{\tiny (1)}}}{\underline{\beta}}}
\newcommand{\blin}{\accentset{\scalebox{.6}{\mbox{\tiny (1)}}}{{\beta}}}
\newcommand{\rlin}{\accentset{\scalebox{.6}{\mbox{\tiny (1)}}}{\rho}}
\newcommand{\slin}{\accentset{\scalebox{.6}{\mbox{\tiny (1)}}}{{\sigma}}}
\newcommand{\Klin}{\accentset{\scalebox{.6}{\mbox{\tiny (1)}}}{K}}

\newcommand{\abh}{\hat{\underline{\mathrm{A}}}}

\DeclareMathAlphabet\mathbfcal{OMS}{cmsy}{b}{n}

\title{Uniform Boundedness for Solutions to the Teukolsky Equation on Schwarzschild from Conservation Laws of Linearised Gravity}

\author[1]{Sam C. Collingbourne\thanks{scc@math.columbia.edu}}
\author[2,3]{Gustav Holzegel\thanks{gholzegel@uni-muenster.de}}

\affil[1]{\small Columbia University, Department of Mathematics, \vskip.1pc \  2990 Broadway,~New York,~NY 10027,~USA  \vskip.2pc \ } 
\affil[2]{\small Westf\"alische Wilhelms-Universit\"at M\"unster,
Mathematisches~Institut,~Einsteinstrasse~62,~48149~M\"unster,~Bundesrepublik~Deutschland \vskip.2pc \ }
\affil[3]{\small Imperial College London,
Department of Mathematics,
South~Kensington~Campus,~London~SW7~2AZ,~United~Kingdom}

\begin{document}

\maketitle
\abstract{We consider the equations of linearised gravity on the Schwarzschild spacetime in a double null gauge. Applying suitably commuted versions of the conservation laws derived in earlier work of the second author we establish control on the gauge invariant Teukolsky quantities $\upalpha^{[\pm 2]}$ without any reference to the decoupled Teukolsky wave equation satisfied by these quantities. More specifically, we uniformly bound the energy flux of all first derivatives of $\upalpha^{[\pm 2]}$ along any outgoing cone from an initial data quantity at the level of first derivatives of the linearised curvature and second derivatives of the linearised connection components. Analogous control on the energy fluxes along any ingoing cone is established a posteriori directly from the Teukolsky equation using the outgoing bounds. 
}
\tableofcontents

\section{Introduction}

The linear Teukolsky equation of spin $\pm 2$ associated with the $2$-parameter family of Kerr spacetimes $g_{M,a}$ plays a central role in understanding the dynamical stability properties of the Kerr family of black holes as solutions to the vacuum Einstein equations~\cite{Teukolsky}. In Boyer-Lindquist coordinates it is given by
\begin{align} \label{Teuintro}
 &\Box_{g_{M,a}} \upalpha^{[\pm 2]}+\frac{2(\pm2)}{\rho^2}(r-M)\partial_r\upalpha^{[\pm 2]}+\frac{2(\pm 2)}{\rho^2}\Big[\frac{M(r^2-a^2)}{\Delta}-r-ia\cos\theta\Big]\partial_t\upalpha^{[\pm 2]}\\
 &\qquad\nonumber +\frac{2(\pm 2)}{\rho^2}\Big[\frac{a(r-M)}{\Delta}+\frac{i \cos\theta}{\sin^2\theta}\Big]\partial_{\phi}\upalpha^{[\pm 2]}+\frac{1}{\rho^2}\Big((\pm 2)-(\pm 2)^2\cot^2\theta\Big)\upalpha^{[\pm 2]}=0,
\end{align}
where the quantities $\upalpha^{[\pm2]}$ are spin-weighted functions 
corresponding to the two ``extremal" linearised null Riemann curvature components (with respect to the algebraically special null frame of Kerr) in the full linearisation of the Einstein equations near Kerr. Remarkably, the $\upalpha^{[\pm2]}$ remain invariant both under infinitesimal coordinate diffeomorphisms and under infinitesimal frame changes, i.e.~they only change quadratically in the linearisation parameter. Moreover, it is expected that vanishing of both $\upalpha^{[+2]}$ and $\upalpha^{[-2]}$ reduces the space of solutions of the full system of linearised Einstein equations to (the infinite dimensional family of) infinitesimal diffeomorphisms and frame changes and (the finite dimensional family of) linearised Kerr solutions~\cite{Wald}. These facts, together with the decoupled equation (\ref{Teuintro}) that they satisfy, make quantitative bounds on $\upalpha^{[\pm2]}$ an essential ingredient in many approaches to the stability problem~\cite{DHR,DHR2,DHRT,YakovRita2,ABBSM19,KS20,KSK21, Elena, Gabriele1, Gabriele2}.

\subsection{Previous Results on the Teukolsky Equation}
The difficulty in establishing quantitative estimates for solutions of (\ref{Teuintro}) is rooted in  the equation itself not originating from a Lagrangian, which reflects in the fact that no conserved energy at the level of first derivatives is known for (\ref{Teuintro}). Nevertheless, in the past decade a complete understanding of the global behaviour of solutions to the Teukolsky equation in the form of quantitative bounds in term of initial data has been obtained, culminating in the recent~\cite{YakovRita1, YakovRita2}. See also~\cite{Millet} for precise asymptotic decay rates. Progress started with a complete treatment of the Schwarzschild case  in~\cite{DHR}, where the authors prove (amongst other things) boundedness and integrated decay estimates for the Teukolsky equation. 
The key insight was a physical space version of the transformation theory of Chandrasekhar~\cite{Chandrasekhar}: By applying a second order differential operator, the Teukolsky equation is transformed into a Regge-Wheeler-type equation, which \emph{does} originate from a Lagrangian and for which estimates are known -- in particular, there is a conserved energy!  Estimates for the $\upalpha^{[\pm 2]}$ themselves then follow by interpreting the aforementioned transformation as two successive transport equations. The inherent loss of derivatives can finally be recovered by elliptic estimates using the original Teukolsky equation.\footnote{Note, however, that as soon as the higher order transformation theory is invoked, the most elementary energy boundedness statement for $\upalpha^{[\pm2]}$, that does not lose derivatives, is necessarily at the level of at least three derivatives of $\upalpha^{[\pm 2]}$.} This strategy has since been generalised to the $|a|\ll M$ case~\cite{DHR2, Ma} and, as already mentioned, recently to the full sub-extremal case $|a|<M$~\cite{YakovRita2}. 
It should be stressed that there is a significant increase in complexity in the analysis for $|a|<M$. In that case both the classical Teukolsky-Starobinski identities and the seminal result of Whiting \cite{Whiting}, who proved mode stability for solutions of (\ref{Teuintro}) (i.e.~the absence of solutions which grow exponentially in time) constitute important ingredients in the analysis. See \cite{Yakov2, AMPBReal, Rita, AHWModeLE} for further quantitative mode stability results and a discussion of their role in the stability problem. 

\subsection{Canonical Energy and Conservation Laws of Linearised Gravity}
In view of its non-Lagrangian structure, the Killing symmetries of the Kerr metric do not directly translate into conservation laws for the Teukolsky equation (\ref{Teuintro}). There is nevertheless a sense in which the linearised Einstein equations inherit conservation laws. However, these are necessarily formulated at the level of the \emph{entire} system of gravitational perturbations and hence involve many of the linearised Ricci-coefficients and curvature components, not the Teukolsky quantities alone. The most well-known approach to this is via the canonical energy, which originated in Friedman~\cite{Friedman78} (see also~\cite{Chandrasekhar2,Chandrasekhar3}) and was further developed some years ago by Hollands and Wald~\cite{HollandsWald}. The main difficulty to successfully apply the resulting conservation laws, is to establish appropriate coercivity for the associated canonical energy on spacelike slices, a property that in view of the coupled nature of the Einstein equations is generally difficult to establish, even in a suitably chosen gauge. In full generality, this has remained an open problem that has so far prevented the direct use of the canonical energy for the stability problem.\footnote{On the other hand, the canonical energy has been used successfully to establish certain instability results in the asymptotically anti-de Sitter setting. See~\cite{GHIW16}.}

A resolution of this problem was proposed and carried out in the Schwarzschild case, $a=0$, in~\cite{ghconslaw}. Central to this approach are two conservation laws, which are established ``by direct inspection" from the system of gravitational perturbations expressed in double null gauge. In the next step, these conservation laws are expressed for regions bounded by null cones as indicated in the Penrose diagram below. 

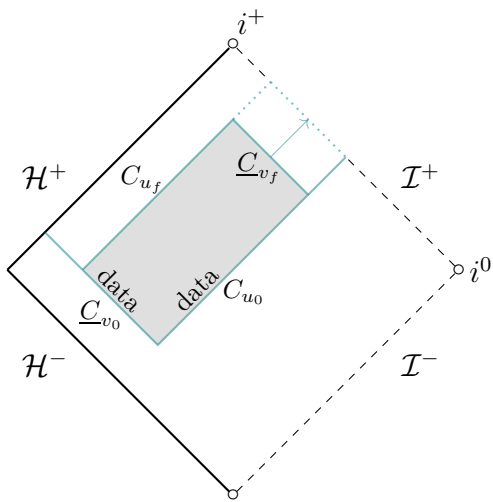
\begin{figure}[H]
\centering
\begin{tikzpicture}
      \draw [thick] (6,3) -- (8.95,5.95);
      \draw [dashed]  (11.95,3.05)--(10.5,4.5);
       \draw [dashed] (9.05,5.95) -- (9.5,5.5);
       \draw [dashed] (9.05,0.05) -- (11.95,2.95);
    \draw [thick] (6,3) -- (8.95,0.05);

    \draw [thick,teal,fill=lightgray,opacity=0.5] plot  coordinates {(7,3)(9,5)(10,4)(8,2)(7,3)};
    
     \draw [thick,teal,opacity=0.5] (6.5,3.5) -- (7,3);
     \draw [dotted, thick,teal,opacity=0.5] (9,5) -- (9.5,5.5);
      \draw [thick,teal,opacity=0.5] (10,4) -- (10.5,4.5);
       \draw [dotted,thick,teal,opacity=0.5] (9.5,5.5) -- (10.5,4.5);
       
           \draw[teal,->,opacity=0.5] (9.5,4.5) -- (10,5);

      
      \node[mark size=2pt,inner sep=2pt] at (12,3) {$\circ$};
      \node[mark size=2pt] at (9,6) {$\circ$};
          \node[mark size=2pt] at (9,0) {$\circ$};

      \node (scriplus) at (11.5,4.25) {\large $\mathcal{I}^{+}$};
        \node (Hplus) at (6.5,4.25) {\large $\mathcal{H}^{+}$};   
        \node (scriplus) at (11.5,1.75) {\large $\mathcal{I}^{-}$};
        \node (Hplus) at (6.5,1.75) {\large $\mathcal{H}^{-}$};   
      \node (iplus) at (9.25,6.3) {\large $i^{+}$};
      \node (inaught) at (12.3,3) {\large $i^{0}$};
       \node at (7.8,4.2) {$C_{u_f}$};
            \node  at (9.15,2.7) {$C_{u_0}$};
            \node at (7.25,2.4) {$\underline{C}_{v_0}$};
            \node at (9.35,4.35) {$\underline{C}_{{v_f}}$}; 
            \node[rotate=45] at (8.5,2.75) {$\mathrm{data}$};
            \node[rotate=-45] at (7.5,2.75) {$\mathrm{data}$};

      \end{tikzpicture}  
\end{figure}

The associated fluxes on these cones are then shown to be gauge invariant up to boundary terms on spheres, a fact that can be used -- by adding an appropriately constructed pure gauge solution normalised to the outgoing cone $C_{u_f}$ -- to establish positivity of the fluxes after sending also the ingoing cone $\underline{C}_{v_f}$ to null infinity.\footnote{One uses here the fact that the Einstein constraint equations become ODEs on null hypersurfaces, which considerably simplifies the task to show coercivity of the fluxes. In addition, exploiting the decay towards null infinity of the linearised quantities (as following from local existence) and a cancellation of terms on the limiting sphere $S^2_{u_f,\infty}$ is paramount. See~\cite{ghconslaw}.} This leads, for instance, to a priori control on the linearised shear of the outgoing null cone $C_{u_f}$ (the shear remaining invariant under the addition of the pure gauge solution) and to a priori control on the flux radiated to null infinity and hence a weak form of linear stability. We remark that the conservation laws central to this approach have recently been shown to be in direct correspondence with the canonical energy~\cite{scthesis,ColCan}. In particular, with every conservation law stated in~\cite{ghconslaw} one can associate a suitably modified canonical energy in the sense of~\cite{HollandsWald}. Moreover, novel conservation laws in the style of~\cite{ghconslaw}  can be derived from suitable modifications of the canonical energy in~\cite{HollandsWald}.

\subsection{The Main Uniform Boundedness Theorem}

While it is expected that the approach outlined in the previous section (using a double null gauge for the linearised equations, expressing the conservation laws on null cones and carrying out the limiting argument at null infinity) also leads to stability statements for axisymmetric perturbations in the Kerr case, we return in this paper to the Schwarzschild case. More specifically, we will revisit the conservation laws proven in~\cite{ghconslaw} to 
provide a proof of boundedness for the Teukolsky quantities, which does not use the Teukolsky equation at all. 

To state the main theorem we need some standard notation that the reader can pick up from the beginning of Section \ref{sec:Preliminaries}. In particular, we employ standard Eddington-Finkelstein $(u,v,\theta,\phi)$-coordinates on Schwarzschild and use the notion of $S^2_{u,v}$-tensors (instead of spin-weighted functions) to formulate the equations in a double null gauge. The quantities $\alin$, $\ablin$ below are therefore symmetric traceless tensors and in direct one-to-one correspondence with the quantities $\upalpha^{[+2]}$, $\upalpha^{[-2]}$ introduced earlier. 

\subsubsection{Controlling fluxes on outgoing cones} \label{sec:introoutgoing}
Our main theorem can be formulated as follows (see Theorem~\ref{theo:gaugeinvariantoutgoing} for a more precise statement):

\begin{theorem} \label{theo:gaugeinvariantoutgoingintro}
Given a smooth solution of the system of gravitational perturbations on the Schwarzschild exterior arising from characteristic initial data on $C_{u_0} \cup \underline{C}_{v_0}$, we have the following flux estimates for the gauge invariant Teukolsky quantities $A= r \Omega^2\alin$ and $\underline{A}= r\Omega^2\ablin$ on outgoing cones $C_{u_f}$ with $u_f \in [u_0, \infty)$: 
\begin{align} \label{finalalphaintro}
\int_{v_0}^{\infty} \int_{S^2} dv  d\theta d\phi \sin \theta  \left[ r^4|\Omega^{-1} \slashed{\nabla}_3 A|^2 + | \Omega \slashed{\nabla}_4 A|^2+ |r\slashed{\nabla} A|^2 + |A|^2 \right] \left(u_f\right)   \lesssim  \mathbb{E}^2_{data} (u_f)  \, ,
\end{align}
\begin{align} \label{finalalphabintro}
 \int_{v_0}^{\infty} \int_{S^2} dv  d\theta d\phi \sin \theta \Omega^2 \left[  r^2 | \Omega \slashed{\nabla}_4\underline{A}|^2+ | \slashed{\nabla}\underline{A}|^2 + \frac{1}{r^2} |\underline{A}|^2 \right]  \left(u_f\right)  \lesssim \mathbb{E}^2_{data} (u_f) \, ,
\end{align}
where $\mathbb{E}^2_{data} (u_f)$ is an initial data energy on $\left(C_{u_0} \cap \{u \leq u_f\} \right) \cup \underline{C}_{v_0}$ with $\sup_{u_f \in [u_0,\infty)} \mathbb{E}^2_{data} (u_f)<\infty$ for regular initial data.
\end{theorem}
We immediately remark that the fact that the \emph{transversal} derivative of $A$ is also appearing in the outgoing flux (\ref{finalalphaintro}) arises from the fact that there is a Bianchi equation relating that derivative to linearised curvature and Ricci coefficients whose flux we do control.

The energy $ \mathbb{E}^2_{data} (u_f)$  appearing on the right hand side of (\ref{finalalphaintro})--(\ref{finalalphabintro}) is defined in (\ref{energyinmaintheorem}) below. It contains bounds on initial data quantities on the cones $C_{u_0} \cap \{u \leq u_f\}$ and $\underline{C}_{v_0}$, which involve only first derivatives of linearised curvature components and up to second derivatives of connection coefficients.\footnote{While $ \mathbb{E}^2_{data} (u_f)$ is non-negative and uniformly bounded for smooth asymptotically flat data, it is not necessarily monotone in $u_f$. This is because it involves both (non-manifestly coercive) fluxes along $C_{u_0} \cup \underline{C}_{v_0}$ and bounds on the sphere $S^2_{u_f,v_0}$.} Hence, from the point of view of regularity of the full system of gravitational perturbations, there is no loss of regularity in the estimates (\ref{finalalphaintro})--(\ref{finalalphabintro}). It is in this sense that Theorem~\ref{theo:gaugeinvariantoutgoingintro} improves previous boundedness results for the Teukolsky equation in terms of regularity: Those results relied on the Chandrasekhar transformation theory, which requires certain third derivatives of $\alin$, $\ablin$ to be bounded initially even to establish boundedness of the first order fluxes appearing in Theorem~\ref{theo:gaugeinvariantoutgoingintro}.

We note that unlike (\ref{finalalphaintro}), the estimate (\ref{finalalphabintro}) degenerates strongly towards the horizon since it is $\Omega^{-2} \ablin$ and $\Omega^2 \alin$ which extend regularly to the event horizon. This degeneration is expected given that we are proving boundedness of the $T$-energy for the system of gravitational perturbations. If one is willing to invoke the Teukolsky equation itself, one can exploit the redshift (at the cost of including a weighted initial data energy for $\Omega^{-2} \ablin$ on the ingoing data cone $\underline{C}_{v_0}$) and obtain improved estimates leading finally to uniform boundedness of $\Omega^{-2} r\ablin$ as desired. See Corollary \ref{cor:main} below for more details.

\subsubsection{Controlling fluxes on ingoing cones}
Having bounded the fluxes on arbitrary \emph{outgoing} cones as in Theorem \ref{theo:gaugeinvariantoutgoingintro} it is easiest to obtain boundedness of the fluxes on arbitrary \emph{ingoing} cones from energy estimates for the Teukolsky equation.\footnote{The authors have so far been unable to prove an analogous ``ingoing” version of Theorem \ref{theo:gaugeinvariantoutgoingintro} directly from the conservation laws. This may or may not suggest that the structure of null infinity is essential to the argument.} We present the result for the quantity $A$ here. We first define the ingoing characteristic flux naturally associated with the wave equation for $A$ (see (\ref{TeuA})) and determined by a globally timelike vectorfield which behaves like the timelike Killing field $T$ asymptotically:
\begin{align}
\mathbb{F} [A] (v)= \int_{u_0}^\infty \int_{S^2} du d\theta d\phi \Omega^2 \sin \theta \left[ |\Omega^{-1} \slashed{\nabla}_3 A|^2 + |\slashed{\nabla} A|^2 + \frac{1}{r^2} |A|^2 \right] \, .
\end{align}
Using the estimate (\ref{finalalphaintro}) of Theorem \ref{theo:gaugeinvariantoutgoingintro} in conjunction with an energy estimate for $A$ we then obtain:
\begin{corollary}
Under the assumptions of Theorem \ref{theo:gaugeinvariantoutgoingintro} we have also the following estimate
\begin{align}
\sup_{v \geq v_0} \mathbb{F} [A] (v) \lesssim \mathbb{F} [A] (v_0) + \sup_{u_f \in [u_0, \infty)} \mathbb{E}^2_{data} (u_f) \, .
\end{align}
\end{corollary}
See Section \ref{sec:ingoing} for the analogous statement for $\underline{A}$ and the detailed proofs.

\subsubsection{Pointwise boundedness}
We finally observe that the estimate (\ref{finalalphaintro}) immediately implies
\begin{align} 
\sup_{u \geq u_0, v \geq v_0} \int_{S^2_{u,v}} \sin \theta d\theta d\phi |\Omega^2 \alin r|^2 \lesssim  \sup_{u_f \in [u_0, \infty)} \mathbb{E}^2_{data} (u_f) \label{coji} \, .
\end{align}
 From the fact that the equations of linearised gravity are invariant under suitable commutations with angular momentum operators one may obtain angular commuted versions of the estimates (\ref{finalalphaintro}) and (\ref{coji}) at the cost of a higher order angular commuted energies on the right hand side. Commuting twice in this way one obtains from (\ref{coji}) an $L^\infty$-estimate on $\Omega^2 \alin r$ through Sobolev embedding on the sphere. As this is standard, we will omit the details. There are analogous statements for $\ablin$, which a priori degenerate near the horizon but can be improved using the redshift effect. See Section \ref{sec:ubab}.

\subsection{The Main Ideas of the Proof}
We now indicate the main ideas to prove the estimates of Theorem~\ref{theo:gaugeinvariantoutgoingintro} focussing on the estimate (\ref{finalalphaintro}). We provide a discussion supplemented by some schematic formulae (using notation introduced later in the paper) hoping that readers familiar with the double null formalism can directly grasp the underlying ideas. The main steps are as follows:

\begin{enumerate}[(1)]
\item Apply the conservation laws and the coercivity argument of~\cite{ghconslaw} to appropriate angular commuted equations of linearised gravity (this exploits the spherical symmetry of the background) to estimate additional quantities \emph{in the gauge adapted to the cone $C_{u_f}$}, i.e.~where coercivity is manifest. This leads, for instance, to control of angular derivatives of $\xlin$, $\elin$, $\eblin$ on the outgoing cone in the notation of this paper. 

\item Apply the conservation laws and the coercivity argument of~\cite{ghconslaw} to the $T$-commuted equations of linearised gravity (this exploits the stationarity of the background) to estimate additional quantities \emph{in the gauge adapted to the cone $C_{u_f}$}. This leads, for instance, to control of the flux of $\slashed{\nabla}_T \xlin$ in the notation of this paper. The main difficulty in this step is that in the gauge adapted to the cone, coercivity of the outgoing flux is no longer manifest as certain cross-terms appearing in the flux no longer vanish after commutation with $T$ (since $T$ is not tangent to $C_{u_f}$). However, the bad terms can be controlled invoking the fluxes from Step 1 after inserting null structure equations.

\item Use the relevant transport equations in the linearised null structure equations to estimate additional quantities \emph{in the gauge adapted to the cone $C_{u_f}$}. For instance, $\xblin$ can be estimated on the outgoing cone $C_{u_f}$ using only the fluxes from Step 1 above.

\item Combine the control on quantities in the gauge adapted to the outgoing cone $C_{u_f}$ to estimate the \emph{gauge invariant} quantities $\alin$ and $\ablin$. For instance, we can write schematically in the notation of this paper (ignoring any weights in $r$ and at the horizon):
\[
 \alin \sim \Omega \slashed{\nabla}_4 \xlin + \xlin \sim \slashed{\nabla}_T \xlin + \xlin + \slashed{\nabla}\elin + \xblin \, , 
\]
where we have inserted the null structure equation for $\Omega \slashed{\nabla}_3 \xlin$ and used $2T = \Omega \slashed{\nabla}_3 + \Omega \slashed{\nabla}_4$. Since we control the flux of all quantities in the gauge adapted to the cone $C_{u_f}$ on the right, we control the flux of $\alin$ in the gauge adapted to the cone $C_{u_f}$ and hence -- by gauge invariance -- in any gauge. 
\end{enumerate} 
Of course keeping track of the weights in $r$ and $\Omega^2=1-\frac{2M}{r}$ as well as regularity is one of the technical challenges in the proof. Moreover, the above argument does not produce control on derivatives of $\alin$. These are obtained in a similar fashion using Bianchi and null structure equations with control on $T$-derivatives of certain curvature components on the outgoing cone.

\subsection{Some Open Directions}

From the  point of view of the stability problem, an immediate question is whether Theorem~\ref{theo:gaugeinvariantoutgoingintro} can be used in conjunction with the Teukolsky equation to derive an integrated decay estimate at the level of first derivatives of the Teukolsky quantities. Related to this, one may also ask whether the way to exploit the conversation laws presented here can be used to improve on (the regularity of) the scattering theory for the Teukolsky equations as developed in~\cite{MasaoodI,MasaoodII, Pham}, which so far also relies on the Chandrasekhar transformations. 

A further interesting direction is to try to generalise the argument to spacetime dimensions larger than four. Since the conservation laws persist in higher dimensions (but the decoupling of the Teukolsky quantities does not) following the strategy of this paper may lead to genuinely new boundedness results.

Finally, one would like to generalise the results to the case of axisymmetric perturbations of the Kerr metric, where in view of the absence of superradiance, positivity of the fluxes can reasonably be expected. There are various difficulties starting from the angular operators no longer commuting trivially and the required algebraic manipulations being much more challenging. We leave this for future work.

\subsection{Acknowledgements}
G.H.~acknowledges support by the Alexander von Humboldt Foundation in the framework of the Alexander von Humboldt Professorship endowed by the Federal Ministry of Education and Research, funding through ERC Consolidator Grant 772249 as well as Germany’s Excellence Strategy EXC 2044 390685587, Mathematics M\"unster: Dynamics–Geometry–Structure. The authors thank Gabriele Benomio for helpful comments on the manuscript.

\section{Preliminaries} \label{sec:Preliminaries}
We only provide the minimal set-up here to make the paper self-contained. The reader should consult~\cite{DHR, ghconslaw} for details and more background.

\subsection{The Schwarzschild Background}
\subsubsection{Coordinates}
We cover the exterior region of the Schwarzschild spacetime by double null Eddington--Finkelstein coordinates $(u,v,\theta,\phi)\in \mathbb{R}^2\times{S}^2_{u,v}$ with 
\begin{align*}
u=\frac{1}{2}(t-r_{\star}),\qquad v=\frac{1}{2}(t+r_{\star}),\qquad r_{\star}(r)\doteq (r-4M) + 2 M \ln\Big(\frac{r-2M}{2M}\Big) \, , 
\end{align*}
where $(t,r,\theta,\phi)$ are the standard Schwarzschild coordinates. The Schwarzschild metric in $(u,v,\theta,\phi)$ coordinates is
\begin{align}\label{eqn:SchDN}
g=-2\Omega^2(u,v)(du\otimes dv+dv\otimes du)+\slashed{g},\qquad \slashed{g}=r^2(u,v){\gamma},\qquad \Omega^2(u,v)=1-\frac{2M}{r(u,v)},
\end{align}
where~${{\gamma}}$ is the standard metric on the unit~$2$-sphere. In these coordinates constant $v$ or $u$ level sets, which we denote $\underline{C}_{v}$ and $C_u$, are ingoing or outgoing null hypersurfaces, respectively. Strictly speaking, the coordinates $(u,v,\theta,\phi)$ do not cover the future event horizon, $\mathcal{H}^+$, or future null infinity, $\mathcal{I}^+$. However, formally we may parameterise the future event horizon as $(\infty,v,\theta,\phi)$ and future null infinity as $(u,\infty,\theta,\phi)$.

\subsubsection{Null Frames and $S^2_{u,v}$-tensors}
We recall the standard future directed ingoing and outgoing null directions,
\begin{align*}
e_3\doteq\frac{1}{\Omega}\partial_u,\quad e_4\doteq\frac{1}{\Omega}\partial_v
\end{align*}
which may be completed to a local null frame (a basis for the tangent space at each point) by choosing a local frame $\{e_A\}_{A=1,2}$ tangent to the double null spheres $S^2_{u,v}$.

A major role in this paper will be played by (covariant) tensors on $S^2_{u,v}$, called 
$S^2_{u,v}$-tensors in short. A covariant $S^2_{u,v}$-tensor may be identified with a covariant spacetime tensor that vanishes when it acts on the vector fields $e_3$ or $e_4$. Two types of $S^2_{u,v}$-tensors are particularly important in this work, namely $S^2_{u,v}$ one-forms and $S^2_{u,v}$ symmetric-traceless $2$-tensors.

\subsubsection{Background Ricci and Curvature Quantities}
Using notation standard from~\cite{CK93}, we record the non-vanishing connection coefficients of Schwarzschild with respect to the above null frame. They are
\begin{align*}
\omega&\doteq-\frac{\Omega}{2}g(\nabla_{e_4}e_4,e_3)=\frac{M}{r^2},& \underline{\omega}&\doteq-\frac{\Omega}{2}g(\nabla_{e_3}e_3,e_4)=-\frac{M}{r^2},\\
\Omega\mathrm{tr}\chi&\doteq\Omega \slashed{g}^{AB}g(\nabla_{e_A}e_4,e_B)=\frac{2\Omega^2}{r},& \Omega\mathrm{tr}\underline{\chi}&\doteq\Omega \slashed{g}^{AB}g(\nabla_{e_A}e_3,e_B)=-\frac{2\Omega^2}{r} .
\end{align*}
The only non-vanishing null Weyl curvature component is
\begin{align*}
\rho\doteq\frac{1}{4}\mathrm{Riem}(e_3,e_4,e_3,e_4)=-\frac{2M}{r^3} \, .
\end{align*}
Finally, for $(S^2_{u,v},\slashed{g})$ we denote the connection coefficients by $\slashed{\Gamma}_{AB}^C$ and the Gaussian curvature by $K=\frac{1}{r^2}$.

\subsubsection{Killing Fields}
The Schwarzschild spacetime is static. We denote the Killing vector field associated to staticity with $T$, which in double null Eddington Finkelstein coordinates given by
\begin{align*}
T=\frac{1}{2}(\partial_u+\partial_v)=\frac{\Omega}{2}(e_3+e_4).
\end{align*}

Additionally, the Schwarzschild spacetime is spherically symmetric, so $SO(3)$ acts by isometry. We  define a basis of angular momentum operators $\Omega_1, \Omega_2, \Omega_3$ which generate the Lie algebra $\mathfrak{so}(3)$ by demanding that in the standard coordinates $(\theta,\phi)$ on the sphere we have
\begin{align} \label{angularoperators}
\Omega_1=\sin\phi\partial_{\theta}+\cot\theta\cos\phi\partial_{\phi},\qquad\Omega_2=\cos\phi\partial_{\theta}-\cot\theta\sin\phi\partial_{\phi},\qquad \Omega_3=\partial_{\phi}.
\end{align}

\subsubsection{Norms and Inner Products}
For $\Theta,\Phi\in\Gamma\big(\bigotimes^n_{i=1}TS^2_{u,v}\big)$ we define
\begin{align*}
\langle \Theta,\Phi\rangle&\doteq\slashed{g}^{A_1B_1}...\slashed{g}^{A_nB_n}\Theta_{A_1...A_n}\Phi_{B_1...B_n},\qquad
|\Theta|^2\doteq\langle\Theta,\Theta\rangle
\end{align*}
and the following $L^2$-inner product on the spheres $S^2_{u,v}$:
\begin{align*}
\langle \Theta,\Phi\rangle_{u,v}\doteq\int_{S^2}\langle \Theta,\Phi\rangle(u,v,\theta,\phi)\varepsilon_{S^2},
\end{align*}
where $\varepsilon_{S^2}$ denotes the induced volume form on $S^2$. This induces the $L^2$-norm on the sphere $S^2_{u,v}$, 
\begin{align*}
||\Theta||_{u,v}\doteq ||r^{-1}\Theta||_{S^2_{u,v}}=\sqrt{\int_{S^2}|\Theta|^2(u,v,\theta,\phi)\varepsilon_{S^2}}.
\end{align*}

\subsubsection{Differential Operators on $S^2_{u,v}$}
Let $\Theta$ be a $p$-covariant $S^2_{u,v}$-tensor field. We define the projected covariant derivative $\ns_3$, $\ns_4$, $\ns_A$ by extending $\Theta$ trivially to a spacetime $p$-covariant tensor field and projecting onto the local frame $\{e_A\}$ for $S_{u,v}^2$. One can then check that the following formulas hold:
\begin{align*}
\ns_3\Theta_{A_1...A_p}&=e_3(\Theta_{A_1...A_p})-\frac{1}{2}\mathrm{tr}\underline{\chi}\Theta_{A_1...A_p} \, , \\
\ns_4\Theta_{A_1...A_p}&=e_4(\Theta_{A_1...A_p})-\frac{1}{2}\mathrm{tr}{\chi}\Theta_{A_1...A_p} \, , \\
\ns_A\Theta_{A_1...A_p}&=\partial_A(\Theta_{A_1...A_p})-\slashed{\Gamma}_{AA_i}^B\Theta_{A_1\dots\hat{A}_iB\dots A_p}.
\end{align*}
Additionally, we define
\begin{align*}
\ns_T\Theta&\doteq\frac{\Omega}{2}\big(\ns_3\Theta+\ns_4\Theta\big).
\end{align*}
We define the projected Lie derivatives $\slashed{\mathcal{L}}_T$, $\slashed{\mathcal{L}}_{\Omega_k}$ as
\begin{align*}
(\slashed{\mathcal{L}}_T\Theta)_{A_1...A_p}&\doteq(\mathcal{L}_T\Theta)_{A_1...A_p}\qquad(\slashed{\mathcal{L}}_{\Omega_k}\Theta)_{A_1...A_p}\doteq({\mathcal{L}}_{\Omega_k}\Theta)_{A_1...A_p}
\end{align*}
and note that
\begin{align*}
(\slashed{\mathcal{L}}_T\Theta)_{A_1...A_p}=(\ns_T\Theta)_{A_1...A_p}=T(\Theta_{A_1...A_p}).
\end{align*}
We define the $(p-1)$-covariant tensor fields $\slashed{\mathrm{div}}\Theta$ and $\slashed{\mathrm{curl}}\Theta$ as
\begin{align*}
(\slashed{\mathrm{div}}\Theta)_{A_1...A_{p-1}}&\doteq\slashed{g}^{BC}(\slashed{\nabla}_B\Theta)_{CA_1...A_{p-1}},\qquad
(\slashed{\mathrm{curl}}\Theta)_{A_1...A_{p-1}}\doteq\slashed{\varepsilon}^{BC}(\slashed{\nabla}_B\Theta)_{CA_1...A_{p-1}},
\end{align*}
where $\slashed{\varepsilon}$ is the induced volume form on ${S}^2_{u,v}$. 

For a ${S}_{u,v}^2$ one-form $\xi$ we define
\begin{align*}
\slashed{\mathcal{D}}_2^{\star}\xi&\doteq-\frac{1}{2}\Big[(\slashed{\nabla}_A\xi)_B+(\slashed{\nabla}_B\xi)_A-(\slashed{\mathrm{div}}\xi)\slashed{g}_{AB}\Big] \, ,
\end{align*}
which is the formal $L^2$-adjoint of $\divs$. For $f_1,f_2\in C^{\infty}(S^2_{u,v})$ we define
\begin{align*}
    [\slashed{\mathcal{D}}_1^{\star}(f_1,f_2)]_A&\doteq-\slashed{\nabla}_Af_1+\slashed{\varepsilon}_{AB}\slashed{\nabla}^Bf_2,
\end{align*}
which is the formal $L^2$ adjoint of $\slashed{\mathcal{D}}_1$, which maps an $S^2_{u,v}$-$1$-form to the pair of functions $(\divs\xi,\curls\xi)$.

Lastly, we denote the Laplacian on the $S^2_{u,v}$ sphere by
\begin{align*}
\ds\Theta\doteq\divs(\ns\Theta). 
\end{align*}
Note that the Laplacian on the \textit{unit} sphere $S^2$, which we denote as $\Delta_{S^2}$, is related by $\Delta_{S^2}=r^2\ds$.

\subsubsection{Support on $\ell\geq 2$ Spherical Harmonics}

In this work, we will deal with functions, $S^2_{u,v}$ one-forms and $S^2_{u,v}$ symmetric-traceless $2$-tensors that have support on $\ell\geq 2$. For a function $f$ to be supported on $\ell\geq 2$ this simply means that $f$ has vanishing projection on the $\ell=0,1$ spherical harmonics.
For a $S^2_{u,v}$ one-form, $\xi$, we note that there exists unique decomposition into two functions $f,g$ supported on $\ell\geq 1$ such that
\begin{align*}
\xi=r\Dso(f,g).
\end{align*}
A $S^2_{u,v}$ one-form is supported on $\ell\geq 2$ if $f$ and $g$ in this unique representation are supported on $\ell\geq 2$.

For a $S^2_{u,v}$ symmetric-traceless $2$-tensor, $\Theta$, we note that there exists unique decomposition into two functions $f,g$ supported on $\ell\geq 2$ such that
\begin{align*}
\Theta=r^2\Dst\Dso(f,g).
\end{align*}
Therefore, $S^2_{u,v}$ symmetric-traceless $2$-tensors are automatically supported on $\ell\geq 2$. 

We remind the reader if a function on $S^2_{u,v}$ has vanishing spherical mean (as is the case for functions supported on $\ell\geq 2$) then the Poincar\'e inequality holds
\begin{align*}
\sqrt{2} \|f\|_{u,v}\leq \|[r\ns]f\|_{u,v}.
\end{align*} 
This will be used liberally throughout this work. 

\subsubsection{Angular Identities}
\begin{lemma}\label{lem:basicangid}
We have the following identities for the angular momentum operators defined in (\ref{angularoperators}):
\begin{align*}
 \sum_k(\slashed{\mathrm{curl}}\Omega_k)^2=4,\qquad \sum_k(\slashed{\mathrm{curl}}\Omega_k)\Omega_k=0,\qquad \sum_{k}\Omega_k^A\Omega_k^B&=r^2\slashed{g}^{AB}.
\end{align*}

\end{lemma}
\begin{proof}
Direct computation noting that 
$
\curls\Omega_1=-2\cos\phi\sin\theta, \curls\Omega_2=2\sin\theta\sin\phi,\curls\Omega_3=2\cos\theta.
$
\end{proof}
\begin{proposition}\label{prop:angident}
We have the following identities for Lie derivatives of products of $S^2_{u,v}$ tensor fields: 
\begin{align*}
\sum_{k=1}^3\langle\slashed{\mathcal{L}}_{\Omega_k}\Theta_1,\slashed{\mathcal{L}}_{\Omega_k}\Theta_2\rangle&=\langle[r\ns]\Theta_1,[r\ns]\Theta_2\rangle \ \ \ \ \textrm{ \ \ \ \ \ \ \ \ \phantom{Xxi} \ if  $\Theta_1,\Theta_2$ are functions,}
\\
\sum_{k=1}^3\langle\slashed{\mathcal{L}}_{\Omega_k}\Theta_1,\slashed{\mathcal{L}}_{\Omega_k}\Theta_2\rangle&=\langle[r\ns]\Theta_1,[r\ns]\Theta_2\rangle+\langle\Theta_1,\Theta_2\rangle  \ \ \ \ \phantom{x} \textrm{if $\Theta_1,\Theta_2$ are $1$-forms,}
\\
\sum_{k=1}^3\langle\slashed{\mathcal{L}}_{\Omega_k}\Theta_1,\slashed{\mathcal{L}}_{\Omega_k}\Theta_2\rangle&=\langle[r\ns]\Theta_1,[r\ns]\Theta_2\rangle+4\langle\Theta_1,\Theta_2\rangle  \ \ \ \ \textrm{if $\Theta_1,\Theta_2$ are symmetric-traceless $2$-tensors.}
\end{align*}
\end{proposition}
\begin{proof}
Since $\Omega_k$ are Killing, one can compute that for a $S^2_{u,v}$ covariant tensor field of rank $n$,
\begin{align*}
\slashed{\mathcal{L}}_{\Omega_k}\Theta_{A_1...A_n}&=\ns_{\Omega_k}\Theta_{A_1...A_n}+\frac{1}{2}\sum_{i=1}^n(\curls\Omega_k){\slashed{\varepsilon}_{A_i}}^B\Theta_{A_1...\hat{A}_iB...A_n} \, .
\end{align*}
The result now follows after using Lemma~\ref{lem:basicangid} and the basic identity 
$
{\slashed{\varepsilon}_{A_i}}^B{\slashed{\varepsilon}_{C_j}}^D=\slashed{g}_{A_iC_j}\slashed{g}^{BD}-\slashed{\delta}_{A_i}^D\slashed{\delta}_{C_j}^B \, .
$
\end{proof}

\subsection{The Equations of Linearised Gravity}\label{sec:lineqs}
The system of linearised gravitational perturbations in double null gauge is encoded by the linearised metric quantities
\begin{align} \label{metricq}
\glinto, \glinh_{AB}, \bmlin_A,\Olin,
\end{align}
the linearised connection coefficients
\begin{align} \label{connectionq}
\otx,\otxb,\olin,\olinb,\elin_A,\eblin_A, \xlin_{AB},\xblin_{AB} ,
\end{align}
and the linearised curvature components
\begin{align} \label{curvatureq}
\rlin,\slin,\blin_A,\bblin_A, \alin_{AB},\ablin_{AB}.
\end{align}
As in~\cite{ghconslaw}, we will speak of a solution $\mathscr{S}$ to the system of gravitational perturbations to mean a a collection of quantities 
\begin{align} \label{scollect}
\mathscr{S}=\bigg(\, \glinh \, , \, \glinto \, , \, \Olino \, , \,  \bmlin\, , \,  \otx \, , \,  \otxb\, , \,  \xlin\, , \, \xblin\, , \,  \eblin \, , \,  \elin \, , \, \olin \, , \,  \olinb \, , \,  \alin \, , \,  \blin \, , \,  \rlin \, , \,  \slin \, , \,  \bblin \, , \,  \ablin \, , \, \Klin \bigg)
\end{align}
satisfying the system (\ref{stos})--(\ref{Bianchi10}) below, which we call the system of linearised gravity on the Schwarzschild background.\footnote{The quantity $\Klin$ appearing in the above is the linearised Gaussian curvature of the double null spheres and given in terms of $\glin$, $\glinto$ by $\Klin=-\frac{1}{2}\ds \mathrm{tr}_{\slashed{g}} \glin +\divs \divs\glinh - \frac{1}{r^2}\mathrm{tr}_{\slashed{g}} \glin$ where $\glinto=\frac{1}{2} 
\sqrt{\slashed{g}} \cdot \mathrm{tr}_{\slashed{g}} \glin$. See~\cite{DHR}.} It consists of the following equations (see \cite{DHR} for a derivation):
\begin{align} \label{stos} 
\partial_u \bigg(\frac{\glinto}{\sqrt{\slashed{g}}}\bigg)  &= \otxb
\qquad , \qquad 
\partial_v \bigg(\frac{\glinto}{\sqrt{\slashed{g}}}\bigg) = \otx -  \divs\, \bmlin  \, ,\\ 
\sqrt{\slashed{g}}\, \partial_u \bigg( \frac{\glinh_{AB}}{\sqrt{\slashed{g}}} \bigg)   &=2\Omega\, \xblin_{AB}
\ \ \ ,  \ \ \
\sqrt{\slashed{g}}\, \partial_v \bigg( \frac{\glinh_{AB}}{\sqrt{\slashed{g}}} \bigg) =2\Omega\, \xlin_{AB} + 2\big(\slashed{\mathcal{D}}_2^\star \bmlin \big)_{AB}  ,\label{stos2}
\end{align}
\begin{align} \label{bequat} 
\partial_u \bmlin^A = 2 \Omega^2\big(\elin - \eblin\big)^A \, ,
\end{align}
\begin{align} \label{oml3}
\partial_v \big( \Olin \big) = \olin \textrm{ \ , \ }  \partial_u \big(\Olin\big)=\olinb \textrm{ \  , \ }   2 \slashed{\nabla}\big(\Olin\big) = \elin + \eblin.
\end{align}
\begin{align} 
\partial_v \otxb  &= \Omega^2 \Big( 2 \divs\, \eblin + 2\rlin + 4 \rho \, \Olin \Big) - \frac{1}{2}  \Omega \mathrm{tr}\chi \Big( \otxb - \otx  \Big) ,\label{dtcb}\\
\partial_u \otx & = \Omega^2 \Big( 2 \divs\, {\elin} + 2 \rlin + 4\rho \, \Olin \Big) - \frac{1}{2}  \Omega \mathrm{tr}\chi \Big( \otxb - \otx  \Big) ,\label{dbtc}
\end{align}
\begin{align} 
\partial_v \otx& = - \left(\Omega \mathrm{tr}\chi\right)\otx + 2 \omega \otx  + 2  \left(\Omega \mathrm{tr}\chi \right) \olin ,\label{uray}\\
\partial_u \otxb &= - \left(\Omega \mathrm{tr} \underline{\chi}\right) \otxb  + 2 \underline{\omega} \otxb + 2  \left(\Omega \mathrm{tr} \underline{\chi} \right) \olinb , \label{vray}
\end{align}
\begin{equation} \label{tchi} 
\begin{split}
\slashed{\nabla}_3  \big(\Omega^{-1} \xblin  \big)  +  \Omega^{-1} \left(\mathrm{tr} \underline{\chi}\right) \xblin = -\Omega^{-1} \ablin \, , \\
\slashed{\nabla}_4  \big(\Omega^{-1} \xlin \big)  +  \Omega^{-1} \left(\mathrm{tr}{\chi}\right) \xlin= -\Omega^{-1} \alin   \, ,
\end{split}
\end{equation}
\begin{align} \label{chih3}
\slashed{\nabla}_3  \left(\Omega \xlin \right)  + \frac{1}{2} \left(\Omega \mathrm{tr} \underline{\chi}\right) \xlin + \frac{1}{2} \left( \Omega \mathrm{tr}\chi\right) \xblin  &= -2 \Omega \slashed{\mathcal{D}}_2^\star \elin \, , \\
\slashed{\nabla}_4  \left(\Omega \xblin  \right) + \frac{1}{2} \left(\Omega \mathrm{tr}\chi \right) \xblin  + \frac{1}{2} \left( \Omega \mathrm{tr} \underline{\chi}\right) \xlin &= -2 \Omega \slashed{\mathcal{D}}_2^\star \eblin \, . \label{chih3b}
\end{align}
\begin{equation}
\begin{split}
\divs \xblin = -\frac{1}{2} \left(\mathrm{tr} \underline{\chi}\right)  \elin + \bblin + \frac{1}{2\Omega} \slashed{\nabla} \otxb , \\
\divs \xlin = -\frac{1}{2} \left( \mathrm{tr} {\chi}\right) \eblin  -\blin + \frac{1}{2\Omega} \slashed{\nabla} \otx \label{ellipchi} \, .
\end{split}
\end{equation}
\begin{align}
\slashed{\nabla}_3 \eblin& =  \frac{1}{2} \left(\mathrm{tr} \underline{\chi}\right) \left( \elin - \eblin\right)  + \bblin\quad ,\quad
\slashed{\nabla}_4 \elin =  -  \frac{1}{2} \left( \mathrm{tr} {\chi}\right) \left( \elin - \eblin\right) - \blin , \label{propeta}\\
\slashed{\nabla}_4 \eblin &= \frac{2}{\Omega}\ns\olin-\mathrm{tr}\chi\eblin  + \blin \quad\quad ,\quad
\slashed{\nabla}_3 \elin =  \frac{2}{\Omega}\ns\olinb-\mathrm{tr}\underline{\chi}\elin  - \bblin , \label{propeta2}
\end{align}
\begin{align} \label{curleta}
\curls \elin = \slin \ \ \ , \ \ \ \curls \eblin = -\slin \, .
\end{align}
\begin{align} 
\partial_v \olinb = -\Omega^2 \Big(\rlin + 2 \rho \Olin \Big) \, ,\label{oml1}\\
\partial_u \olin = -\Omega^2 \Big(\rlin + 2 \rho \Olin \Big) \, ,\label{oml2}
\end{align}
\begin{equation} \label{lingauss}
\Klin = -\rlin - \frac{1}{4} \frac{\mathrm{tr} {\chi}}{\Omega}\Big( \otxb - \otx  \Big) +\frac{1}{2}\Olin \left(\mathrm{tr}\chi \mathrm{tr} \underline{\chi} \right) \, ,
\end{equation}
along with the linearised Bianchi identities,
\begin{align}
\slashed{\nabla}_3 \alin + \frac{1}{2} \mathrm{tr} \underline{\chi}\alin + 2 \Omega^{-1} \underline{\omega} \alin &= -2 \slashed{\mathcal{D}}_2^\star \blin - 3 \rho\, \xlin \, ,  \label{Bianchi1} \\
\slashed{\nabla}_4 \blin + 2 (\mathrm{tr}\chi) \blin - \Omega^{-1}{\omega} \blin &= \divs\, \alin \, , \label{Bianchi2} \\
\slashed{\nabla}_3 \blin + (\mathrm{tr} \underline{\chi}) \blin + \Omega^{-1}\underline{{\omega}} \blin &= \slashed{\mathcal{D}}_1^\star \left(-\rlin \, , \, \slin \, \right) + 3\rho \, \elin \, ,   \label{Bianchi3}\\
\slashed{\nabla}_4 \rlin + \frac{3}{2} (\mathrm{tr}\chi) \rlin& = \divs\, \blin - \frac{3}{2} \frac{\rho}{\Omega}  \otx \, , \label{Bianchi4}\\
\slashed{\nabla}_3 \rlin + \frac{3}{2} (\mathrm{tr} \underline{\chi}) \rlin& = -\divs\, \bblin - \frac{3}{2} \frac{\rho}{\Omega} \otxb \, , \label{Bianchi5}\\
\slashed{\nabla}_4 \slin + \frac{3}{2} (\mathrm{tr}\chi) \slin&= -\curls\, \blin \, , \label{Bianchi6} \\
\slashed{\nabla}_3 \slin + \frac{3}{2} (\mathrm{tr} \underline{\chi}) \slin &= -\curls\, \bblin \, , \label{Bianchi7} \\
\slashed{\nabla}_4 \bblin + (\mathrm{tr}\chi)  \bblin + \Omega^{-1}{\omega} \bblin &= \slashed{\mathcal{D}}_1^\star \left(\rlin \, ,  \, \slin \, \right) - 3 \rho\, \eblin  \, ,  \label{Bianchi8} \\
\slashed{\nabla}_3 \bblin + 2 (\mathrm{tr} \underline{\chi})  \bblin - \Omega^{-1} {\underline{\omega}} \bblin &= - \divs\, \ablin \, , \label{Bianchi9} \\
\slashed{\nabla}_4 \ablin + \frac{1}{2} (\mathrm{tr}\chi) \ablin + 2 \Omega^{-1}{\omega} \ablin &=  2 \slashed{\mathcal{D}}_2^\star \bblin - 3 \rho\, \xblin \, .  \label{Bianchi10}
\end{align}

\subsection{Pure Gauge Solutions} \label{sec:pg}
There exist special solutions to the system of gravitational perturbations above which correspond to infinitesimal coordinate transformations preserving the double null form of the metric. These are called \emph{pure gauge solutions} of the system of gravitational perturbations. A particular subset of them is identified in the following lemma, which is proven as Lemma~6.1.1 of~\cite{DHR}. Recall the notation $\Delta_{S^2}=r^2 \slashed{\Delta}$, so $\Delta_{S^2}$ is the Laplacian on the unit sphere with metric $\gamma$.

\begin{lemma} \label{lem:exactsol} 
For any smooth function $f=f\left(v,\theta,\phi\right)$, the following is a pure gauge solution of the system of gravitational perturbations:
\begin{align}
\frac{\Olino}{\Omega} &= \frac{1}{2\Omega^2} \partial_v \left(\Omega^2f \right)  , & \glinh&= - \frac{4}{r}  [r\slashed{\mathcal{D}}_2^\star][r \slashed{\nabla}] f  \ , &\frac{\glinto}{\sqrt{\slashed{g}}} &= \frac{2\Omega^2 f}{r} +  2\Delta_{S^2}\frac{f}{r}  , \nonumber \\
\bmlin &= -2r[r \slashed{\nabla}] \left[ \partial_v \left(\frac{f}{r}\right)\right] , & \elin &= \frac{\Omega^2}{r^2} [r \slashed{\nabla}] f  , & \eblin &= \frac{1}{\Omega^2}  [r\slashed{\nabla}] \left[\partial_v \left(\frac{\Omega^2}{r}f\right) \right]   \nonumber \, ,
\nonumber \\
\xblin &= -\frac{2\Omega}{r^2} [r \slashed{\mathcal{D}}_2^\star] [r\slashed{\nabla}] f  , & \otx &= 2 \partial_v \left(\frac{f \Omega^2}{r}\right)  , & \otxb &=  \frac{2\Omega^2}{r^2} \left[\Delta_{\mathbb{S}^2} f +  \Big(1-\frac{4M}{r}\Big)f \right]  , \nonumber \\
\rlin &= \frac{6M \Omega^2}{r^4} f  , & \bblin&= \frac{6M\Omega}{r^4} [ r \slashed{\nabla} ]f , & \Klin &= -\frac{\Omega^2}{r^3}\left(\Delta_{\mathbb{S}^2} f + 2f\right) \nonumber
\end{align}
and
\[
\xlin = \alin = \ablin = 0 \ \ \ , \ \ \ \blin = 0 \ \ \ , \ \ \ \slin = 0 \nonumber \, .
\]
We will call $f$ a gauge function. 
\end{lemma}

\begin{remark}
One can write down the analagous lemma for pure gauge solutions generated by a smooth function $f=f(u,\theta,\phi)$ (see Lemma 6.1.2 in~\cite{DHR}).
\end{remark}

\subsection{The Class of Solutions}
We now define the class of (smooth) solutions of the system of gravitational perturbations that we wish to consider in this paper.  We first recall from~\cite{ghconslaw} the notion of a partially initial data normalised solution, see Definition 3.2 of~\cite{ghconslaw}:

\begin{definition}
We call $\So$ a {\bf partially initial data normalised solution supported on $\ell \geq 2$} of the system of gravitational perturbations if $\So$ is supported on $\ell \geq 2$ and the initial data on $C_{u_0} \cup \underline{C}_{v_0}$ satisfies 
\begin{itemize} 
\item 
$
\otx (\infty,v_0,\theta,\phi) = 0 \textrm{\ and \ } (\divs \elin + \rlin ) (\infty,v_0,\theta,\phi) = 0 \ \ \textrm{on the sphere $S^2_{\infty,v_0}$,}
$
\item
$
r^2 \Klin (u_0, \infty, \theta,\phi)=0  .
$
\end{itemize}
\end{definition}
We note that the first two conditions on the sphere $S^2_{\infty,u_0}$ are evolutionary along the event horizon while the second one is evolutionary along null infinity. Furthermore, for a partially initial data normalised solution, the quantity (and angular derivatives thereof)
\begin{align} \label{f0defi}
\Omega^{-2} \left[ 2\Olin -\frac{r}{2\Omega^2} \otx  \right]  \left(u,v_0, \theta,\phi\right)
\end{align}
extends regularly to the event horizon in the limit $u \rightarrow \infty$. This is Proposition~3.1 of~\cite{ghconslaw}. 

We also recall Definition 3.4 of~\cite{ghconslaw}, which we modify slightly here to include also higher derivatives in (\ref{decreten2}), as in this paper we will need to study up to two commutations of the equations.

\begin{definition} \label{def:extendskri}
We call a solution $\mathscr{S}$ {\bf extendible to null infinity} if the following weighted quantities of $\mathscr{S}$ have well-defined finite limits on null infinity\footnote{In particular, for every $u \geq u_0$ fixed the limit of the quantity along the null cone $C_{u}$ as $v \rightarrow \infty$ is well-defined.} for some $0<s<1$
\begin{equation}\label{decreten} 
\begin{split}
r^{3+s} \alin \ , \ r^{3+s}\blin \ , \  r^3\rlin  , \  r^3 \slin  , \  r^2 \bblin \ , \  r\ablin \ , \ r^3 \Klin \, ,  \\
 r^2 \xlin \ , \  r\xblin \ , \ r \elin \ , \  r^2 \eblin \ , \  r^2 \divs  \elin \, , \, r^3 \divs  \eblin  \ , \  r^{2+s} \olin \ , \  \olinb \ , \   r^2 \otx  \ , \  r \otxb  \ , \ \Olin  \, .
\end{split}
\end{equation}
In addition, denoting an arbitrary representative of the quantities in (\ref{decreten}) by $\mathcal{Q}$, we demand that also $[r\slashed{\nabla}]^i [\slashed{\nabla}_T]^j \mathcal{Q}$ has a well-defined limit on null infinity for all $i+j 
\leq 2$, and that for any fixed $u_f$ with $u_0<u_f<\infty$ the estimate
\begin{align} \label{decreten2}
\sum_{i+j\leq 2} \sup_{\left[u_0,u_f\right] \times \{ v\geq v_0\} \times S^2} |[r\slashed{\nabla}]^i [\slashed{\nabla}_T]^j \mathcal{Q}| \leq C\left[u_f\right]
\end{align}
holds with the constant $C\left[u_f\right]$ depending only on $u_f$ (and the initial data) but not on $v$.
\end{definition}

In this paper we will consider {\bf smooth solutions of the system of gravitational perturbations that are partially initial data normalised supported on $\ell \geq 2$ and extendible to null infinity.} 

\begin{remark}
It follows from the well-posedness theory of the system of gravitational perturbations ~\cite{DHR, ghconslaw} that there is essentially no restriction in considering this class of solutions. In particular, given a general solution to the system of gravitational perturbations arising from asymptotically flat (to sufficiently high order) initial data, we can make the solution partially initial data normalised by adding a pure gauge solution which is computable directly from the initial data and the solution thus constructed will be extendible to null infinity.
\end{remark}

For completeness, we collect also the weighted quantities which extend regularly to the event horizon $\mathcal{H}^+$ for our class of solutions. The quantity $[\Omega^{-1} \slashed{\nabla}_3]^i [\Omega \slashed{\nabla}_4]^j  \left[r \slashed{\nabla}\right]^k \widetilde{\mathcal{Q}}$ extends regularly to $\mathcal{H}^+ \cap \{v \geq v_0\}$ for any $i+j+k \geq 0$ provided
\begin{align} \label{hozextendQ}
\widetilde{\mathcal{Q}} \in \Big\{\frac{\Olino}{\Omega},\otx,\frac{\otxb}{\Omega^2},\olin,\frac{\olinb}{\Omega^2},\elin,\eblin, \Omega\xlin,\frac{\xblin}{\Omega},\rlin,\slin,\Klin,\Omega\blin,\frac{\bblin}{\Omega},\Omega^2\alin,\frac{\ablin}{\Omega^2}\Big\}.
\end{align}

Finally, in view of (\ref{decreten}) we will sometimes allow ourselves to refer to quantities $\xi$ ``on null infinity" formally as $\xi (u,\infty,\theta,\phi)$ or refer to the sphere ``$S^2_{u_f,\infty}$" instead of carrying the corresponding limits through the already complicated computations. The procedure is always to carry out all computations for finite $v_f$ and to take the limit $v_f \rightarrow \infty$ in the very end when (\ref{decreten}) can be exploited. All this is always performed for a compact $u$ range, so that no issues with commuting limits ever arise. 


\subsection{The Conservation Laws}

In this section we recall the conservation laws from~\cite{ghconslaw}. To state them in a uniform manner that allows for commutation with Killing vectors, we first define, 
for given $i, j \in \mathbb{N} \cup \{0\}$ and any tuple ${\bf {k}}=(k_1,k_2, \ldots k_j)$ with each $k_s \in \{1,2,3\}$, the (Killing) commutation operator 
\begin{align}
\mathcal{K}^{i,{\bf k}}  = \mathcal{K}^{i,(k_1,\ldots,k_j)} := [\slashed{\mathcal{L}}_T]^i\slashed{\mathcal{L}}_{\Omega_{k_1}}\ldots \slashed{\mathcal{L}}_{\Omega_{k_j}} \, .
\end{align}
Note that $\mathcal{K}^{i,(k_1,\ldots,k_j)}$ maps a $(0,k)$ covariant $S^2_{u,v}$-tensors to a $(0,k)$ covariant $S^2_{u,v}$-tensor. We denote by $|{\bf k}|=j$ the length of the tuple ${\bf k}$.

Let now $\mathscr{S}$ be a smooth solution of the system of gravitational perturbations. For any $u_0\leq u_1<u_2 \leq \infty$ and $v_0\leq v_1 < v_2 \leq \infty$
we define the (commuted) fluxes
\begin{align} \label{vflux}
F^{i,{\bf k}}_v \left[\Gamma , \mathscr{S} \right] \left(u_1,u_2\right)= \int_{u_1}^{u_2}\int_{S^2} \bigg[  & |\mathcal{K}^{i,{\bf k}}\Omega\xblin|^2+2\Omega^2 |\mathcal{K}^{i,{\bf k}}\elin|^2 -2 \mathcal{K}^{i,{\bf k}}\olinb \cdot \mathcal{K}^{i,{\bf k}}\otx - \frac{1}{2} [\mathcal{K}^{i,{\bf k}}\otxb]^2  \nonumber \\
&- \frac{4M}{r^2} \Big(\mathcal{K}^{i,{\bf k}}\frac{\Olino}{\Omega}\Big)\cdot \mathcal{K}^{i,{\bf k}}\otxb  \bigg]r^2du\varepsilon_{S^2} \, ,
\end{align}
\begin{align} \label{uflux}
F^{i,{\bf k}}_u \left[\Gamma, \mathscr{S} \right] \left(v_1,v_2\right)= \int_{v_1}^{v_2}  \int_{S^2} \bigg[&|\mathcal{K}^{i,{\bf k}}\Omega\xlin|^2+2 |\mathcal{K}^{i,{\bf k}}\Omega\eblin|^2   -2\mathcal{K}^{i,{\bf k}}\olin \cdot \mathcal{K}^{i,{\bf k}}\otxb - \frac{1}{2} [\mathcal{K}^{i,{\bf k}} \otx]^2 \nonumber \\
&+ \frac{4M}{r^2} \Big(\mathcal{K}^{i,{\bf k}}\frac{\Olino}{\Omega}\Big) \cdot \mathcal{K}^{i,{\bf k}} \otx  \bigg]r^2dv\varepsilon_{S^2}\, .
\end{align}
Additionally, we denote
\begin{align} \label{vflux2}
F^{i, {\bf k}}_v \left[\mathcal{R},\mathscr{S} \right] \left(u_1,u_2\right)= \int_{u_1}^{u_2} \int_{S^2} \Big[ &\frac{1}{2} r^4 |\mathcal{K}^{i,{\bf k}} \Omega\bblin|^2+\frac{1}{2} \Omega^2 r^4 \left(|\mathcal{K}^{i,{\bf k}}\rlin|^2 + |\mathcal{K}^{i,{\bf k}}\slin|^2\right) -3Mr |\mathcal{K}^{i,{\bf k}}\Omega\elin|^2  \\ &+ 3Mr \mathcal{K}^{i,{\bf k}}\olinb  \cdot \mathcal{K}^{i,{\bf k}}\otx 
-3M \left(1-\frac{4M}{r}\right)   \mathcal{K}^{i,{\bf k}}\otxb \cdot \mathcal{K}^{i,{\bf k} }\Big(\frac{\Olino}{\Omega}\Big) \Big]du\varepsilon_{S^2}, \nonumber
\end{align}
\begin{align} \label{uflux2}
F^{i, {\bf k}}_u \left[\mathcal{R}, \mathscr{S} \right] \left(v_0,v\right)= \int_{v_0}^v \int_{S^2} \Big[& \frac{1}{2} r^4  |\mathcal{K}^{i,{\bf k}} \Omega\blin|^2+\frac{1}{2} \Omega^2 r^4 \left(|\mathcal{K}^{i,{\bf k}}\rlin|^2 + |\mathcal{K}^{i,{\bf k}}\slin|^2\right) -3Mr  |\mathcal{K}^{i,{\bf k}}\Omega\eblin|^2 \\
&+ 3Mr \mathcal{K}^{i,{\bf k}}\olin \cdot \mathcal{K}^{i,{\bf k}}\otxb
 +3M \Big(1-\frac{4M}{r}\Big)  \mathcal{K}^{i,{\bf k}} \otx \cdot \mathcal{K}^{i,{\bf k}} \Big(\frac{\Olino}{\Omega}\Big)   \Big]dv\varepsilon_{S^2} \, . \nonumber
\end{align}

The above fluxes are related by the following conservation law:
\begin{proposition} \label{prop:conslaw2}
For any $u_0<u_1< u_2<\infty$ and $v_0 < v_1<v_2 < \infty$ we have the conservation laws
\begin{align}
F^{i,{\bf k}}_v\left[\mathcal{R},\mathscr{S} \right] \left(u_0,u_1\right)+ F^{i,{\bf k}}_u \left[\Gamma, \mathscr{S} \right] \left(v_0,v_1\right) = F^{i,{\bf k}}_{v_0} \left[\Gamma, \mathscr{S} \right] \left(u_0,u_1\right) +  F^{i,{\bf k}}_{u_0}\left[\Gamma, \mathscr{S} \right] \left(v_0,v_1\right) \, ,
\end{align}
and
\begin{align}
F^{i,{\bf k}}_v\left[\mathcal{R}, \mathscr{S} \right] \left(u_0,u_1\right)+ F^{i,{\bf k}}_u \left[\mathcal{R},\mathscr{S}\right] \left(v_0,v_1\right) = F^{i,{\bf k}}_{v_0} \left[\mathcal{R},\mathscr{S}\right] \left(u_0,u_1\right) +  F^{i,{\bf k}}_{u_0}\left[\mathcal{R},\mathscr{S}\right] \left(v_0,v_1\right) \, .
\end{align}
\end{proposition}

\begin{proof}
Direct computation as in~\cite{ghconslaw}. Note that the projected Lie-derivatives commute trivially through all equations so the proof is identical for the commuted versions.
\end{proof}

\begin{remark}
One can simplify the fluxes in~\eqref{vflux} and~\eqref{uflux} to
\begin{align*}
F_{v}[\Gamma,\So](u_0,u_1)&=\check{F}_{v}[\Gamma,\So](u_0,u_1)+\left[\frac{r^3}{2\Omega^2}\otxb\otx\right]_{(u_0,v)}^{(u_1,v)},\\
F_{u}[\Gamma,\So](v_0,v_1)&=\check{F}_{u}[\Gamma,\So](v_0,v_1)-\left[\frac{r^3}{2\Omega^2}\otxb\otx\right]_{(u,v_0)}^{(u,v_1)},
\end{align*}
where we use the notation $[X]_{a}^b=X(b)-X(a)$, with
\begin{align*}
\check{F}_v^{i,{\bf k}}[\Gamma,\So](u_0,u_1)&\doteq \int_{u_0}^{u_1}\int_{S^2}  \Big[|\mathcal{K}^{i,{\bf k}}\Omega\xblin|^2+2|\mathcal{K}^{i,{\bf k}}\Omega\elin|^2-r\mathcal{K}^{i,{\bf k}}\otxb\mathcal{K}^{i,{\bf k}}(\rlin+\divs\elin)\Big]r^2du\varepsilon_{S^2},\\
\check{F}_u^{i,{\bf k}}[\Gamma,\So](v_0,v_1)&\doteq \int_{v_0}^{v_1}\int_{S^2} \Big[|\mathcal{K}^{i,{\bf k}}\Omega\xlin|^2+2|\mathcal{K}^{i,{\bf k}}\Omega\eblin|^2+r\mathcal{K}^{i,{\bf k}}\otx\mathcal{K}^{i,{\bf k}}(\rlin+\divs\eblin)\Big]r^2du\varepsilon_{S^2}\, .
\end{align*}
In view of Proposition~\ref{prop:conslaw2} one has the analogous conservation law for $\check{F}$. Whilst this is a remarkably simple form of the conservation law it runs into issues with obtaining the optimal regularity result. Indeed, if one is not concerned with regularity, one could use this conservation law to obtain uniform boundedness.  
\end{remark}

\subsubsection{Changing Gauge on the Outgoing Cone}

We next recall from~\cite{ghconslaw} how the outgoing fluxes transform under the addition of a pure gauge solution of the type in Lemma~\ref{lem:exactsol}, namely that they are invariant up to boundary terms on spheres:

\begin{proposition} \label{prop:gaugechange}
Let $\mathscr{S}$ be a solution of the system of gravitational perturbations. Let $f\left(v,\theta,\phi\right)$ be a smooth gauge function generating a pure gauge solution $\mathscr{G}$ of the system of gravitational perturbations as in Lemma~\ref{lem:exactsol}. Finally, set $\mathscr{S} = \tilde{\mathscr{S}} + \mathscr{G}$ thereby defining a new solution $\tilde{\mathscr{S}}$. Then the flux on fixed constant-$u$ hypersurfaces satisfies
\begin{align}
F^{i,{\bf k}}_u \left[\Gamma, \mathscr{S} \right] \left(v_0,v\right)&= F^{i,{\bf k}}_u \big[\Gamma, \tilde{\mathscr{S}}\big] \left(v_0,v\right) + \int_{S^2} \left( \mathcal{G}_\Gamma^{i,{\bf k}} \left( u,v,\theta,\phi\right)-\mathcal{G}_\Gamma^{i,{\bf k}} \left( u,v_0,\theta,\phi\right)\right)\varepsilon_{S^2} \nonumber\\
F^{i,{\bf k}}_u [\mathcal{R}, \mathscr{S} ] \left(v_0,v\right)&= F^{i,{\bf k}}_u \big[\mathcal{R}, \tilde{\mathscr{S}}\big] \left(v_0,v\right) + \int_{S^2} \left( \mathcal{G}_{\mathcal{R}}^{i,{\bf k}} \left(u,v, \theta,\phi\right)-\mathcal{G}_{\mathcal{R}}^{i,{\bf k}} \left( u,v_0,\theta,\phi\right)\right)\varepsilon_{S^2} \nonumber
\end{align}
with
\begin{align} 
\mathcal{G}_\Gamma^{i,{\bf k}}  &=  6M\Big|\mathcal{K}^{i,{\bf k}} \frac{\Omega^2 f}{r}\Big|^2 
+ \mathcal{K}^{i,{\bf k}} \Big(\frac{f \Omega^2}{r}\Big)\Big[ 2r^3 \mathcal{K}^{i,{\bf k}} \left( {\rlin}-\divs \, {\eblin} \right)_{\Sop}
-\Big(1-\frac{4M}{r}\Big) \frac{r^2}{\Omega^2}  \mathcal{K}^{i,{\bf k}}{\otx}_{\Sop}\Big]\label{finge}  \\
&\quad  - \frac{ r^3}{2\Omega^{2}}  \Big[\mathcal{K}^{i,{\bf k}}\otx_{\So}-\mathcal{K}^{i,{\bf k}}{\otx}_{\tilde{\mathscr{S}}} \Big]\mathcal{K}^{i,{\bf k}}\otxb_{\So} \, , \nonumber\\
\mathcal{G}_{\mathcal{R}}^{i,{\bf k}} &= 3M r \Big[ \mathcal{K}^{i,{\bf k}}\Big(\frac{\Olino}{\Omega}\Big)_{\So} \cdot \mathcal{K}^{i,{\bf k}}\otxb_{\So} - \mathcal{K}^{i,{\bf k}}\Big(\frac{\Olino}{\Omega}\Big)_{\Sop} \cdot \mathcal{K}^{i,{\bf k}}{\otxb}_{\Sop} \Big]+ \frac{3M}{2} \Big| r\ns \mathcal{K}^{i,{\bf k}}\frac{f \Omega^2}{r} \Big|^2 \label{finge2} \\
&\quad- \frac{6M^2}{r}  \Big|\mathcal{K}^{i,{\bf k}} \frac{f\Omega^2}{r}\Big|^2 +3M\Big[\frac{4M}{r}   \mathcal{K}^{i,{\bf k}} \Big(\frac{\Olino}{\Omega}\Big)_{\Sop}- \frac{M}{\Omega^2} \mathcal{K}^{i,{\bf k}}{\otx}_{\Sop}- r^2 \mathcal{K}^{i,{\bf k}} \left(\divs {\elin} + {\rlin} \right)_{\Sop}  \Big]\mathcal{K}^{i,{\bf k}}\Big(\frac{f\Omega^2}{r}\Big)
  \nonumber \, .
\end{align}
Here the subscripts $\So$ or $\Sop$ indicate whether the geometric quantity is associated with the solution $\So$ or $\Sop$. In other words, the difference of the fluxes in the old and in the new gauge is a pure boundary term. 
\end{proposition}

\begin{remark}
There is an analogous proposition to Proposition~\ref{prop:gaugechange} for the incoming cone. 
\end{remark}

\subsection{Useful Identities Involving the Vector Field $T$}
We finally collect some identities that follow directly from the system of gravitational perturbations and will be used later in the proof. In fact, the reader may already anticipate their use as the expressions below will appear naturally in the cross-terms of the fluxes (\ref{vflux})--(\ref{uflux2}).

\begin{proposition}\label{prop:TTransport}
We have the following general identities for $T$ applied to $\mathscr{S}$:
{
\allowdisplaybreaks
\begin{align*}
    2T\otx&=2\Omega^2(\divs\elin+\rlin)+4\Omega^2\rho\Big(\frac{\Olino}{\Omega}\Big)-\frac{\Omega^2}{r}\Big[\otxb+\otx\Big]+\frac{2M}{r^2}\otx+\frac{4\Omega^2}{r}\olin,\\
2T\otxb&=2\Omega^2(\divs\eblin+\rlin)+4\Omega^2\rho\Big(\frac{\Olino}{\Omega}\Big)+\frac{\Omega^2}{r}\Big[\otxb+\otx\Big]-\frac{2M}{r^2}\otxb-\frac{4\Omega^2}{r}\olinb,\\
2T\rlin&=\Omega\divs(\blin-\bblin)-\frac{3}{2}\rho\Big[\otx+\otxb\Big],\\
2\ns_T\bblin&=\Omega\Dso(\rlin,\slin)+\frac{6M}{r^3}\Omega\eblin+\frac{2}{r}\Big(1-\frac{3M}{r}\Big)\bblin-\divs(\Omega\ablin),\\
2T\olin&=\partial_v\olin-\Omega^2\Big[\rlin+2\rho\Big(\frac{\Olino}{\Omega}\Big)\Big],\\
2T\Big(\frac{\Olino}{\Omega}\Big)&=\olin+\olinb,
\end{align*}
}
and
\begin{align} 
2\slashed{\nabla}_T \elin &=  -\Omega \blin - \Omega \bblin + \frac{\Omega^2}{r} \left(\elin + \eblin\right)  + 2 \slashed{\nabla} \olinb \label{id1} \,  ,  \\
 \label{id2}
2\slashed{\nabla}_T \eblin &= \Omega \blin + \Omega \bblin - \frac{\Omega^2}{r} \left(\elin + \eblin\right)  +2 \slashed{\nabla} \olin \, , 
\end{align}
which combine to give
\begin{align}
2\slashed{\nabla}_T \left( \elin + \eblin\right) &=2 \slashed{\nabla} \olin+ 2 \slashed{\nabla} \olinb\label{id2b}  \, ,\\
2T(\rlin-\divs\eblin)&=-2\Omega\divs\bblin-2\ds\olin+\frac{\Omega^2}{r}\divs(\elin+\eblin)-\frac{3}{2}\rho\Big[\otx+\otxb\Big],\nonumber\\
2T(\rlin+\divs\elin)&=-2\Omega\divs\bblin+2\ds\olinb+\frac{\Omega^2}{r}\divs(\elin+\eblin)-\frac{3}{2}\rho\Big[\otx+\otxb\Big],\nonumber\\
2T(\rlin+\divs\eblin)&=2\Omega\divs\blin+2\ds\olin-\frac{\Omega^2}{r}\divs(\elin+\eblin)-\frac{3}{2}\rho\Big[\otx+\otxb\Big]\, .\nonumber
\end{align}

We have the following general identities for $T^2$ applied to $\mathscr{S}$:
\begin{align*}
2T^2\otx&=2\Omega^2T(\divs\elin+\rlin)+\frac{2\Omega^2}{r^2}\Big(1-\frac{4M}{r}\Big)\olinb+\frac{2\Omega^2}{r}\Big[\partial_v\olin-2\Omega^2\rlin\Big]-\frac{4\Omega^2\rho}{r}\Big(2-\frac{5M}{r}\Big)\Big(\frac{\Olino}{\Omega}\Big)\\
&\quad -\frac{\Omega^2}{r}\Big[\Omega^2\divs(\elin+\eblin)+\frac{2\Omega^2}{r}\olin\Big]-\frac{2M}{r^3}\Big(1-\frac{3M}{r}\Big)\otx+\frac{2M\Omega^2}{r^2}(\divs\elin+\rlin)\\
&=\frac{2\Omega^2}{r^2}\Big(1-\frac{4M}{r}\Big)\olinb+\frac{2\Omega^2}{r}\Big[\partial_v\olin-2\Omega^2\rlin\Big]-\frac{4\Omega^2\rho}{r}\Big(2-\frac{5M}{r}\Big)\Big(\frac{\Olino}{\Omega}\Big)-2\Omega^2\divs\Omega\bblin+2\Omega^2\ds\olinb\\
&\quad -\frac{2\Omega^4}{r^2}\olin-\frac{2M}{r^3}\Big(1-\frac{3M}{r}\Big)\otx+\frac{2M\Omega^2}{r^2}(\divs\elin+\rlin)-\frac{3\Omega^2}{2}\rho\Big[\otx+\otxb\Big],
\end{align*}
\begin{align*}
2T^2\otxb&=2\Omega^2T(\divs\eblin+\rlin)+\frac{2\Omega^2}{r^2}\Big(1-\frac{4M}{r}\Big)\olin -\frac{2\Omega^2}{r}\Big[\partial_u\olinb-2\Omega^2\rlin\Big]+\frac{4\Omega^2\rho}{r}\Big(2-\frac{5M}{r}\Big)\Big(\frac{\Olino}{\Omega}\Big)\\
&\quad+\frac{\Omega^2}{r}\Big[\Omega^2\divs(\elin+\eblin)-\frac{2\Omega^2}{r}\olinb\Big]-\frac{2M}{r^3}\Big(1-\frac{3M}{r}\Big)\otxb-\frac{2M\Omega^2}{r^2}(\divs\eblin+\rlin)\\
&=\frac{2\Omega^2}{r^2}\Big(1-\frac{4M}{r}\Big)\olin -\frac{2\Omega^2}{r}\Big[\partial_u\olinb-2\Omega^2\rlin\Big]+\frac{4\Omega^2\rho}{r}\Big(2-\frac{5M}{r}\Big)\Big(\frac{\Olino}{\Omega}\Big)+2\Omega^2\divs\Omega\bblin+2\Omega^2\ds\olin\\
&\quad-\frac{2\Omega^4}{r^2}\olinb-\frac{2M}{r^3}\Big(1-\frac{3M}{r}\Big)\otxb-\frac{2M\Omega^2}{r^2}(\divs\eblin+\rlin)-\frac{3\Omega^2}{2}\rho\Big[\otx+\otxb\Big],
\end{align*}
\begin{align*}
4T^2\Big(\frac{\Olino}{\Omega}\Big)&=-2\Omega^2\Big[\rlin+2\rho\Olin\Big]+\partial_v\olin+\partial_u\olinb,\\
4T^2\olin&=\partial_v^2\olin+\Omega^2\Big[\rho\Big(\Big[\frac{3}{r}-\frac{8M}{r^2}\Big]\Big(\frac{\Olino}{\Omega}\Big)-2(2\olin+\olinb)\Big)+\frac{3M}{r^3}\Big(2\otx+\otxb\Big)\Big]\\
&\qquad+\Omega^2\Big[\frac{1}{r}\Big(3-\frac{2M}{r}\Big)\rlin-\Omega\divs(2\blin-\bblin)\Big]
\end{align*}
and
\begin{align*}
2T^2(\rlin-\divs\eblin)&=-2\Omega\divs\ns_T\bblin-\ds \Big[\partial_v\olin-\Omega^2\rlin-2\Omega^2\rho\Big(\frac{\Olino}{\Omega}\Big)\Big]+\frac{\Omega^2}{r}(\ds\olin+\ds\olinb)\\
&\qquad-\frac{3}{2}\rho\Big[\Omega^2(\divs(\elin+\eblin)+2\rlin)+4\Omega^2\rho\Big(\frac{\Olino}{\Omega}\Big)+\frac{M}{r^2}\Big(\otx-\otxb\Big)+\frac{2\Omega^2}{r}\Big(\olin-\olinb\Big)\Big]\\
&=-\ds \Big[\partial_v\olin-2\Omega^2\rlin-2\Omega^2\rho\Big(\frac{\Olino}{\Omega}\Big)\Big]+\frac{\Omega^2}{r}(\ds\olin+\ds\olinb)+\divs\divs(\Omega^2\ablin)-\frac{2}{r}\Big(1-\frac{3M}{r}\Big)\divs\Omega\bblin\\
&\qquad-\frac{3}{2}\rho\Big[\Omega^2(\divs(\elin-\eblin)+2\rlin)+4\Omega^2\rho\Big(\frac{\Olino}{\Omega}\Big)+\frac{M}{r^2}\Big(\otx-\otxb\Big)+\frac{2\Omega^2}{r}\Big(\olin-\olinb\Big)\Big].
\end{align*}
\end{proposition}
\begin{proof}
These follow from elementary computations using the equations of Section \ref{sec:lineqs}.
\end{proof}

\section{Estimates on Any Outgoing Cone} \label{sec:outgoingcone}

We now let $\Si^\prime$ be a partially initial data normalised solution supported on $\ell \geq 2$ which is extendible to null infinity. Given a cone $C_{u_f}$ ($u_0 \leq u_f < \infty$) we will first construct from $\Si^\prime$ a solution $\Sf^\prime = \Si^\prime - \mathscr{G}_f$, where $\mathscr{G}_f$ is a pure gauge solution as in Lemma~\ref{lem:exactsol}, normalised to the cone $C_{u_f}$. We will determine $f$ and hence $\mathscr{G}_f$ in Section \ref{sec:detpg}. We will then define appropriate energies for the solution(s) in Section \ref{sec:energiesout} and be able to state the main theorems in Section \ref{sec:mtout}. The remainder of this section is then concerned with the proof.

\subsection{The Gauge Normalised to the Outgoing Cone} \label{sec:detpg}
Given a $u_f \in (u_0, \infty)$, we fix an outgoing cone $C_{u_f}$. We define\footnote{This definition differs from the one in \cite{ghconslaw} only in the normalisation on the sphere $S^2_{u_f,v_0}$. Note also the typo in (59) of \cite{ghconslaw}, where a factor of $\frac{1}{2}$ is missing.}
\begin{align} \label{choicef}
f \left(v, \theta,\phi \right) = \frac{1}{2} \frac{r}{\Omega^2} \left(u_f, v,\theta,\phi\right) \Big[ \int_{v_0}^{v} \otx_{\Si^\prime} \left(u_f,\bar{v}, \theta, \phi\right) d\bar{v} + f_0(u_f,\theta,\phi) \Big] \, ,
\end{align}
where
\begin{align} \label{f0def}
f_0 (u_f, \theta, \phi) := 2 \left[ 2\Olin_{\Si^\prime} -\frac{r}{2\Omega^2} \otx_{\Si^\prime}\right] \left(u_f,v_0,\theta,\phi\right) \, .
\end{align}
We note $f\left(v_0, \theta,\phi\right)=\frac{r}{2\Omega^2}f_0 (u_f, \theta, \phi)$ and recall from (\ref{f0defi}) that $\Omega^{-2} f_0(u_f, \theta, \phi)$ extends regularly to the horizon as $u_f \rightarrow \infty$ for an initial data normalised solution. Note also that in view of Definition \ref{def:extendskri}, $\frac{f}{r}$ extends regularly to infinity.
As mentioned above, we denote from now on the original solution by $\Si^\prime$, and by $\Sf^\prime$ the solution $\Sf^\prime = \Si^\prime - \mathscr{G}_f$, which is normalised to the cone $C_{u_f}$ in the following sense:

 \begin{proposition} \label{prop:propnewgauge}
The geometric quantities of the solution $\Si^\prime$ have the following properties:
\begin{align} \label{propertiesnewgauge}
\otx_{\Sf^\prime}|_{C_{u_f}} = 0 \ \ \ , \ \ \  \Olin_{\Sf^\prime} |_{C_{u_f}} = 0   \ \ \ , \ \ \ \olin_{\Sf^\prime} |_{C_{u_f}} = 0 \ \  , \ \  (\elin + \eblin)_{\Sf^\prime}|_{C_{u_f}}  = 0.
\end{align}
\end{proposition}
\begin{proof}
In the following we suppress the angular dependence of all quantities in the notation. From Lemma~\ref{lem:exactsol} we have
\begin{align*}
\otx_{\Sf'}(u_f,v)&=\otx_{\Si'}(u_f,v)-\partial_v\Big(\frac{2\Omega^2}{r}f\Big)=-\partial_v\Big(\frac{r(u_f,v)}{\Omega^2(u_f,v)}\frac{\Omega^2}{r}\Big)\frac{2\Omega^2}{r}f=0.
\end{align*}
The Raychauduri equation~\eqref{vray} then gives $\olin_{\Sf'}(u_f,v)=0$ on $C_{u_f}$. Hence, from Lemma~\ref{lem:exactsol}, that
\begin{align*}
\Big(\frac{\Olino}{\Omega}\Big)_{\Sf'}(u_f,v)=\Big(\frac{\Olino}{\Omega}\Big)_{\Sf'}(u_f,v_0)=\Big(\frac{\Olino}{\Omega}\Big)_{\Si'}(u_f,v_0)-\frac{r}{4\Omega^2}\otx_{\Si'}(u_f,v_0)-\frac{1}{4}f_0=0.
\end{align*}
where we use that $\otx_{\mathcal{G}_f}=\otx_{\Si'}$ along $C_{u_f}$ and, $f(v_0)=\frac{r}{2\Omega^2}f_0$.

The last identity in~\eqref{propertiesnewgauge} follows from taking a $\ns$ derivative of $\Olin_{\Sf^\prime} |_{C_{u_f}}=0$ and using equation~\eqref{oml3}.
\end{proof}

For future reference we already note that along the cone $C_{u_f}$
\begin{align} \label{Tfalongufin}
2Tf |_{C_{u_f}}=\partial_vf= \frac{1}{2} \frac{r}{\Omega^2}\otx_{\Si^\prime} |_{C_{u_f}}+\Big[1-\frac{4M}{r}\Big]\frac{f}{r}\Big|_{C_{u_f}}\, ,
\end{align}
and also 
\begin{align} \label{TTfalongufin}
4TTf |_{C_{u_{f}}}&=\partial_v^2f=  - \frac{M}{r \Omega^2}\otx_{\Si'}\Big|_{C_{u_f}}+ 2\olin_{\Si'}\Big|_{C_{u_f}}-\frac{2M}{r^3}f\Big|_{C_{u_f}} \, .
\end{align}
As an immediate consequence, for $ i\geq 2$ and $j \geq 0$ we have the limits
\begin{align} \label{Tflimit}
\lim_{v \rightarrow \infty} | T  [r\slashed{\nabla}]^j f \left(u_f,v, \theta,\phi\right) | =\frac{|[r\slashed{\nabla}]^j f|}{r}  \ \ \ \textrm{and} \ \ \ 
\lim_{v \rightarrow \infty} | T^i  [r\slashed{\nabla}]^j f \left(u_f,v, \theta,\phi\right) | =0 \, .
\end{align}
\subsection{The Energies} \label{sec:energiesout}
\subsubsection{The Master Energies along the Cone $C_{u_f}$}
We define the master energies for the solution $\Sf^\prime$ along the cone $C_{u_f}$. 

For $i=0,1,2$
\begin{align}
{\mathbb{E}}^{i,0} [\Gamma, \Sf^\prime] (u_{f})&=  
\int_{v_0}^{\infty} \int_{S^2}   \Big[  r^2 | [\slashed{\nabla}_T]^{i} \Omega\xlin_{\Sf^\prime}|^2+r^2 \Omega^2 |[\slashed{\nabla}_T]^{i}  \eblin_{\Sf^\prime}|^2 \Big] \Big|_{u=u_{f}}dv\varepsilon_{S^2} \, .
\end{align}
For $i=0,1$
\begin{align}
{\mathbb{E}}^{0,i} [\mathcal{R}, \Sf^\prime] (u_{f})&=  
\int_{v_0}^{\infty}  \int_{S^2} \Big[ r^4|[r\slashed{\nabla}]^{i} \Omega\blin_{\Sf^\prime}|^2+ \Omega^2 r^4 \left(|[r\slashed{\nabla}]^{i}\rlin_{\Sf^\prime}|^2 + |[r\slashed{\nabla}]^{i}\slin_{\Sf^\prime}|^2 \right)  \Big]\Big|_{u=u_{f}}dv\varepsilon_{S^2} \, , \nonumber \\
{\mathbb{E}}^{i,0} [\mathcal{R}, \Sf^\prime] (u_{f})&=  
\int_{v_0}^{\infty} \int_{S^2}  \Big[ r^4|[\slashed{\nabla}_T]^{i} \Omega\blin_{\Sf^\prime}|^2+ \Omega^2 r^4 \left(|[\slashed{\nabla}_T]^{i}\rlin_{\Sf^\prime}|^2 + |[\slashed{\nabla}_T]^{i}\slin_{\Sf^\prime}|^2 \right)  \Big]\Big|_{u=u_{f}}dv\varepsilon_{S^2} \, .
\end{align}

\subsubsection{The Initial Data Energies} \label{sec:dataenergies}
We define for $i=0,1,2$ the energies of the connection coefficients
\begin{align} \label{Ei0datadef}
\mathbb{E}^{i,0}_{data}[\Gamma] (u_{f}) &\doteq   F^{i,0}_{u_0}[\Gamma, \Si^\prime] (v_0,\infty) + F^{i,0}_{v_0}[\Gamma, \Si^\prime] (u_0, u_{f}) \nonumber \\
&\qquad+\frac{1}{2}  \lim_{v\rightarrow \infty} \int_{S^2} r^3 [\slashed{\nabla}_T]^i \otx_{\Si^\prime} \cdot [\slashed{\nabla}_T]^i \otxb_{\Si^\prime}\varepsilon_{S^2} \, ,
\end{align}
as well for $i=0,1$ the energies of the curvature components
\begin{align}
\mathbb{E}^{i,0}_{data} [\mathcal{R}](u_f) &\doteq F^{i,0}_{u_0}[\mathcal{R}, \Si^\prime](v_0, \infty) + F^{i,0}_{v_0}[\mathcal{R}, \Si^\prime](u_0, u_f) \, ,  \label{Ecd} \\
\mathbb{E}^{0,1}_{data} [\mathcal{R}](u_f) &\doteq \sum_{|{\bf k}|=1}F^{0,{\bf k}}_{u_0}[\mathcal{R}, \Si^\prime](v_0, \infty) +\sum_{|{\bf k}|=1} F^{0,{\bf k}}_{v_0}[\mathcal{R}, \Si^\prime](u_0, u_f) \, . 
\end{align}
We complement the above with energies for quantities in the $\Si'$ gauge on the initial data sphere $S^2_{u_f,v_0}$:
\begin{align}
{\mathbb{D}}_{data}^1(u_f,v_0):= \Omega^2(u_f,v_0) &\Big[\|\Omega^{-2} \otx, \Omega^{-2} \otxb , \olin, \olinb, \Omega^{-2} \Delta_{S^2} f_0\|_{u_f,v_0}^2 \Big] \nonumber \\
+ \Omega^2(u_f,v_0) & \Big[ \| [r\slashed{\nabla} {\Omega^{-2}} \otx, r\slashed{\nabla} {\Omega^{-2}} \otxb, \divs  \elin +\rlin , \divs  \eblin +\rlin,[r\slashed{\nabla}] \olinb,{\Omega^{-1}} \underline{\beta} \|^2_{u_f,v_0}  \Big] \, \nonumber
\end{align}
and 
\begin{align}
\mathbb{D}_{data}^2(u_f,v_0) \doteq  \mathbb{D}_{data}^1(u_f,v_0)+ \Omega^2\Big\|{\Omega^2\eblin,\partial_v\olin,\Omega^{-1}\partial_u\olinb, \slin,\rlin,\frac{\bblin}{\Omega},\Omega^2\ablin}\Big\|^2_{u_f,v_0} &\nonumber\\
 +\Omega^2\Big\| T(\divs\eblin+\rlin),T(\divs\elin+\rlin), \Delta_{S^2} \Omega^{-2} \olinb - r^2 \divs  \Omega^{-1} \bblin\Big\|^2_{u_f,v_0} &
\nonumber \\
+\Omega^2 \Bigg[\Big|\Big\langle\ns_T\frac{\bblin}{\Omega},\ns\olin\Big\rangle_{u_f,v_0}\Big|+\Big|\Big\langle\ns\rlin,\ns\frac{\otx}{\Omega^2}\Big\rangle_{u_f,v_0}\Big|+|\langle\ns\rlin,\ns\olin\rangle_{u_f,v_0}| +\Big|\Big\langle\ns_T\frac{\bblin}{\Omega},\frac{\ns\otx}{\Omega^2}\Big\rangle_{u_f,v_0}\Big|\Bigg] &\nonumber \, , 
\end{align}
where as mentioned all quantities that appear above are in the $\Si'$ gauge and we denote
\begin{align*}
|||f,g,h,...||_{u,v}=||f||_{u,v}+||g||_{u,v}+||h||_{u,v}+\dots.
\end{align*}

\begin{remark}
The data energies $\mathbb{E}^{i,0}_{data}[\Gamma] (u_{f})$, $\mathbb{E}^{i,0}_{data} [\mathcal{R}](u_f)$, $\mathbb{E}^{0,1}_{data} [\mathcal{R}](u_f)$ defined above are not necessarily coercive. What we will show successively in the proof is that these energies are indeed coercive provided a large constant (depending only on $M$) times an appropriate $\mathbb{D}^i_{data}$-energy is added. This is familiar from \cite{ghconslaw}. Alternatively, the aforementioned data energies can be made manifestly coercive by suitably normalising the initial data gauge. 
\end{remark}

\begin{remark}
Note that we have $\mathbb{D}^1_{data}(\infty,v_0)=0$ and $\mathbb{D}^2_{data}(\infty,v_0)=0$ for regular data. 
\end{remark}

\begin{remark}
The energy $\mathbb{D}^1_{data}$ contains up to first angular derivatives of connection coefficients and $\Delta_{S^2} f_0$, which involves second angular derivatives. The reason we have included the latter in $\mathbb{D}^1_{data}$ is that $\Delta_{S^2} f_0$ will play a distinguished role in the proof and could be brought to vanish by normalising the initial data gauge slightly differently, namely with respect to sphere $S^2_{u_f,v_0}$ instead of $S^2_{\infty,v_0}$.
\end{remark}

\begin{remark}
To appreciate the structure of $\mathbb{D}_{data}^2(u_f,v_0)$ note that the terms in the first line involve (besides $\Delta_{S^2} f_0$ in  $\mathbb{D}_{data}^1(u_f,v_0)$ as mentioned) only first derivatives of Ricci coefficients and curvature components. The terms in the second line involve first derivatives of the mass aspect functions and a renormalised second derivative of $\olinb$. Importantly, control on these quantities on the double null spheres can be propagated in evolution at the basic regularity level considered here, namely fluxes of up to one derivative of curvature in $L^2$. This is well-known from \cite{CK93} in the non-linear context. Similarly, the terms in the third line involve (norms of) inner products on double null spheres which can be propagated in evolution along the cones. Note that one derivative of curvature on spheres cannot in general be propagated without loss of regularity which is why we have retained the inner products.  
\end{remark}

\subsection{The Main Theorems} \label{sec:mtout}

We are ready to state the main theorems. In the following we write $X \lesssim Y$ if $X \leq C \cdot Y$ holds for a constant $C$ depending only on $M$.

\begin{theorem} \label{prop:firstround}
Let $\Si^\prime$ be a partially initial data normalised solution supported on $\ell \geq 2$ of the system of gravitational perturbations that is also extendible to null infinity. Let $\Sf^\prime$ be as defined in Section \ref{sec:detpg} and recall the energies defined in Section \ref{sec:energiesout}. We have the estimate
\begin{align} \label{secondorderest}
\sum_{i=0}^2 {\mathbb{E}}^{i,0} [\Gamma, \Sf^\prime] (u_{f})  +\sum_{i=0}^1 {\mathbb{E}}^{0,i} [\mathcal{R}, \Sf^\prime] (u_{f}) +\sum_{i=0}^1 {\mathbb{E}}^{i,0} [\mathcal{R}, \Sf^\prime] (u_{f})  
\lesssim \mathbb{E}^2_{data}(u_f) 
\end{align}
where 
\begin{align} \label{energyinmaintheorem}
\mathbb{E}^2_{data}(u_f) :=& \mathbb{E}^{2,0}_{data}[\Gamma] (u_{f}) +  \mathbb{E}^{1,0}_{data}[\mathcal{R}] (u_{f}) + \mathbb{E}^{1,0}_{data}[\Gamma] (u_{f}) \nonumber \\
+&  \mathbb{E}^{0,1}_{data}[\mathcal{R}] (u_{f}) + \mathbb{E}^{0,0}_{data}[\mathcal{R}] (u_{f}) +  \mathbb{E}^{0,0}_{data}[\Gamma] (u_{f})  + C_{max} \cdot \mathbb{D}_{data}^2(u_{f},v_0) 
\end{align}
for a constant $C_{max}$ depending only on $M$ determined in the course of the proof.
\end{theorem}
Note that using Definition \ref{def:extendskri} and (\ref{hozextendQ}) near the horizon one easily verifies $\sup_{u_f \in [u_0,\infty)} E^2_{data}(u_f)<\infty$ for the class of initial data considered, as was claimed in Theorem \ref{theo:gaugeinvariantoutgoingintro} of the introduction already. 

The proof of Theorem \ref{prop:firstround} will be carried out in Section \ref{sec:2oe}. Before that we provide some preliminary estimates in Section \ref{sec:prelimest} that will be used in the proof. 

From Theorem \ref{prop:firstround} we will easily deduce the following theorem, which concerns the gauge invariant quantities. It can be stated (but not proven) without reference to the solution $\Sf^\prime$:

\begin{theorem} \label{theo:gaugeinvariantoutgoing}
Let $\Si^\prime$ be a partially initial data normalised solution supported on $\ell \geq 2$ of the system of gravitational perturbations that is also extendible to null infinity. We have the following estimates for the fluxes of the gauge invariant Teukolsky quantities  $A := r\Omega^2 \alin_{\Si^\prime} = r \Omega^2 \alin_{\Sf^\prime}$ and $\underline{A} := r\Omega^2 \ablin_{\Si^\prime} = r \Omega^2 \ablin_{\Sf^\prime}$: Along any outgoing cone $C_{u_f}$ with $u_f \in [u_0,\infty)$, 
\begin{align} \label{finalalpha}
 \int_{v_0}^{\infty} \int_{S^2} dv  d\theta d\phi \sin \theta  \left[ r^4|\Omega^{-1} \slashed{\nabla}_3 A|^2 + | \Omega \slashed{\nabla}_4 A|^2+ |r\slashed{\nabla} A|^2 + |A|^2 \right] (u_f)  \lesssim  \mathbb{E}^2_{data} (u_f)  \, ,
\end{align}
\begin{align} \label{finalalphab}
\int_{v_0}^{\infty}\int_{S^2} dv  d\theta d\phi \sin \theta \Omega^2 \left[  r^2 | \Omega \slashed{\nabla}_4\underline{A}|^2+ | \slashed{\nabla}\underline{A}|^2 + \frac{1}{r^2} |\underline{A}|^2 \right] (u_f) \lesssim  \mathbb{E}^2_{data} (u_f) \, .
\end{align}
Here $ \mathbb{E}^2_{data} (u_f) $ is defined in (\ref{energyinmaintheorem}).
\end{theorem}

\begin{remark} \label{rem:degw}
Note the strong degeneration in $\Omega^2$ as actually $\ablin \Omega^{-2}$ has a finite limit on $\mathcal{H}^+$. Note also that the $r$-weights are not symmetric between the two estimates. 
\end{remark}

Theorem \ref{theo:gaugeinvariantoutgoing} will be proven in Section \ref{sec:gaugeinvariantoutgoing}. Afterwards, we will use the estimates of Theorem  \ref{theo:gaugeinvariantoutgoing} to obtain analogous estimates for the ingoing fluxes. See Corollary \ref{cor:main}. For $\underline{A}$ this will resolve the problem of degenerating $\Omega^2$-weights mentioned in Remark \ref{rem:degw}.

\subsection{Preliminaries: Useful Identities and Estimates Along $C_{u_f}$} \label{sec:prelimest}
Before we begin with the proof of Theorem \ref{prop:firstround}, we prove in this section some preliminary estimates (both on quantities of the solution $\Si^\prime$ and of the solution $\Sf^\prime$) along the cone $C_{u_f}$ that will be used in the sequel. We begin by exploiting the fact that the mass aspect is conserved along $C_{u_f}$ for the solution $\Sf^\prime$.
\subsubsection{Bounding the Mass Aspect along $C_{u_f}$}
\begin{lemma} \label{lem:fromoldpaper}
We have the following bounds for the lapse of the solution $\Sf^\prime$ along $C_{u_f}$:
\begin{align}
\sup_{v \in [v_0,\infty)} \| r^3\divs  \elin_{\Sf^\prime} + r^3\rlin_{\Sf^\prime} \|_{u_f,v} \lesssim    
{\mathbb{D}}^1_{data}(u_f,v_0) \, .
\end{align}
\end{lemma}

\begin{proof}
We compute the propagation equation (valid for any solution)
\begin{align}
\Omega \slashed{\nabla}_4 \left(r^3 \divs  \elin + r^3 \rlin \right) = \Omega^2 \Delta_{S^2} 2\Olin + 3M \otx \, .
\end{align}
For the solution $\Sf^\prime$, the right hand side vanishes  along $C_{u_f}$ by Proposition \ref{prop:propnewgauge} 
\begin{align*}
(r^3 \divs  \elin + r^3 \rlin)_{\Sf'}(u_f,v) =(r^3 \divs  \elin + r^3 \rlin)_{\Sf'}(u_f,v_0)=(r^3 \divs  \elin + r^3 \rlin)_{\Si'}(u_f,v_0)-\frac{r}{2} \Delta_{S^2} f_0 -3M f_0,
\end{align*}
where in the last equality we use Lemma~\ref{lem:exactsol} and $f(u_f,v_0)=\frac{r}{2\Omega^2}f_0$. Therefore, one has
\begin{align} \label{divetaplusrho}
||(r^3 \divs  \elin + r^3 \rlin)_{\Sf'}||_{u_f,v}&\leq ||(r^3 \divs  \elin + r^3 \rlin)_{\Si'}||_{u_f,v_0}+\frac{r(u_f,v_0)}{2}\| \Delta_{S^2} f_0\|_{u_f,v_0}+3M\|f_0\|_{u_f,v_0}.
\end{align}
Since our solutions are supported on $\ell\geq 2$, the last term can be estimated by the prenultimate one by standard elliptic estimates on $S^2$ and the result hence follows from the definition of $\mathbb{D}^1_{data}(u_f,v_0)$.
\end{proof}

\subsubsection{Bounds on the Sphere $S^2_{u_f,\infty}$}

\begin{proposition} \label{prop:spherebounds}
We have for $i=0,1,2$ the following estimate on the limiting sphere $S^2_{u_f,\infty}$:
\begin{align}
\Big\| [r\slashed{\nabla}]^i \left(\Olin \right)_{\Si^\prime}\Big\|^2_{u_f,\infty} &=\frac{1}{4} \Big\| [r\slashed{\nabla}]^i \left( \frac{f}{r}\right)\Big\| ^2_{u_f,\infty} \, .  \label{omegadata} \end{align}
\end{proposition}
\begin{proof}
We first observe that with our choice of $f$ from (\ref{choicef}) we have
\begin{align}
r\partial_v \left( \frac{f\Omega^2}{r} \right) \left(u_f, v,\theta,\phi\right) ={\frac{1}{2}} r \otx_{\Si^\prime} \left(u_f,v, \theta, \phi\right) \, .
\end{align}
The right hand side goes to zero as $v \rightarrow \infty$ by the asymptotics of the solution in the initial data gauge (see (\ref{decreten2})) and hence so does the left hand side. It follows that asymptotically we have the identity
\begin{align} \label{hur}
\partial_v (f\Omega^2) \left(u_f, \infty,\theta,\phi\right)= \frac{\Omega^2f}{r} \left(u_f, \infty,\theta,\phi\right) \, .
\end{align}
Moreover, this identity continues to hold if we commute with derivatives tangential to the spheres. Now by definition of the pure gauge solution we have from Lemma \ref{lem:exactsol} the relation
\begin{align}
\Big\| [r\slashed{\nabla}]^i \left(\Olin \right)_{\Si^\prime}\Big\|^2_{u_f,\infty} =  \Big\| [r\slashed{\nabla}]^i \frac{1}{2\Omega^2} \partial_v (f \Omega^2) \Big\|^2_{u_f,\infty} = \frac{1}{4} \Big\| [r\slashed{\nabla}]^i  \frac{f}{r} \Big\|^2_{u_f,\infty} \, ,
\end{align}
where we have used that $\Olin_{\Sf^\prime}=0$ along $C_{u_f}$. 
\end{proof}

\begin{proposition} \label{prop:spherebounds2}
We have for $i=0,1,2$ the following estimates on the limiting sphere $S^2_{u_f,\infty}$:
\begin{align}
\| [r\slashed{\nabla}]^i r \otxb_{\Si^\prime}\|^2_{u_f,\infty} &\lesssim 
 \Big\| [r\slashed{\nabla}]^i \left( \frac{f}{r}\right)\Big\| ^2_{u_f,\infty}  \, ,  \label{tchibdata}  \\
\| r \otxb_{\Sf^\prime}\|^2_{u_f,\infty} &\lesssim 
 \sum_{i=0}^2 \Big\| \left[r\slashed{\nabla}\right]^i \left(\frac{f}{r} \right)\Big\|^2_{u_f,\infty}  \, .\label{tchibfut} 
\end{align}
\end{proposition}

\begin{proof}
For the first estimate, we consider the linearised Gauss equation (\ref{lingauss}) multiply by $r^2$ and take the limit $v \rightarrow \infty$ along $C_{u_f}$ to produce (using the asymptotics of the solution $\Si^\prime$ following from Definition \ref{def:extendskri}) the asymptotic identity
\begin{align} \label{asymO}
r\otxb_{\Si^\prime}  \left(u_f,\infty\right)  = -4 \Big( \frac{\Olino}{\Omega} \Big)_{\Si^\prime} \left(u_f,\infty\right).
\end{align}
Combining this with (\ref{omegadata}) of Proposition~\ref{prop:spherebounds} the estimate (\ref{tchibdata}) follows. 

The second estimate follows from the fact that by definition of the gauge function in Lemma~\ref{lem:exactsol} we have the relation
\begin{align}
\Big\|r \otxb_{\Sf^\prime}\Big\|^2_{u_f,\infty} \lesssim \Big\| r \otxb_{\Si^\prime}\Big\|^2_{u_f,\infty} + \sum_{i=0}^2 \Big\| \left[r\slashed{\nabla}\right]^i \left(\frac{f}{r} \right)\Big\|^2_{u_f,\infty} \, .
\end{align}
Using now (\ref{tchibdata}) for the first term on the right  we obtain the second estimate.
\end{proof}

\subsubsection{Estimating the Gauge Terms at Infinity}
In this section we estimate the terms $\mathcal{G}^{i, {\bf k}}_\Gamma$ and $\mathcal{G}^{i, {\bf k}}_\mathcal{R}$
appearing in Proposition~\ref{prop:gaugechange} on the (asymptotic) sphere $S^2_{u_f,\infty}$ as this is what will enter in the conservation laws applied later.

\begin{proposition} \label{prop:Gatinfinity}
On the sphere at infinity, $S^2_{u_f,\infty}$, we have for $i \geq 1$
\begin{align} \label{GammaGonetime}
\int_{S^2} \mathcal{G}_\Gamma^{i,{\bf k}}\left(u_f,\infty\right)   \varepsilon_{S^2} =  -  \int_{S^2}\frac{ r^3}{2\Omega^2} \mathcal{K}^{i,{\bf k}}\otx_{\Si^\prime}\cdot \mathcal{K}^{i,{\bf k}}\otxb_{\Si^\prime} \left(u_f,\infty, \theta,\phi\right)\varepsilon_{S^2} \, .
\end{align}
We also have for a constant ${C}$ depending only on $M$,
\begin{align}\label{GammaGnone}
\int_{S^2} \mathcal{G}_\Gamma^{0,0}\left(u_f,\infty\right) \varepsilon_{S^2}  \geq  & \  4M \Big\| \frac{f\Omega^2}{r} \Big\|_{u_f,\infty}^2 
- \int_{S^2_{u_f,\infty}} \frac{r^3}{2\Omega^2} \otx_{\Si^\prime} \otxb_{\Si^\prime} \varepsilon_{S^2} - {C} \mathbb{D}^1_{data}(u_f,v_0) \, . 
\end{align}
For the $\mathcal{G}^{i,{\bf k}}_{\mathcal{R}}$-terms we have the estimate
\begin{align} \label{RGtime}
\Big| \int_{S^2}   \mathcal{G}_{\mathcal{R}}^{1,0}   \left(u_f,\infty, \theta,\phi\right)\varepsilon_{S^2}\Big| \leq 3M\Big( \| \olinb_{\Si^\prime}\|_{u_f,\infty}^2 + \| \olinb_{\Sf^\prime}\|_{u_f,\infty}^2\Big) \leq 6M \| \olinb_{\Sf^\prime}\|_{u_f,\infty}^2 \, , 
\end{align}
as well as the identities
\begin{align} \label{RG}
\int_{S^2} \mathcal{G}_{\mathcal{R}}^{0,0}   \left(u_f,\infty, \theta,\phi\right)\varepsilon_{S^2} = \frac{3}{2} M \| [r\slashed{\nabla}] \frac{f \Omega^2}{r}\|_{u_f,\infty} -3M  \|  \frac{f \Omega^2}{r}\|_{u_f,\infty} \, ,
\end{align}
and
\begin{align} \label{RGangular}
\sum_{|{\bf k}|=1}\int_{S^2} \mathcal{G}_{\mathcal{R}}^{0,{\bf k}}   \left(u_f,\infty, \theta,\phi\right)\varepsilon_{S^2} = \frac{3}{2} M \| [r\slashed{\nabla}]^2 \frac{f \Omega^2}{r}\|_{u_f,\infty} - 3M  \| [r\slashed{\nabla}] \frac{f \Omega^2}{r}\|_{u_f,\infty} \, .
\end{align}
\end{proposition}

\begin{proof}
Recall throughout that we are free to use the properties (\ref{propertiesnewgauge}) when evaluating the expressions (\ref{finge}) and (\ref{finge2}) in the proof. 

The estimate (\ref{GammaGonetime}) follows from (\ref{Tflimit}) and the decay towards null infinity of the linearised quantities.

For the second estimate in~\eqref{GammaGnone}, we have from (\ref{finge2}) that
\begin{align*}
\mathcal{G}_\Gamma^{0,{\bf 0}}  \left(u,v,\theta,\phi\right)  &=  6M\left(\frac{\Omega^2 f}{r}\right)^2 
 - \frac{r^3}{2\Omega^2} \otx_{\Si'}\otxb_{\Si'}+ \left(\frac{2f \Omega^2}{r}\right)r^3  ({\rlin}+\divs \, {\elin})_{\Sf'} \, , 
\end{align*}
where one uses that $\eblin_{\Sf^\prime} = - \elin_{\Sf^\prime}$ by Proposition \ref{prop:propnewgauge}. We then have
\begin{align*}
\left(\frac{2f \Omega^2}{r}\right) r^3  ({\rlin}+\divs \, {\elin})_{\Sf'} \geq -\Big|\frac{2\sqrt{M}f \Omega^2}{r}\Big|\Big| \frac{r^3}{\sqrt{M}}  ({\rlin}+\divs \, {\elin})_{\Sf'} \Big|\geq -2M\Big|\frac{f \Omega^2}{r}\Big|^2-\frac{1}{2M}  \Big| {r^3}({\rlin}+\divs \, {\elin})_{\Sf'} \Big|^2,
\end{align*}
where we have applied Young's inequality. So,  
\begin{align*}
\mathcal{G}_\Gamma^{0,{\bf 0}}  \left(u,v,\theta,\phi\right)  &\geq   4M\left(\frac{\Omega^2 f}{r}\right)^2 
 - \frac{r^3}{2\Omega^2} \otx_{\Si'}\otxb_{\Si'}-\frac{1}{2M}  \Big| r^3({\rlin}+\divs \, {\elin})_{\Sf'} \Big|^2.
 \end{align*}
 The conclusion then follows from (\ref{divetaplusrho}).

To prove the estimate (\ref{RGtime}) we note that, using equation~(\ref{Tflimit}) and the decay towards null infinity of the linearised quantities, all terms except the ones in the first line of~\eqref{finge2} vanish in the $v\rightarrow \infty$ limit. So,
\begin{align*}
\mathcal{G}_{\mathcal{R}}^{1,{\bf 0}}  \left(u_f,\infty,\theta,\phi\right)  = 3M \Big(\olinb_{\Sf'}^2-\olinb_{\Si'}^2\Big) \, , \nonumber 
\end{align*}
where we use the asymptotic relations $2T \Olin = \olinb$ and $r2T \otxb = -4 \olinb$, valid for both $\Si^\prime$ and $\Sf^\prime$, on the asymptotic sphere $S^2_{u_f, \infty}$. The last inequality of (\ref{RGtime}) follows from $\olinb_{\Sf^\prime} = \olinb_{\Si^\prime}$ as $v\rightarrow \infty$ which is a consequence of the general relation $\olinb_{\Sf^\prime} = \olinb_{\Si^\prime}- \frac{2M}{r^3}\Omega^2f$ following in turn from Lemma~\ref{lem:exactsol} using that $f$ is independent of $u$.

We now directly prove (\ref{RGangular}) as the proof of (\ref{RG}) is analogous but without commutation. We consider the expression (\ref{finge2}) for ${\bf k}$ with $|{\bf k}|=1$ on the limiting sphere $S^2_{u_f, \infty}$. Most terms vanish using the asymptotics of the solution $\Si^\prime$ from (\ref{decreten}) and the properties of the solution $\Sf^\prime$ along $C_{u_f}$ (in particular that $\elin_{\Sf^\prime}=- \eblin_{\Sf^\prime}$ and that $r^2 \eblin_{\Sf^\prime}$ still extends regularly to infinity). We are left with
\begin{equation} \label{samut}
\lim_{v \rightarrow \infty} \sum_{|\bf k|=1}\mathcal{G}_{\mathcal{R}}^{0,{\bf k}}  \left(v, u_f,\theta,\phi\right)  = \lim_{v \rightarrow \infty} \Bigg[ 3M  \Big\langle[r\ns]\Big(\frac{\Olino}{\Omega}\Big)_{\Si'} , [r\ns]r \otxb_{\Si'}\Big\rangle + \frac{3M}{2}\sum_{i=1}^2 \Big|[r\ns]^i \frac{f \Omega^2}{r} \Big|^2 \Bigg] \, ,
\end{equation}
where we have used the angular identities of Proposition~\ref{prop:angident}. For the first term on the right of (\ref{samut}) we recall that as $v\rightarrow \infty $ along the cone $C_{u_f}$ we have the asymptotic identity (\ref{asymO}), so
\begin{align}
\lim_{v\rightarrow \infty}\int_{S^2}\Big\langle[r\ns]\left(\Olin\right)_{\Si'} , [r\ns]r\otxb_{\Si'}\Big\rangle \varepsilon_{S^2}=-4\lim_{v\rightarrow \infty} \int_{S^2}   \Big|[r\ns] \Olin_{\Si^\prime}\Big|^2 \varepsilon_{S^2}.
\end{align}
The result now follows from inserting (\ref{omegadata}).
\end{proof}

\subsubsection{Estimating the Gauge Terms at the Data Sphere in the Conservation Laws}
In this section we estimate the terms $\mathcal{G}^{i, {\bf k}}_\Gamma$ and $\mathcal{G}^{i, {\bf k}}_\mathcal{R}$
appearing in Proposition~\ref{prop:gaugechange} on the (data) sphere $S^2_{u_f,v_0}$ as this is what will enter in the conservation laws applied later.
\begin{proposition}
On the data sphere $S^2_{u_f,v_0}$ we have 
\begin{align}\label{eqn:dataG00}
 \Big| \int_{S^2} \mathcal{G}_\Gamma^{0,0} \left(u_f,v_0, \theta,\phi\right) \varepsilon_{S^2} 
 \Big| \lesssim \mathbb{D}^1_{data}(u_f,v_0) \, , 
\end{align}
\begin{align}\label{eqn:dataG10}
\Big | \int_{S^2}  \mathcal{G}_\Gamma^{1,0} \left(u_f,v_0, \theta,\phi\right)\varepsilon_{S^2} \Big| \lesssim {\mathbb{D}}_{data}^1(u_f,v_0) \,  , 
\end{align}
\begin{align}\label{eqn:dataG20}
\Big | \int_{S^2}  \mathcal{G}_\Gamma^{2,0} \left(u_f,v_0, \theta,\phi\right)\varepsilon_{S^2} \Big| \lesssim {\mathbb{D}}_{data}^2(u_f,v_0) \, .
\end{align}
Furthermore,
\begin{align}\label{eqn:dataGR00}
 \Big| \int_{S^2} \mathcal{G}_{\mathcal{R}}^{0,0} \left(u_f,v_0, \theta,\phi\right) \varepsilon_{S^2} \Big|\lesssim {\mathbb{D}}_{data}^1(u_f,v_0)  \, ,
\end{align}
\begin{align}\label{eqn:dataGR01}
\Big| \int_{S^2} \mathcal{G}_{\mathcal{R}}^{0,1} \left(u_f,v_0, \theta,\phi\right) \varepsilon_{S^2} \Big| \lesssim {\mathbb{D}}_{data}^1(u_f,v_0) \, , 
\end{align}
\begin{align}\label{eqn:dataGR10}
\Big | \int_{S^2}  \mathcal{G}_{\mathcal{R}}^{1,0} \left(u_f,v_0, \theta,\phi\right)\varepsilon_{S^2} \Big| \lesssim {\mathbb{D}}_{data}^1(u_f,v_0) \, .
\end{align}

\end{proposition}
\begin{remark}
Note that we gain regularity here; naively one would expect to see third derivatives of Ricci coefficients in $G_{\Gamma}^{2,0}$ but $TTf$ (see~\eqref{TTfalongufin}) allows us to gain regularity via integration by parts.
\end{remark}
\begin{proof}
For the estimate~\eqref{eqn:dataG00}, we have from Proposition~\ref{prop:gaugechange} and the properties of $\Sf^\prime$ in (\ref{propertiesnewgauge}):
\begin{align} \label{oals}
\mathcal{G}_\Gamma^{0,{\bf 0}}  \left(u_f,v_0,\theta,\phi\right)  &=  \frac{3M}{2}f_0^2 
 - \frac{r^3}{2\Omega^2} \otx_{\Si'}\otxb_{\Si'}+ f_0r^3  ({\rlin}+\divs \, {\elin})_{\Sf'} \, .
\end{align}
We can use Young's inequality and Lemma~\ref{lem:fromoldpaper} to conclude the result. 

Turning to~\eqref{eqn:dataGR00} and~\eqref{eqn:dataGR01} we note  from Proposition~\ref{prop:gaugechange} and the properties of $\Sf^\prime$ in (\ref{propertiesnewgauge}) for $i=0,1$
\begin{align*} 
\mathcal{G}_{\mathcal{R}}^{0,i}  \left(u_f,v_0,\theta,\phi\right)&  = 3M r \big\langle[r\ns]^i\big(\Olin\big)_{\Si'} , [r\ns]^i\otxb_{\Si'}\big \rangle-\frac{3Mr^2}{2}  \langle[r\ns]^i f_0, [r\ns]^i(\divs\elin+\rlin)_{\Si'}\rangle \nonumber \\
& \ \  -\frac{3M}{4} \langle[r\ns]^i f_0, [r\ns]^i\Delta_{{S}^2} f_0\rangle+ \sum_{k=0}^i\frac{3M}{2} \Big| [r\ns]^{1+k}\frac{f_0}{2} \Big|^2+ \frac{3M^2}{r}  \Big|[r\ns]^if_0\Big|^2 \, ,
\end{align*}
where we have used (\ref{oals}) and the angular identity for scalars of Proposition~\ref{prop:angident}.
Up to a total divergence on $S^2$ this reduces to
\begin{align*} 
\mathcal{G}_{\mathcal{R}}^{0,i}  \left(u_f,v_0,\theta,\phi\right)
&=\frac{3M r}{4}(-\Delta_{S^2})^if_0\otxb_{\Si'} -\frac{3M}{4} (-\Delta_{S^2})^if_0\Delta_{S^2}f_0-\frac{3M r^2}{2} (-\Delta_{S^2})^if_0(\divs\elin+\rlin)_{\Si'} \\
&\quad +\frac{3Mr^2}{4\Omega^2}\langle[r\ns]^i\otx_{\Si'},[r\ns]^i\otxb_{\Si'}\rangle+  \sum_{k=0}^i\frac{3M}{2} \Big| [r\ns]^{1+k}\frac{f_0}{2} \Big|^2+ \frac{3M^2}{r}  \Big|[r\ns]^if_0\Big|^2,
\end{align*}
where we have inserted from (\ref{f0def}) the definition of $f_0$ for the terms involving $\Olin$ and $\otx$. The result is now immediate from Young's inequality and that $\Delta_{S^2} f_0$ controls all second angular derivatives of $f_0$ by elliptic estimates (recall that $f_0$ is supported on $\ell \geq 2$).

We now move onto the $T$-commuted estimates. Since we are working on the sphere $S^2_{u_f,v_0}$, we will adopt the convention that $\sim$ means equality up to $r$-weights and numerical factors, i.e.~we only keep track of $\Omega^2$-weights. We recall from (\ref{Tfalongufin}) the identity
\begin{align*}
\Omega^2Tf (u_f,v_0)&\sim \Omega^2\big[\Omega^{-2}\otx_{\Si^\prime}(u_f,v_0)+\Omega^{-2}f_0\big].
\end{align*}
The reduction of the equations in Proposition~\ref{prop:TTransport} in the $\Sf'$-gauge to $\Si'$-gauge on~$C_{u_f}$ is
\begin{align*}
2T\otx_{\Sf'}&\sim \Omega^2\Big[(\divs\elin+\rlin)_{\Si'}+\otxb_{\Si'}+f_0\Big],\\
2T\otxb_{\Si'}&\sim\Omega^2\Big[(\divs\eblin+\rho)+\Olin+\Omega^{-2}\otxb+\otx+\olinb\Big] \, , \\ 
2T(\rlin-\divs\eblin)_{\Sf'}&\sim\Omega^2\Big[\divs\big(\Omega^{-1}\bblin\big)+\Omega^{-2}\otxb+\Delta_{S^2} (\Omega^{-2}f_0)+\Omega^{-2}f_0\Big]
\end{align*}
and
\begin{align*}
2T\otx_{\mathscr{G}_f}&=2T\otx_{\Si'}-2T{\otx}_{\Sf'}\sim \Omega^2\Big[(\Olin)_{\Si'}+\Omega^{-2} \otx_{\Si'}+\olin_{\Si'}+f_0\Big],
\end{align*}
where we use Lemma~\ref{lem:exactsol} and $f(v_0)=\Omega^{-2}rf_0$. 

For the estimate~\eqref{eqn:dataG10}, we use Young's inequality to estimate the terms in $\mathcal{G}^{1,\mathbf{0}}_{\Gamma}(u_f,v_0)$ (everything is evaluated at $(u_f,v_0)$ and quantities appearing on the right-hand side are in the $\Si'$ gauge):
{\allowdisplaybreaks
\begin{align*}
\Big|T\frac{\Omega^2 f}{r}\Big|^2&\lesssim\Omega^4\Big[\big| \Omega^{-2}\otx\big|^2+\big|\Omega^{-2}f_0\big|^2\Big] ,\\
\Big|\frac{r^3}{\Omega^2}T \left(\frac{f \Omega^2}{r}\right)   T {\otx}_{\Sf'}
\Big|&\lesssim\Omega^2\Big[\big|\Omega^{-2}\otxb\big|^2+\big|\Omega^{-2}\otx\big|^2+\big|\Omega^{-2}f_0\big|^2+|(\divs\elin+\rlin)|^2\Big],\\
 \Big|\frac{r^3}{\Omega^2}  T\otx_{\mathscr{G}_f}T\otxb_{\Si'}
\Big|&\lesssim\Omega^2\Big[\Big|(\divs\eblin+\rlin)\Big|^2+\Big|\frac{\otx}{\Omega^2}\Big|^2+\Big|\frac{\otxb}{\Omega^2}\Big|^2+|\olin|^2+|\olinb|^2+\Big|\frac{f_0}{\Omega^2}\Big|^2\Big] ,
\end{align*}}
and
\begin{align*}
\Big|\int_{S^2}T \left(\frac{f \Omega^2}{r}\right)T\big( {\rlin}-\divs {\eblin} \big)_{\Sf'}\varepsilon_{S^2}\Big|&\lesssim\Omega^4\Big[\sum_{i=0}^1\Big|[r\ns]^i\frac{\otx}{\Omega^2}\Big|^2+\sum_{i=0}^1\Big|[r\ns]^i\frac{f_0}{\Omega^2}\Big|^2+\Big|\frac{\bblin}{\Omega}\Big|^2+\Big|\frac{\otxb}{\Omega^2}\Big|^2\Big]\, ,
\end{align*}
where we have integrated by parts to achieve the last estimate. Therefore, all terms arising in $G^{1,{\bf{0}}}_{\Gamma}(u_f,v_0)$ can be accounted for in $\mathbb{D}^1_{data}(u_f,v_0)$ establishing (\ref{eqn:dataG10}). 

Before proving (\ref{eqn:dataG20}), we prove the estimate~\eqref{eqn:dataGR10}. In addition to the identities collected above for $\mathcal{G}_{\Gamma}^{1,\mathbf{0}}$ we require
\begin{align*}
T(\divs\elin+\rlin)_{\Sf'}&\sim\Omega^2\Big[\divs\big(\Omega^{-1}\bblin\big)_{\Si'}+\Omega^{-2}\otxb_{\Si'}+\Delta_{S^2}\Omega^{-2}f_0+\Omega^{-2}f_0\Big]+\Delta_{S^2}\olinb_{\Si'},\\
2T\otxb_{\Sf'}&\sim \Omega^2\Big[(\rlin+\divs\eblin)_{\Si'}+\Omega^{-2}\otxb_{\Si'}+\olinb_{\Si'}+\Omega^{-2}f_0 +\Delta_{S^2}(\Omega^{-2}f_0)+\Delta_{S^2}\Omega^{-2}\otx_{\Si'}\Big],
\end{align*}
where we use that $\eblin_{\Sf'}=-\elin_{\Sf'}$ and the definition of $f_0$ to write
\begin{align*}
(\rlin-\divs\elin)_{\Si'}=(\rlin+\divs\eblin)_{\Si'}-2\ds\Olin_{\Si'}=(\rlin+\divs\eblin)_{\Si'}-\frac{1}{2}\ds \Big[f_0+\frac{r}{\Omega^2}\otx_{\Si'}\Big]
\end{align*}
to deal with the $(\rho+\divs\eblin)_{\Sf'}$ term arising in $T{\otxb}_{\Sf}$ from Proposition~\ref{prop:TTransport}.

We can now estimate all terms that arise in $\mathcal{G}_{\mathcal{R}}^{1,{\bf 0}}  (u_f,v_0)$ via Young's inequality (note that $\olin_{\Sf'}=0$ on $C_{u_f}$):
\begin{align*}
\Big|(\olin+\olinb)_{\Si} T\otxb_{\Si}\Big|&\lesssim\Omega^2\Big[|\olin|^2+|\olinb|^2+|\divs\eblin+\rho|^2+|\Olin|^2+\big|\Omega^{-2}\otxb\big|^2+|\otx|^2\Big] \, , \\
\Big|\frac{1}{\Omega^2} T\Big(\frac{f \Omega^2}{r}\Big)T{\otx}_{\Sf}\Big|&\lesssim \Omega^2\Big[\big|\Omega^{-2}\otx\big|^2+\big|\Omega^{-2}f_0\big|^2+|\divs\elin+\rlin|^2+|\otxb|^2\Big],\\
\Big| T\Big(\frac{f\Omega^2}{r}\Big)\olinb_{\Sf}\Big|  &\lesssim \Omega^2\Big[\big|\Omega^{-2}\otx\big|^2+\big|\Omega^{-2}f_0\big|^2+|\olinb|^2\Big], \\
 \Big| [r\ns]^i T\frac{f \Omega^2}{r} \Big|^2&\lesssim \Omega^4\Big|[r\ns]^i\Omega^{-2}\otx\Big|^2+\Omega^4\Big|[r\ns]^i\Omega^{-2}f_0\Big|^2,\\
\Big|\int_{S^2} \olinb_{\Sf}  T{\otxb}_{\Sf}\varepsilon_{S^2}\Big|&\lesssim\Omega^2\Big[|\rlin+\divs\eblin|^2+\sum_{i=0}^1\Big(|[r\ns]^i\olinb|^2+\Big|[r\ns]^i\frac{f_0}{\Omega^2}\Big|^2\Big)+\Big|\frac{\otxb}{\Omega^2}\Big|^2+\Big|[r\ns] \frac{\otx}{\Omega^2}\Big|^2\Big] \, , \\
 \Big|\int_{S^2}T \Big(\frac{f\Omega^2}{r}\Big) T (\divs {\elin} + {\rlin} )_{\Sf}\varepsilon_{S^2} \Big|&\lesssim\Omega^4\Big[\sum_{i=0}^1\Big(\Big|[r\ns]^i\frac{\otx}{\Omega^2}\Big|^2+\Big|[r\ns]^i\frac{f_0}{\Omega^2}\Big|^2\Big)+\Big|\frac{\bblin}{\Omega}\Big|^2+\Big|[r\ns]\frac{\olinb}{\Omega^2}\Big|^2+\Big|\frac{\otxb}{\Omega^2}\Big|^2\Big],
\end{align*}
where, once again, the latter two inequalities follow after integration by parts for angular derivatives on $\divs(\Omega\bblin)_{\Si'}$, $\Delta_{S^2}\olinb$ and $\Delta_{S^2}f_0$. We have therefore established~\eqref{eqn:dataGR10}. 

We finally move to the $T^2$-estimate~\eqref{eqn:dataG20}. We have from (\ref{TTfalongufin}) the identity
\begin{align*}
\Omega^2TTf (u_f,v_0)&\sim\Omega^2\Big[\Omega^{-2}\otx+\olin+\Omega^{-2}f_0\Big] \, .
\end{align*}
From Proposition~\ref{prop:TTransport} we deduce (after converting quantities in the $\Sf'$ gauge into $\Si'$ gauge) on $C_{u_f}$
\begin{align*}
2T^2\otx_{\Sf'}&\sim\Omega^2\Big[T(\divs\elin+\rlin)+\Delta_{S^2} \Big[\otx+f_0\Big] +\otx+f_0+\Omega^2\rlin +\olinb+(\divs\elin+\rlin)\Big],\\
2T^2(\rlin-\divs\eblin)_{\Sf'}&\sim\Omega^2\Big[\divs\big(\ns_T(\Omega^{-1}\bblin))\big)+\Delta_{S^2}\Big[\Omega^{-2}\otx+\Omega^{-2}f_0+\rlin+\olinb\Big]+f_0+\rlin+\Omega^{-2}\otxb+\olinb\Big].
\end{align*}
Further, 
\begin{align*}
2T^2\otxb_{\Si'}&\sim\Omega^2\Big[T(\divs\eblin+\rlin)+\olin +\Big[\partial_u\olinb+\Omega^2\rlin\Big]+\Olin+\Omega^2(\divs\elin+\rlin)+(\divs\eblin+\rlin)+\olinb+\Omega^{-2}\otxb\Big],\\
2T^2\otx_{\mathscr{G}_f}&=2T^2\otx_{\Si'}-2T^2\otx_{\Sf'}\sim\Omega^2\Big[\partial_v\olin+\Olin+\Omega^{-2}\otx+\Omega^2\olin+f_0\Big] \, .
\end{align*}
We can now estimate the following terms in $\mathcal{G}^{2,\mathbf{0}}_{\Gamma}(u_f,v_0)$ via Young's inequalty:
{\allowdisplaybreaks
\begin{align*}
\Big|T^2\frac{\Omega^2 f}{r}\Big|^2&\lesssim\Omega^4\Big[\big| \Omega^{-2}\otx\big|^2+\big|\Omega^{-2}f_0\big|^2+|\olin|^2\Big] ,\\
 \Big|\frac{r^3}{\Omega^2}  T^2\otx_{\mathscr{G}_f}T^2\otxb_{\Si'}\Big|&\lesssim\Omega^2\Big[|\partial_v\olin|^2+|\Olin|^2+\big|\Omega^{-2}\otx\big|^2+|f_0|^2+|T(\divs\eblin+\rlin)|^2+|\olin|^2+|\partial_u\olinb|^2 \\
 &\qquad\qquad+|\Omega^2\rlin|^2+|\Omega^2(\divs\elin+\rlin)|^2+|\divs\eblin+\rlin|^2+|\olinb|^2+\big|\Omega^{-2}\otxb\big|^2\Big]
\end{align*}}
and, similarly, after integrating by parts,
\begin{align*}
\Big|\int_{S^2} \frac{r^3}{\Omega^2}T^2 \left(\frac{f \Omega^2}{r}\right)   T^2 {\otx}_{\Sf'}\varepsilon_{S^2}\Big|&\lesssim\Omega^2\Big[\sum_{i=0}^1\Big(\Big|[r\ns]^i(\Omega^{-2}\otx)\Big|^2+|[r\ns]^i\olin|^2+\Big|[r\ns]^i(\Omega^{-2}f_0)\Big|^2\Big)\\
&\qquad\qquad+|T(\divs\elin+\rlin)|^2+|\Omega^2\rlin|^2 +|\olinb|^2+|\divs\elin+\rlin|^2\Big],\\
 \Big|\int_{S^2}T^2\left(\frac{f \Omega^2}{r}\right)T^2\big( {\rlin}-\divs {\eblin} \big)_{\Sf'}\varepsilon_{S^2}\Big|&\lesssim\Omega^4\Big[\sum_{i=0}^1\Big(\Big|[r\ns]^i\frac{\otx}{\Omega^2}\Big|^2+|[r\ns]^i\olin|^2+|[r\ns]^i\olinb|^2+\Big|[r\ns]^i\frac{f_0}{\Omega^2}\Big|^2\Big)\Big]\\
&\qquad+\Omega^4\Big[|\rlin|^2+\Big|\frac{\otxb}{\Omega^2}\Big|^2+\Big|\Big\langle[r\ns]\Big(\frac{\otx}{\Omega^2}+\olin\Big) ,\ns_T\frac{\bblin}{\Omega}+[r\ns] \rlin\Big\rangle\Big|\\
&\qquad+\Big|\Big\langle[r\ns]\frac{f_0}{\Omega^2},\Dso (\rlin,\slin)+\eblin+\Omega^{-1}\bblin+\divs\ablin\Big\rangle\Big|\Big] \, .
\end{align*}
To establish (\ref{eqn:dataG20}) note now that the terms involving one derivative of curvature on the right hand side, are either already contained as inner products in $\mathbb{D}^2_{data}(u_f,v_0)$ (see Section \ref{sec:dataenergies}) or can be further reduced by an integration by parts. (The latter applies to the terms in the last line.)
\end{proof}

\subsection{Proof of Theorem \ref{prop:firstround}} \label{sec:2oe}
We now embark on the proof of Theorem \ref{prop:firstround} proper. The proof will consist in applying, in a specific order, various appropriately commuted conservation laws. For each such conservation law, we will eventually send the ingoing cone $\underline{C}_{v_f}$ to infinity, $v_f \rightarrow \infty$. To simplify notation, we will, for the entire proof, denote by $R(v_k)$ a smooth function defined on the cone $\underline{C}_{v_f} \cap [u_0,u_f]$ (or one of the spheres $S^2_{u_0,v_f}$, $S^2_{u_f,v_f}$) which vanishes in the limit $v_f \rightarrow \infty$. The vanishing of $R(v_k)$ is always a consequence of Definition~\ref{def:extendskri}. In this context, we will also make use of the following general identity 
\begin{align} \label{niflux}
 F^{i,{\bf k}}_{v_f} \left[\Gamma , \Si^\prime \right] (u_0,u_f) =  \int_{u_0}^{u_f} \int_{S^2}  |\mathcal{K}^{i,{\bf k}}r\Omega\xblin|^2 du\varepsilon_{S^2} +   \left[\int_{S^2} \frac{r^3}{2} \mathcal{K}^{i,{\bf k}}\otx\mathcal{K}^{i,{\bf k}}\otxb \varepsilon_{S^2}\right]^{(u_f,v_f)}_{(u_0,v_f)}  + R(v_f) ,
\end{align}
where all quantities on the right-hand side are evaluated at $v_f$, which is proven as in \cite{ghconslaw}, Section~4.1.2. 

We first provide a short overview of the proof in Section \ref{sec:logic} before carrying out each step in detail.

\subsubsection{The Logic of the Proof of Theorem \ref{prop:firstround}}  \label{sec:logic}
We first prove (step 1) that for a constant $C_1$ depending only on $M$
\begin{align} \label{step1m}
{\mathbb{E}}^{0,0}[\Gamma, \Sf^\prime] (u_{f}) +  \Big\| \frac{f}{r} \Big\|^2_{u_f,\infty} &\leq {\mathbb{E}}_{data}^{0,0}[\Gamma](u_{f}) + C_1 \cdot {\mathbb{D}}_{data}^1(u_f,v_0) \, .
\end{align}
This is essentially the estimate of \cite{ghconslaw}. Note that this implies, as a corollary, positivity of the initial data energy
\begin{align}
 \overline{\mathbb{E}}^{0,0}_{data} [\Gamma] (u_f):= {\mathbb{E}}_{data}^{0,0}[\Gamma](u_{f})+ C_1 \cdot {\mathbb{D}}_{data}^1(u_f,v_0)
\end{align}
appearing on the right, something that is not manifest from the form of the energy. 

Next we prove (step 2) that for a constant $C_2$ depending only on $M$
\begin{align} \label{step1mb}
{\mathbb{E}}^{0,0}[\mathcal{R}, \Sf^\prime] (u_{f}) + \Big\| \left[r\slashed{\nabla}\right] \left(\frac{f}{r} \right)\Big\|^2_{u_f,\infty} 
\leq {\mathbb{E}}_{data}^{0,0}[\mathcal{R}](u_{f}) + C_2 \left( {\mathbb{D}}_{data}^1(u_f,v_0) + \overline{\mathbb{E}}^{0,0}_{data}[\Gamma] (u_f)\right) \, .
\end{align}
Note that this implies positivity of the initial data energy
\begin{align}
 \overline{\mathbb{E}}^{0,0}_{data} (u_f):= {\mathbb{E}}_{data}^{0,0}[\mathcal{R}](u_{f}) + C_2 \left( {\mathbb{D}}_{data}^1(u_f,v_0) + \overline{\mathbb{E}}^{0,0}_{data} [\Gamma] (u_f)\right) \, .
\end{align}

In step 3, we use elliptic  and transport estimates with the fluxes controlled in Steps 1 and 2 to deduce
\begin{align}
 \int_{v_0}^\infty dv r^2 \left(\| r \slashed{\nabla} \Omega \xlin\|^2_{u_f,v} +  \|r \slashed{\nabla} \elin\|^2_{{u_{f},v}}+\|r \slashed{\nabla} \eblin\|^2_{{u_{f},v}}\right) &\lesssim  \overline{\mathbb{E}}^{0,0}_{data} (u_f) \, ,  \\
\sup_ v \|\olinb\|^2_{u_{f},v} + \sup_ v \|r \otxb\|^2_{{u_{f},v}} +\|r\Omega\xblin\|^2_{u_f,v} &\lesssim  \overline{\mathbb{E}}^{0,0}_{data} (u_f) \, .
\end{align}

In step 4, we show that there exists a constant $C_3$ depending only on $M$ such that 
\begin{align} \label{step2m}
{\mathbb{E}}^{0,1}[\mathcal{R}, \Sf^\prime] (u_{f})  +  \Big\| \left[r\slashed{\nabla}\right]^2 \left(\frac{f}{r} \right)\Big\|^2_{u_f,\infty} 
\leq {\mathbb{E}}_{data}^{0,1}[\mathcal{R}](u_{f}) + C_3 \left( {\mathbb{D}}_{data}^1(u_f,v_0) + \overline{\mathbb{E}}^{0,0}_{data} (u_f) \right) \, .
\end{align}
Note that this implies positivity of the initial data energy
\begin{align}
 \overline{\mathbb{E}}^{0,1}_{data} (u_f):=  {\mathbb{E}}_{data}^{0,1}[\mathcal{R}](u_{f}) + C_3 \left( {\mathbb{D}}_{data}^1(u_f,v_0) + \overline{\mathbb{E}}^{0,0}_{data} (u_f) \right) \, .
\end{align}

In step 5, we show there exists a constant $C_4$ depending only on $M$ such that 
\begin{align} \label{step5m}
{\mathbb{E}}^{1,0}[\Gamma, \Sf^\prime] (u_{f})&\leq  {\mathbb{E}}_{data}^{1,0}[\Gamma](u_{f}) + C_4 \left(\mathbb{D}_{data}^1(u_f,v_0) +  \overline{\mathbb{E}}^{0,1}_{data} (u_f)\right) \, .
\end{align}
Note that this implies positivity of the initial data energy
\begin{align}
 \overline{\mathbb{E}}^{1,0}_{data} [\Gamma] (u_f):={\mathbb{E}}_{data}^{{1,0}}[\Gamma](u_{f}) + C_4 \left(\mathbb{D}_{data}^1(u_f,v_0) +  \overline{\mathbb{E}}^{0,1}_{data} (u_f)\right) \, .
\end{align}
In step 6, we show there exists a constant $C_5$ depending only on $M$ such that 
\begin{align} \label{step6m}
{\mathbb{E}}^{1,0}[\mathcal{R}, \Sf^\prime] (u_{f})&\lesssim  {\mathbb{E}}_{data}^{1,0}[\mathcal{R}](u_{f}) + C_5 \left(\mathbb{D}_{data}^1(u_f,v_0) +  \overline{\mathbb{E}}^{1,0}_{data}[\Gamma] (u_f)\right) \, .
\end{align}
Note that this implies positivity of the initial data energy
\begin{align}
 \overline{\mathbb{E}}^{1,0}_{data} (u_f):={\mathbb{E}}_{data}^{1,0}[\mathcal{R}](u_{f}) + C_5 \left(\mathbb{D}_{data}^1(u_f,v_0) +  \overline{\mathbb{E}}^{1,0}_{data} [\Gamma] (u_f)\right) \, .
\end{align}
In step 7, we use transport and elliptic estimates to prove estimates on 
\begin{align}
 \sup_ v \|[r\slashed{\nabla}]\olinb , \Omega \slashed{\nabla}_3 \olinb \|^2_{{u_{f},v}} + \int_{v_0}^\infty dv \frac{\Omega^2}{r^2} \left( \|[r^2\slashed{\Delta}]\olinb \|^2_{{u_{f},v}} +  \|[r\slashed{\nabla}] \bblin \|^2_{{u_{f},v}} \right) \lesssim  \overline{\mathbb{E}}^{1,0}_{data} (u_f)+\mathbb{D}^2_{data}(u_f,v_0) \, .
\end{align}
In step 8, we show there exists a constant $C_6$ depending only on $M$ such that 
\begin{align} \label{step8m}
{\mathbb{E}}^{2,0}[\Gamma, \Sf^\prime] (u_{f})&\leq  {\mathbb{E}}_{data}^{2,0}[\Gamma](u_{f}) + C_6 \left(\mathbb{D}_{data}^2(u_f,v_0) +  \overline{\mathbb{E}}^{1,0}_{data} (u_f)\right) \, .
\end{align}
Note that this implies positivity of the initial data energy
\begin{align}
 \overline{\mathbb{E}}^{2,0}_{data} (u_f):={\mathbb{E}}_{data}^{2,0}[\Gamma](u_{f}) + C_6 \left(\mathbb{D}_{data}^2(u_f,v_0) +  \overline{\mathbb{E}}^{1,0}_{data} (u_f)\right) \, .
\end{align}
Adding (\ref{step8m}), (\ref{step6m}), (\ref{step5m}), (\ref{step2m}), (\ref{step1mb}) and (\ref{step1m}) and setting $C_{max} = \sum_{i=1}^6 C_i$ yields the desired (\ref{secondorderest}).

\subsubsection{Step 1: Proof of the Uncommuted $\Gamma$-estimate}

We first prove (\ref{step1m}). The $\Gamma$-conservation law in the $\Si'$ gauge gives:
\begin{align} \label{lowlcl}
F^{0,0}_{v_f} \left[\Gamma ,\Si' \right](u_0,u_f)+F^{0,0}_{u_f} \left[\Gamma ,\Si' \right](v_0,v_f)=F^{0,0}_{v_0} \left[\Gamma ,\Si' \right](u_0,u_f)+F^{0,0}_{u_0} \left[\Gamma ,\Si' \right](v_0,v_f),
\end{align}
where, using Proposition~\ref{prop:gaugechange}, we have
\begin{align*}
F^{0,0}_{u_f} [\Gamma , \Si'] \left(v_0,v_f\right)=& F^{0,0}_{u_f} [\Gamma , \Sf'] \left(v_0,v_f\right)
+ \left[\int_{S^2} \mathcal{G}_{\Gamma}^{0,0} (u,v,\theta,\phi) \varepsilon_{S^2} \right]_{(u_f,v_0)}^{(u_f,v_f)}
\end{align*}
with
\begin{align*}
F^{0,0}_{u_f} [\Gamma ,\Sf' ](v_0,v_f)= \int_{v_0}^{v_f} \int_{S^2}  \Big[&|\Omega\xlin_{\Sf'}|^2+2 |\Omega\eblin_{\Sf'}|^2\Big]r^2dv\varepsilon_{S^2}\, ,
\end{align*}
and invoking (\ref{niflux}),
\begin{align*}
F^{0,0}_{v_f} [\Gamma , \Si'] \left(u_0,u_f\right)=\int_{u_0}^{u_f} \int_{S^2}  |\Omega\xblin_{\Si'}|^2 r^2 du\varepsilon_{S^2} + \Big[ \frac{1}{2} \int_{S^2} r^3 \otx_{\Si'}   \otxb_{\Si'} \varepsilon_{S^2}\Big]_{(u_0,v_f)}^{(u_f,v_f)}+\mathrm{R}(v_f) \,  ,
\end{align*}
where we recall $\mathrm{R}(v_f)\rightarrow 0$ as $v_f\rightarrow \infty$. Using the estimates on $\mathcal{G}_{\Gamma}^{0,0}$ from (\ref{GammaGnone}) and (\ref{eqn:dataG00}), the identity (\ref{lowlcl}) turns into the estimate (note the cancellation of the $\otx_{\Si^\prime} \otxb_{\Si^\prime}$-term at $(u_f,v_f)$)
\begin{align*}
&F^{0,0}_{u_f} [\Gamma , \Sf'] \left(v_0,v\right)+\int_{u_0}^{u_f} \int_{S^2}  |\Omega\xblin_{\Si'}|^2 r^2 du\varepsilon_{S^2}+\int_{S^2}4M\Big| \frac{\Omega^2 f}{r}\Big|^2\varepsilon_{S^2}\Big|_{(u_f,v_f)} - \tilde{C} \cdot \mathbb{D}^1_{data}(u_f,v_0) \\
&\leq F^{0,0}_{v_0} \left[\Gamma ,\Si' \right](u_0,u_f)+F^{0,0}_{u_0} \left[\Gamma ,\Si' \right](v_0,v_f)+\frac{1}{2} \int_{S^2} r^3 \otx_{\Si'}   \otxb_{\Si'} \varepsilon_{S^2}\Big|_{(u_0,v_f)}+\mathrm{R}(v_f) \, .
\end{align*}
Taking the limit $v_f \rightarrow \infty$ and recalling definition (\ref{Ei0datadef}) we deduce 
\begin{align} \label{basicuc}
\int_{v_0}^{\infty} \int_{S^2}  \Big[  r^2 |  \Omega\xlin_{\Sf^\prime}|^2+r^2 \Omega^2 | \eblin_{\Sf^\prime}|^2  \Big]\Big|_{u=u_{f}}dv\varepsilon_{S^2} + 4M \Big\|  \frac{f}{r} \Big\|^2_{u_f,\infty} 
\leq  {\mathbb{E}}_{data}^{0,0}[\Gamma](u_{f}) + \tilde{C} {\mathbb{D}}_{data}^1(u_{f},v_0) .
\end{align}

\subsubsection{Step 2: Proof of the Uncommuted $\mathcal{R}$-estimate}
To prove (\ref{step1mb}), we apply the non-commuted second conservation law for $\Si'$:
\begin{align*}
F^{0,\bf{0}}_{v_f} [\mathcal{R},\Si' ](u_0,u_f)+F^{0,\bf{0}}_{u_f} [\mathcal{R},\Si' ](v_0,v_f)=F^{0,\bf{0}}_{v_0} [\mathcal{R},\Si' ](u_0,u_f)+F^{0,\bf{0}}_{u_0} [\mathcal{R},\Si' ](v_0,v_f) \, , 
\end{align*}
where
\begin{align*}
F^{0, {\bf 0}}_{u_f} [\mathcal{R}, \Si'] \left(v_0,v_f\right)&= F^{0, {\bf 0}}_{u_f} [\mathcal{R}, \Sf'] \left(v_0,v_f\right)+\Big[\int_{S^2}\mathcal{G}_{\mathcal{R}}^{0,0}(u,v,\theta,\phi)\varepsilon_{S^2}\Big]_{(u_f,v_0)}^{(u_f,v_f)},\nonumber\\
F^{0, {\bf 0}}_{u_f} [\mathcal{R}, \Sf'] \left(v_0,v_f\right)&=\int_{v_0}^{v_f} \int_{S^2}  \Big[ \frac{1}{2} r^4  |\Omega\blin_{\Sf'}|^2+\frac{1}{2} \Omega^2 r^4 \left(|\rlin_{\Sf'}|^2 + |\slin_{\Sf'}|^2\right) -3Mr  |\Omega\eblin_{\Sf'}|^2 \Big]dv\varepsilon_{S^2} \, , \nonumber\\
F^{0, {\bf 0}}_{v_f} \left[\mathcal{R},\Si' \right] \left(u_0,u_f\right)&= \int_{u_0}^{u_f} \int_{S^2}  \frac{1}{2} r^4 | \Omega\bblin_{\Si' }|^2du\varepsilon_{S^2}+\mathrm{R}(v_f).
\end{align*}
Therefore, combining and using equation~\eqref{RG} from proposition~\ref{prop:Gatinfinity} gives
\begin{align} \label{combok}
&\int_{v_0}^{v_f} \int_{S^2} \Big[ \frac{r^4}{2}   |\Omega\blin_{\Sf'}|^2+\frac{\Omega^2 r^4}{2}  \left(|\rlin_{\Sf'}|^2 + |\slin_{\Sf'}|^2\right) \Big]dv\varepsilon_{S^2}+\int_{u_0}^{u_f}\int_{S^2}  \frac{r^4}{2}  | \Omega\bblin_{\Si' }|^2du\varepsilon_{S^2}\\
&+\frac{3M}{2} \int_{S^2}\Big| r\ns \frac{f \Omega^2}{r} \Big|^2 \varepsilon_{S^2}\Big|_{(u_f,v_f)}=\int_{v_0}^{v_f} \int_{S^2}   3Mr  |\Omega\eblin_{\Sf'}|^2 dv\varepsilon_{S^2}+3M \int_{S^2}\Big| \frac{f \Omega^2}{r} \Big|^2 \varepsilon_{S^2}\Big|_{(u_f,v_f)} \nonumber \\
&\quad +\int_{S^2} \mathcal{G}_{\mathcal{R}}^{0,0}(u,v,\theta,\phi)\varepsilon_{S^2}\Big|_{(u_f,v_0)} +F^{0,\bf{0}}_{v_0} [\mathcal{R},\Si' ](u_0,u_f)+F^{0,\bf{0}}_{u_0} [\mathcal{R},\Si' ](v_0,v_f)+\mathrm{R}(v_f). \nonumber
\end{align}
The first  two terms on the right-hand side of (\ref{combok}) is also controlled by (\ref{basicuc}). 
The last two terms on the right hand side of (\ref{combok}) are equal to ${\mathbb{E}}_{data}^{0,0}[\mathcal{R}](u_{f})$. The remaining $\mathcal{G}_{\mathcal{R}}^{0,0}$ is estimated by $\mathbb{D}^1_{data}(u_f,v_0)$ via~\eqref{eqn:dataGR00}. It follows that there exist constants $\tilde{C}_1, \tilde{C}_2$ depending only on $M$ such that
\begin{align}
\int_{v_0}^{\infty}\int_{S^2} \Big[ \frac{r^4}{2}  |\Omega\blin_{\Sf'}|^2+\frac{\Omega^2 r^4}{2}  \left(|\rlin_{\Sf'}|^2 + |\slin_{\Sf'}|^2\right) \Big]dv\varepsilon_{S^2}+\int_{u_0}^{u_f}\int_{S^2} \frac{r^4}{2} | \Omega\bblin_{\Si' }|^2du\varepsilon_{S^2} \nonumber \\+ \int_{S^2} \frac{3M}{2} \Big| r\ns \frac{f \Omega^2}{r} \Big|^2\varepsilon_{S^2}\Big|_{(u_f,\infty)} 
\leq  {\mathbb{E}}_{data}^{0,0}[\mathcal{R}](u_{f}) + \tilde{C}_1  \mathbb{D}^1_{data}(u_f, v_0) + \tilde{C}_2 \left(  {\mathbb{E}}_{data}^{0,0}[\Gamma](u_{f}) + \tilde{C} {\mathbb{D}}_{data}^1(u_{f},v_0) \right),  \nonumber
\end{align}
which is the required estimate (\ref{step1mb}).

\subsubsection{Step 3: Refined Bounds on $C_{u_{f}}$ via Elliptic and Transport Equations I}
We now use the control obtained so far to deduce further bounds for quantities on the cone $C_{u_{f}}$. We begin with a flux bound:

\begin{proposition} \label{prop:angularetab}
We have the following estimate along $C_{u_f}$:
\begin{align}
\int_{v_0}^{\infty}  \int_{S^2} \Big[ r^2 \Omega^2 | [r\slashed{\nabla}]\xlin_{\Sf^\prime}|^2 + r^2 \Omega^2 | [r\slashed{\nabla}]\eblin_{\Sf^\prime}|^2 + r^2 \Omega^2 | [r\slashed{\nabla}]\elin_{\Sf^\prime}|^2  \Big]\Big|_{u_f}dv\varepsilon_{S^2} \lesssim  \overline{\mathbb{E}}^{0,0}_{data} (u_f) \, .
\end{align}
\end{proposition}

\begin{proof}
We begin with the bound on $\xlin$. From Codazzi, we have 
\begin{align*}
|r \divs(\Omega\xlin_{\Sf'})|^2&=\Big|\Omega^2\eblin_{\Sf'}+r\Omega\blin_{\Sf'}\Big|^2\leq 2|\Omega^2\eblin_{\Sf'}|^2+2|r\Omega\blin_{\Sf'}|^2 \leq 2 \Big(|\Omega\eblin_{\Sf'}|^2+r^2|\Omega\blin_{\Sf'}|^2\Big) \, , 
\end{align*}
hence
\begin{align*}
\int_{v_0}^{\infty} \int_{S^2} |r\divs(\Omega\xlin_{\Sf'})|^2r^2dv\varepsilon_{S^2}\leq 2\int_{v_0}^{\infty}\int_{S^2} \Big(|\Omega\eblin_{\Sf'}|^2+r^2|\Omega\blin_{\Sf'}|^2\Big)r^2dv\varepsilon_{S^2} \, ,
\end{align*}
the right-hand side of which we already control from~\eqref{step1m} and~\eqref{step1mb}. Next recall the elliptic identity
\begin{align*}
\int_{S^2}|r\divs(\Omega\xlin_{\Sf'})|^2\varepsilon_{S^2}=\frac{1}{2}\int_{S^2}\Big[|r \ns(\Omega
\xlin)_{\Sf'}|^2+|2\Omega\xlin_{\Sf'}|^2\Big]\varepsilon_{S^2},
\end{align*}
which immediately yields the required estimate
\begin{align*}
\int_{v_0}^{\infty} \int_{S^2} ||r\ns(\Omega\xlin_{\Sf'})|^2r^2dv\varepsilon_{S^2}\lesssim {\mathbb{E}}_{data}^{0,0}[\mathcal{R}](u_{f})  + C \left(  {\mathbb{E}}_{data}^{0,0}[\Gamma](u_{f}) + \tilde{C} {\mathbb{D}}_{data}^1(u_{f},v_0) \right).
\end{align*}
We turn to the fluxes of $\elin, \eblin$. In view of $\eblin_{\Sf^\prime} =- \elin_{\Sf^\prime}$ along $C_{u_f}$ it suffices to show this for $[r\slashed{\nabla}] \elin_{\Sf^\prime}$ and in view of elliptic estimates to show this separately for 
$r\divs  \elin_{\Sf^\prime}$ and $r\curls \elin_{\Sf^\prime}$. The latter follows from (\ref{curleta}) and control on the $\sigma$-flux in (\ref{step1mb}). For $r\divs  \elin_{\Sf^\prime}$, Lemma~\ref{lem:fromoldpaper} gives
\begin{align*}
\int_{v_0}^{\infty} \int_{S^2} \frac{\Omega^2}{r^2}|r^3(\divs\elin+\rlin)_{\Sf'}|^2dv\varepsilon_{S^2}\leq {\mathbb{D}}^1_{data}(u_f,v_0)\int_{v_0}^{\infty}\frac{\Omega^2}{r^2}dv=\frac{1}{2M}{\mathbb{D}}^1_{data}(u_f,v_0).
\end{align*}
Since we have control on the flux of $|\Omega r^2\rlin|^2$ through (\ref{step1mb}), we conclude
\begin{align*}
\int_{v_0}^{\infty}\int_{S^2} r^4|\Omega\divs\elin_{\Sf'}|^2dv\varepsilon_{S^2}\leq \int_{v_0}^{\infty}\int_{S^2}\Big( \frac{\Omega^2}{r^2}|r^3(\divs\elin+\rlin)_{\Sf'}|^2+ \Omega^2r^4|\rlin_{\Sf'}|^2\Big)dv\varepsilon_{S^2}\lesssim \overline{\mathbb{E}}^{0,0}_{data} (u_f) 
\end{align*}
as desired. 
\end{proof}

We next deduce a few $L^\infty_v$ bounds for spheres along the cone $C_{u_f}$, all of which rely on the following basic transport lemma:

\begin{lemma} \label{lem:tplemma}
Let $\xi, \Xi$ be $S^2_{u,v}$ tensors satisfying $\Omega \slashed{\nabla}_4 \xi = \Omega^2 \Xi$ along the cone $C_{u}$ for some $u_0 \leq u \leq \infty$. Then for any $v \geq v_0$:
\[
\| \xi\|_{u,v} \leq \|\xi\|_{u,v_0} + \int_{v_0}^v \Omega^2 \| \Xi \|_{u,\bar{v}} d\bar{v}  \leq  \|\xi\|_{u,v_0} + \frac{1}{\sqrt{r(u,v_0)}} \sqrt{\int_{v_0}^v \Omega^2 r^2 \| \Xi \|^2_{u,\bar{v}} d\bar{v} }.
\]
\end{lemma}
\begin{proof}
See Lemma~3.1 of~\cite{ChristBH} for the first inequality and apply Cauchy-Schwarz using $\Omega^2=\partial_v r$ for the second.
\end{proof}

\begin{proposition} \label{prop:ombpure}
We have $\olinb_{\Si^\prime}=\olinb_{\Sf^\prime}$ on  the sphere $S^2_{u_{f},\infty}$. Moreover we have
\begin{align} \label{ombprop}
\sup_{v \in [v_0,\infty)} \|\olinb_{\Sf^\prime}\|^2_{u_f,v} 
&\lesssim  \overline{\mathbb{E}}^{0,0}_{data} (u_f)+ {\mathbb{D}}_{data}^1(u_f,v_0) \, .
\end{align}
\end{proposition}
\begin{proof}
The claim $\olinb_{\Si^\prime}=\olinb_{\Sf^\prime}$ on $S^2_{u_f,\infty}$ is immediate from the relation $\olinb_{\Sf^\prime} = \olinb_{\Si^\prime}- \frac{2M}{r^3} \Omega^2 f$ (see Lemma~\ref{lem:exactsol}) and the fact that $\frac{f}{r}$ is bounded in the limit as $v \rightarrow \infty$. As for the relation at $v=v_0$, we have $\olinb_{\Sf^\prime} = \olinb_{\Si^\prime}- \frac{M}{r^2}f_0$. Hence, 
\begin{align} \label{btom}
\| \olinb_{\Sf^\prime}\|_{u_f,v_0}\leq \|\olinb_{\Si^\prime}\|_{u_f,v_0}+\Big\|\frac{M}{r^2}f_0\Big\|_{u_f,v_0} \lesssim \sqrt{\mathbb{D}^1_{data}}\, .
\end{align}
Finally, one applies the basic transport lemma to (\ref{oml1}) resulting in the estimate
\begin{align*}
||\olinb_{\Sf'}||_{u_f,v}\leq ||\olinb_{\Sf'}||_{u_f,v_0}+\int_{v_0}^{v}\Big\|\Omega^2\rlin_{\Sf'}\Big\|_{u_f,v}dv \leq ||\olinb_{\Sf'}||_{u_f,v_0} + \frac{1}{(r(u_f,v_0))^\frac{3}{2}}\sqrt{\int_{v_0}^{v} r^4 \Omega^2 \Big\|\rlin_{\Sf'}\Big\|_{u_f,v}^2 dv} \, ,
\end{align*}
where we have applied Cauchy-Schwarz in the last step. Squaring and using (\ref{btom}) yields  (\ref{ombprop}).
\end{proof}

We can also obtain uniform bounds along the cone $C_{u_f}$ for $\xblin$ and $\otxb$:
\begin{proposition} \label{prop:chibarpure}
We have 
\begin{align} \label{trxbprop}
\sup_{v \in [v_0,\infty)} \|r \otxb_{\Sf^\prime}\|^2_{u_f,v} 
\lesssim   \overline{\mathbb{E}}^{0,0}_{data} (u_f)+ {\mathbb{D}}_{data}^1(u_f,v_0) \, ,
\end{align}
\begin{align} \label{trxbpropb}
\sup_{v \in [v_0,\infty)} \|r \Omega \xblin_{\Sf^\prime}\|^2_{u_f,v} 
\lesssim  \overline{\mathbb{E}}^{0,0}_{data} (u_f)+ {\mathbb{D}}_{data}^1(u_f,v_0) \, .
\end{align}
\end{proposition}

\begin{proof}
Estimate \eqref{trxbprop} follows from control on this quantity on the sphere $S^2_{u_f,v_0}$ and using the transport equation (\ref{dtcb}), which, in $\Sf^\prime$-gauge along $C_{u_f}$, can be written as
\begin{align*}
\Omega \slashed{\nabla}_4 (r\otxb)  = \Omega^2 r \left( 2 \divs\, \eblin + 2\rlin \right)  \, .
\end{align*}
 Applying Lemma~\ref{lem:tplemma} we have
\begin{align*}
\|r\otxb_{\Sf'}\|^2_{u_f,v}\lesssim \|r\otxb_{\Sf'}\|^2_{u_f,v_0}+\int_{v_0}^v \Omega^2 r^4 \|(\divs \eblin+\rlin)_{\Sf'}\|^2_{u,\bar{v}}d\bar{v}.
\end{align*}
Note that the first term on the right can be converted to $r\otxb_{\Si^\prime}$ using Lemma~\ref{lem:exactsol} after which it is controlled by ${\mathbb{D}}_{data}^1(u_f,v_0)$. The estimate~\eqref{trxbprop} now follows from the fluxes controlled in equation~(\ref{step1mb}) and Proposition~\ref{prop:angularetab}. 

Estimate (\ref{trxbpropb}) is similar now using the transport equation (\ref{chih3b}), which along $C_{u_f}$ takes the form
\begin{align*}
\Omega \slashed{\nabla}_4  \left(r \Omega \xblin  \right) =  \Omega^2 \Omega \xlin-2 \Omega^2 r \slashed{\mathcal{D}}_2^\star \eblin \, . 
\end{align*}
Applying Lemma~\ref{lem:tplemma} we have
\begin{align*}
\|r\Omega\xblin_{\Sf’}\|^2_{u_f,v}\lesssim \|r\Omega\xblin_{\Sf'}\|^2_{u_f,v_0}+\int_{v_0}^v \Omega^2 r^2 \|r\Dst \eblin_{\Sf’}\|_{u_f,\bar{v}}^2 d\bar{v} +\int_{v_0}^v  \Omega^2 r^2 \| \Omega\xlin_{\Sf^\prime} \|^2_{u_f,\bar{v}} d\bar{v} \, .
\end{align*}
Note that the first term on the right can be converted to $\xblin_{\Si^\prime}$ using Lemma~\ref{lem:exactsol} after which it is controlled by ${\mathbb{D}}_{data}^1(u_f,v_0)$. The estimate~\eqref{trxbpropb} follows from the fluxes controlled in \eqref{step1mb} and Proposition~\ref{prop:angularetab}. 
\end{proof}

\subsubsection{Step 4: Proof of the Angular Commuted Conservation Law}

We now prove (\ref{step2m}).  Let us fix the notation that
\begin{align*}
F^{0,{1}}_{v} [\mathcal{R},\Si' ](u_0,u)=\sum_{|\mathbf{k}|=1}F^{0,\bf{k}}_{v} [\mathcal{R},\Si' ](u_0,u),\qquad F^{0,{1}}_{u} [\mathcal{R},\Si' ](v_0,v)=\sum_{|\mathbf{k}|=1}F^{0,\bf{k}}_{u} [\mathcal{R},\Si' ](v_0,v).
\end{align*}
We apply the once angular commuted second conservation law (with summation):
\begin{align*}
F^{0,{1}}_{v_f} [\mathcal{R},\Si' ](u_0,u_f)+F^{0,{1}}_{u_f} [\mathcal{R},\Si' ](v_0,v_f)=F^{0,{1}}_{v_0} [\mathcal{R},\Si' ](u_0,u_f)+F^{0,{1}}_{u_0} [\mathcal{R},\Si' ](v_0,v_f) \, , 
\end{align*}
where, using Proposition~\ref{prop:gaugechange} and Proposition~\ref{prop:angident}, we have
\begin{align*}
F^{0, { 1}}_{u_f} [\mathcal{R}, \Si'] \left(v_0,v_f\right)&= F^{0, {\bf 1}}_{u_f} [\mathcal{R}, \Sf'] \left(v_0,v_f\right)+\left[\int_{S^2}\mathcal{G}_{\mathcal{R}}^{0,1}(u,v,\theta,\phi)\varepsilon_{S^2}\right]_{(u_f,v_0)}^{(u_f,v_f)},\\
F^{0, { 1}}_{u_f} [\mathcal{R}, \Sf'] \left(v_0,v_f\right)&=\int_{v_0}^{v_f} \int_{S^2} \Big[ \frac{1}{2} r^4  |[r\ns]\Omega\blin_{\Sf'}|^2+\frac{1}{2} r^4  |\Omega\blin_{\Sf'}|^2+\frac{1}{2} \Omega^2 r^4 \left(|[r\ns]\rlin_{\Sf'}|^2 + |[r\ns]\slin_{\Sf'}|^2\right)\\
&\qquad -3Mr  |[r\ns]\Omega\eblin_{\Sf'}|^2-3Mr  |\Omega\eblin_{\Sf'}|^2 \Big]dv\varepsilon_{S^2} \, , \nonumber\\
F^{0, { 1}}_{v_f} \left[\mathcal{R},\Si' \right] \left(u_0,u_f\right)&= \int_{u_0}^{u_f} \int_{S^2} \Big[\frac{1}{2} r^4 | [r\ns]\Omega\bblin_{\Si' }|^2+ \frac{1}{2} r^4 | \Omega\bblin_{\Si' }|^2\Big]du\varepsilon_{S^2}+\mathrm{R}(v_f).
\end{align*}
Therefore, combining, and using Proposition~\ref{prop:Gatinfinity}, gives
\begin{align*}
&\int_{v_0}^{v_f}\int_{S^2} \Big[ \frac{1}{2} r^4  |[r\ns]\Omega\blin_{\Sf'}|^2+\frac{1}{2} r^4  |\Omega\blin_{\Sf'}|^2+\frac{1}{2} \Omega^2 r^4 \big(|[r\ns]\rlin_{\Sf'}|^2 + |[r\ns]\slin_{\Sf'}|^2\big) \Big]dv\varepsilon_{S^2} \\ &+\int_{S^2}\Big[\frac{3M}{2} \Big| [r\ns]^2 \frac{f \Omega^2}{r} \Big|^2\varepsilon_{S^2}\Big|_{v_f}
+\int_{u_0}^{u_f} \int_{S^2} \Big[\frac{1}{2} r^4 | [r\ns]\Omega\bblin_{\Si' }|^2+ \frac{1}{2} r^4 | \Omega\bblin_{\Si' }|^2\Big]du\varepsilon_{S^2}\\
&=F^{0,{1}}_{v_0} [\mathcal{R},\Si' ](u_0,u_f)+F^{0,{1}}_{u_0} [\mathcal{R},\Si' ](v_0,v_f)+\int_{S^2}3M \Big| r\ns \frac{f \Omega^2}{r} \Big|^2\varepsilon_{S^2}\Big|_{(u_f,v_f)}+\mathrm{R}(v_f)\\
&+\int_{S^2}\mathcal{G}_{\mathcal{R}}^{0,1}(u,v,\theta,\phi)\varepsilon_{S^2}\Big|_{(u_f,v_0)}+\int_{v_0}^{v_f} \int_{S^2}\Big[3Mr  |[r\ns]\Omega\eblin_{\Sf'}|^2+3Mr  |\Omega\eblin_{\Sf'}|^2 \Big]dv\varepsilon_{S^2}.
\end{align*}
Taking the limit $v_f \rightarrow \infty$ the desired estimate (\ref{step2m}) follows after using that we have control over the terms on the right-hand side from~\eqref{eqn:dataGR01} and~\eqref{step1mb} and Proposition \ref{prop:angularetab}.

\subsubsection{Step 5: Closing the Once $T$-commuted $\Gamma$-estimate}

We now prove (\ref{step5m}). We apply the first conservation law now once $\mathcal{L}_T$-commuted:
\begin{align} \label{mainT1}
F^{1,\bf{0}}_{v_f} [\Gamma,\Si' ](u_0,u_f)+F^{1,\bf{0}}_{u_f} [\Gamma,\Si' ](v_0,v_f)=F^{1,\bf{0}}_{v_0} [\Gamma,\Si' ](u_0,u_f)+F^{1,\bf{0}}_{u_0} [\Gamma,\Si' ](v_0,v_f) \, .
\end{align}
By the change of gauge formula on the cone $C_{u_f}$ of Proposition~\ref{prop:gaugechange} and Proposition~\ref{prop:Gatinfinity} we have
\begin{align} \label{changeofgT1}
F^{1, {\bf 0}}_{u_f} [\Gamma, \Si'] \left(v_0,v_f\right)&= F^{1, {\bf 0}}_{u_f} [\Gamma, \Sf'] \left(v_0,v_f\right)-\int_{S^2}\mathcal{G}^{1,0}_{\Gamma}(u,v,\theta,\phi)\varepsilon_{S^2}\Big|_{(u_f,v_0)} \nonumber \\
&\qquad- \frac{r^3}{2\Omega^{2}} \int_{S^2} T\otx_{\Si^\prime}T\otxb_{\Si^\prime} \varepsilon_{S^2}\Big|_{(u_f,v_f)}+\mathrm{R}(v_f) \, ,
\end{align}
with the new flux given by
\begin{align} \label{Tfluxng}
F^{1, {\bf 0}}_{u_f} [\Gamma, \Sf'] \left(v_0,v_f\right)&=\int_{v_0}^{v_f} \int_{S^2} \Bigg[  |\ns_{T}\Omega\xlin_{\Sf'}|^2+2\Omega^2 |\ns_{T}\eblin_{\Sf'}|^2 -2 \ns_{T}\olin_{\Sf'} \ns_{T}\otxb_{\Sf'} - \frac{1}{2} [\ns_{T}\otx_{\Sf'}]^2 \nonumber \\
&\qquad + \frac{4M}{r^2} \ns_{T}\Big(\frac{\Olino}{\Omega}\Big)_{\Sf'} \ns_{T}\otx_{\Sf'}  \Bigg]r^2dv\varepsilon_{S^2} \, . 
\end{align}
For the ingoing cone $\underline{C}_{v_f}$ we have from (\ref{niflux}) the limiting expression
\begin{align*}
F^{1, {\bf 0}}_{v_f} \left[\Gamma,\Si' \right] \left(u_0,u_f\right)&=  \int_{u_0}^{u_f} \int_{S^2} |\ns_{T}\Omega\xblin_{\Si'}|^2r^2du\varepsilon_{S^2}
+\left[\frac{r^3}{2}\int_{S^2}T\otx_{\Si^\prime}T\otxb_{\Si^\prime}\varepsilon_{S^2}\right]_{(u_0,\infty)}^{(u_f,\infty)}+\mathrm{R}(v_f)\, .
\end{align*}
Note the cancellation that is going to appear on the sphere $S^2_{u_f,\infty}$ when taking the limit $v_f \rightarrow \infty$ of (\ref{mainT1}) and inserting the above expressions. 
 
We next evaluate the non-coercive cross-terms in (\ref{Tfluxng}). Using Proposition~\ref{prop:TTransport} in the $\Sf'$ gauge we have (all quantities appearing on the right being with respect to the solution $\Sf^\prime$)
\begin{align}
2T\otx_{\Sf'}&=2\Omega^2(\divs\elin+\rlin)-\frac{\Omega^2}{r}\otxb, \label{infutg1}  \\
2T\otxb_{\Sf'}&=2\Omega^2(\divs\eblin+\rlin)+\frac{1}{r}\Big(1-\frac{4M}{r}\Big)\otxb-\frac{4\Omega^2}{r}\olinb,\\
2T\olin_{\Sf'}&=-\Omega^2\rlin,\\
2T\big(\Olin\big)_{\Sf'}&=\olinb. \label{infutg2}
\end{align}
We therefore have
\begin{align}
&\Bigg| \int_{v_0}^{v_f} \int_{S^2} \Big[  -2 \ns_{T}\olin_{\Sf'} \ns_{T}\otxb_{\Sf'} - \frac{1}{2} [\ns_{T}\otx_{\Sf'}]^2 + \frac{4M}{r^2} \ns_{T}\Big(\frac{\Olino}{\Omega}\Big)_{\Sf'} \ns_{T}\otx_{\Sf'}  \Big]r^2dv\varepsilon_{S^2} \Bigg| \nonumber \\
&\lesssim  \int_{v_0}^{v_f} \int_{S^2} \frac{\Omega^2}{r^2} \Big[|\olinb|^2 + |r \otxb|^2 + r^4|\divs  \elin + \rlin|^2 + r^2  |\divs  \eblin + \rlin|^2\Big]dv\varepsilon_{S^2} +\int_{v_0}^{v_f} \int_{S^2}\Omega^2r^4|\rlin|^2dv\varepsilon_{S^2}
\nonumber \\ 
&\lesssim  \overline{\mathbb{E}}^{0,0}_{data} (u_f)+ {\mathbb{D}}_{data}^1(u_f,v_0) \, , 
\end{align}
with the last step following from Proposition~\ref{prop:ombpure}, \ref{prop:chibarpure} and \ref{prop:angularetab} as well as the fluxes controlled by (\ref{step1mb}).

We conclude from (\ref{mainT1}) and~\eqref{eqn:dataG10} after inserting the previous estimates and  taking the limit $v_f \rightarrow \infty$ 
\begin{align*}
 &\int_{u_0}^{u_f} \int_{S^2} |\ns_{T}\Omega\xblin_{\Si'}|^2(u,\infty) r^2du\varepsilon_{S^2}+\int_{v_0}^{\infty}\int_{S^2} \Big[  |\ns_{T}\Omega\xlin_{\Sf'}|^2+2\Omega^2 |\ns_{T}\eblin_{\Sf'}|^2 \Big]r^2dv\varepsilon_{S^2}\\
 &\leq F^{1,\bf{0}}_{v_0} [\Gamma,\Si' ](u_0,u_f)+F^{1,\bf{0}}_{u_0} [\Gamma,\Si' ](v_0,\infty) +\frac{r^3}{2}\int_{S^2}T\otx_{\Si^\prime}T\otxb_{\Si^\prime}\varepsilon_{S^2}\Big|_{(u_0,\infty)} \nonumber \\
 & \ \ \ +C \left(\overline{\mathbb{E}}^{0,0}_{data} (u_f)+ {\mathbb{D}}_{data}^1(u_f,v_0)\right) \, .
\end{align*}
The estimate (\ref{step5m}) now follows after recalling the definition (\ref{Ei0datadef}).

\subsubsection{Step 6: Closing the Once $T$-commuted $\mathcal{R}$-estimate}
We now prove (\ref{step6m}). We apply the second conservation law now once $\mathcal{L}_T$-commuted:
\begin{align*}
F^{1,\bf{0}}_{v_f} [\mathcal{R},\Si' ](u_0,u_f)+F^{1,\bf{0}}_{u_f} [\mathcal{R},\Si' ](v_0,v_f)=F^{1,\bf{0}}_{v_0} [\mathcal{R},\Si' ](u_0,u_f)+F^{1,\bf{0}}_{u_0} [\mathcal{R},\Si' ](v_0,v_f) \, .
\end{align*}
By the change of gauge formula on the cone $C_{u_f}$ of Proposition~\ref{prop:gaugechange} we have
\begin{align}
F^{1, {\bf 0}}_{u_f} [\mathcal{R}, \Si'] \left(v_0,v_f\right)&= F^{1, {\bf 0}}_{u_f} [\mathcal{R}, \Sf'] \left(v_0,v_f\right)+\left[\int_{S^2}\mathcal{G}^{1,0}_{\mathcal{R}}(u,v,\theta,\phi)\varepsilon_{S^2}\right]_{(u_f,v_0)}^{(u_f,v_f)} \, ,
\end{align}
with the new flux given by 
\begin{align} \label{newfluxT2}
F^{1, {\bf 0}}_{u_f} [\mathcal{R}, \Sf'] \left(v_0,v\right)&=\int_{v_0}^{v_f} \int_{S^2} \Big[ \frac{1}{2} r^4 |\ns_T \Omega\blin|^2+\frac{1}{2} \Omega^2 r^4 \left(|T\rlin|^2 + |T\slin|^2\right) -3Mr |\ns_T\Omega\eblin|^2  \\
&\qquad+ 3Mr T\olin T\otxb 
+3M \left(1-\frac{4M}{r}\right)  T\otx T\Big(\frac{\Olino}{\Omega}\Big) \Big]dv\varepsilon_{S^2} \nonumber \, .
\end{align}
Using (\ref{decreten2}) we also have the limiting flux 
\begin{align*}
\lim_{v_f \rightarrow \infty} F^{1, {\bf 0}}_{v_f} \left[\mathcal{R},\Si' \right] \left(u_0,u_f\right)&=  \int_{u_0}^{u_f} \int_{S^2} \frac{1}{2} r^4  |\ns_T\Omega\bblin|^2 (u,\infty) du\varepsilon_{S^2} \, .
\end{align*}
We estimate (plugging in the expressions (\ref{infutg1})--(\ref{infutg2}) in the integrand) the non-coercive terms in (\ref{newfluxT2}):
\begin{align}
&\Bigg| \int_{v_0}^{v_f}\int_{S^2}\Big[  -3Mr |\ns_T\Omega\eblin|^2 + 3Mr T\olin T\otxb 
+3M \left(1-\frac{4M}{r}\right)  T\otx T\Big(\frac{\Olino}{\Omega}\Big)  \Big]dv\varepsilon_{S^2} \Bigg| \nonumber \\
& \qquad \lesssim\int_{v_0}^{v_f}\int_{S^2}\frac{\Omega^2}{r^2}\Big[r^2|\divs\eblin+\rlin|^2+r^4|\divs\elin+\rlin|^2+|\olinb|^2+(r\otxb)^2\Big]dv\varepsilon_{S^2}+\int_{v_0}^{v_f}\int_{S^2}r|\ns_T\eblin|^2dv\varepsilon_{S^2} \nonumber \\
& \qquad \lesssim\overline{\mathbb{E}}^{0,0}_{data} (u_f)+ {\mathbb{D}}_{data}^1(u_f,v_0)\, ,
\end{align}
where we use~\eqref{step5m} to estimate the $\ns_T\eblin$-flux. In summary, using also the estimates for $\mathcal{G}^{1,0}_{\mathcal{R}}$ from (\ref{eqn:dataGR10}) and (\ref{RGtime}) (the latter in conjunction with Proposition~\ref{prop:ombpure}) we deduce
\begin{align*}
 &\int_{u_0}^{u_f} \int_{S^2} \frac{1}{2} r^4  |\ns_T\Omega\bblin|^2(u,\infty) du\varepsilon_{S^2}+\int_{v_0}^{\infty}\int_{S^2} \Bigg[ \frac{1}{2} r^4 |\ns_T \Omega\bblin|^2+\frac{1}{2} \Omega^2 r^4 \left(|T\rlin|^2 + |T\slin|^2\right)\Bigg] (u_f,v) dv\varepsilon_{S^2}\\
 &\leq F^{1,\bf{0}}_{v_0} [\mathcal{R},\Si' ](u_0,u_f)+F^{1,\bf{0}}_{u_0} [\mathcal{R},\Si' ](v_0,v_f) +C \left(\overline{\mathbb{E}}^{0,0}_{data} (u_f)+ {\mathbb{D}}_{data}^1(u_f,v_0)  \right)\\
 &\leq E_{data}^{1,0}[\mathcal{R}](u_f)+C \left(\overline{\mathbb{E}}^{0,0}_{data} (u_f)+ {\mathbb{D}}_{data}^1(u_f,v_0)  \right) \, .
\end{align*}
The estimate (\ref{step6m}) follows after recalling the definition (\ref{Ecd}).

\subsubsection{Step 7: Refined Bounds on $C_{u_f}$ via Transport Equations II}
With  (\ref{step6m}) proven we now obtain control for the flux of angular derivatives of $\bblin$:

\begin{proposition} \label{prop:angularbetabflux}
We have the estimates
\begin{align*}
\int_{v_0}^{\infty}\int_{S^2}|[r\ns]^2\xlin_{\Sf'}|^2r^2\Omega^2\Big |_{u=u_f} dv\varepsilon_{S^2} \lesssim
 \overline{\mathbb{E}}^{1,0}_{data} (u_f)  \, , 
\end{align*}
\[
\int_{v_0}^\infty \int_{S^2} \Omega^2 r^2 |[r\slashed{\nabla}]\bblin_{\Sf^\prime}|^2 \Big |_{u=u_f}  dv \varepsilon_{S^2} 
\lesssim
 \overline{\mathbb{E}}^{1,0}_{data} (u_f) \, .
\]
\end{proposition}

\begin{proof}
For the bound on $\xlin$ we note that the Codazzi equation implies
\begin{align} \label{laplacexlin}
\ds\Omega\xlin=\frac{2\Omega^2}{r}\slashed{\mathcal{D}}_2^{\star}\eblin-\slashed{\mathcal{D}}_2^{\star}\ns\otx+\frac{2}{r^2}\Omega\xlin+2\slashed{\mathcal{D}}_2^{\star}\Omega\blin.
\end{align}
By standard elliptic estimates the Laplacian controls all second derivatives up to lower order terms which are in turn controlled by Proposition~\ref{prop:angularetab}.
The right hand side of (\ref{laplacexlin}) is easily controlled since $\otx$ vanishes on $C_{u_f}$ in $\Sf'$-gauge and the remaining terms are controlled 
by Proposition~\ref{prop:angularetab} and (\ref{step2m}).

For the bound on $\bblin$ we compute
\begin{align}
 \divs \Omega \bblin &= - \Omega \slashed{\nabla}_3 \rlin +\frac{3M}{r^3} \otxb +\frac{3\Omega^2}{r} \rlin = - 2T \rlin + \divs  \Omega \blin  - \frac{3M}{r^3} \otx - \frac{3M}{r^3} \otxb \, ,  \nonumber \\
  \curls \Omega \bblin  &=- \Omega \slashed{\nabla}_3 \slin +\frac{3\Omega^2}{r} \slin = - 2T \slin - \curls \Omega\blin \, \, , 
 \end{align}
to conclude (note $\otx$ vanishes along $C_{u_f}$ in the $\Sf'$-gauge)
 \begin{align}
 | \slashed{\nabla} \Omega \bblin_{\Sf^\prime}|^2  \lesssim |T\rlin_{\Sf^\prime}|^2 + |T\slin_{\Sf^\prime}|^2 +  | \slashed{\nabla} \Omega \blin_{\Sf^\prime}|^2 + \frac{1}{r^6} |\otxb_{\Sf^\prime}|^2 \, .
 \end{align}
 The result follows after multiplying with $r^4 \Omega^2$ and integrating using (\ref{step6m}), (\ref{step2m}) and the bound (\ref{trxbprop}).
\end{proof}
 
We next prove a higher order version of Proposition~\ref{prop:ombpure}:
\begin{proposition} \label{prop:higheromegab}
We have $\olinb_{\Si^\prime}=\olinb_{\Sf^\prime}$ on the sphere $S^2_{u_f,\infty}$. Moreover, we have
\begin{align} \label{ombprop2}
\sup_{v \in [v_0,\infty)}  \|[r\slashed{\nabla}] \olinb_{\Sf^\prime}\|^2_{u_f,v} 
\lesssim  \overline{\mathbb{E}}^{0,1}_{data} (u_f)+\mathbb{D}^1_{data}(u_f,v_0) \, ,
\end{align}
and
\begin{align} \label{ombprop3}
||\partial_u\olinb_{\Sf'}||^2_{u_f,v}\lesssim \overline{\mathbb{E}}^{1,0}_{data} (u_f) + \mathbb{D}^2_{data}(u_f,v_0) \, .
\end{align}
We also have the top order bound
\begin{align} \label{ombprop4}
\sup_{v \in [v_0,\infty)}  \|\Delta_{S^2} \olinb_{\Sf^\prime}-r^2 \divs  \Omega\bblin_{\Sf^\prime}\|^2_{u_f,v} 
\lesssim \overline{\mathbb{E}}^{1,0}_{data} (u_f) + \mathbb{D}^2_{data}(u_f,v_0)\, .
\end{align}
\end{proposition}

\begin{proof}
The estimate (\ref{ombprop2}) is proven as the estimate (\ref{ombprop}) in Proposition~\ref{prop:ombpure} with one angular commutation which we leave to the reader.

For \eqref{ombprop3}, one can compute from Bianchi and commuting partials that in the $\Sf'$ gauge along $C_{u_f}$
\begin{align*}
\partial_v(\partial_u\olinb)=\Omega^2\Big[\Big(\frac{4M}{r^2}-\frac{3\Omega^2}{r}\Big)\rlin+\divs(\Omega\bblin)-\frac{3M}{r^3}\otxb+\frac{4M}{r^3}\olinb\Big].
\end{align*}
From Lemma~\ref{lem:tplemma} we deduce the estimate
\begin{align*}
||\partial_u\olinb_{\Sf'}||_{u_f,v}\lesssim ||\partial_u\olinb_{\Sf'}||_{u_f,v_0}+\sqrt{\int_{v_0}^{v}\Omega^2\left( \| \rlin \|_{u_f,\bar{v}}^2 + \|r\divs(\Omega\bblin)\|_{u_f,\bar{v}}^2 + \frac{1}{r^4} \Big[\| \otxb\|^2_{u_f,\bar{v}} +\| \olinb\|^2_{u_f,\bar{v}}\Big] \right) d\bar{v}}.
\end{align*}
The flux is controlled using (\ref{step1mb}), Proposition~\ref{prop:angularbetabflux}, \eqref{ombprop} and (\ref{trxbprop}) while for the data term we can use
\begin{align*}
\partial_u\olinb_{\Sf'}(u_f,v_0)=\partial_u\olinb_{\Si'}(u_f,v_0)-\frac{3M}{r^3}\Omega^2f_0+\frac{2M^2}{r^4}f_0 \, ,
\end{align*}
whose right-hand side is easily controlled by $\mathbb{D}^2_{data}(u_f,v_0)$.

Turning to (\ref{ombprop4}) we note that 
from commuting the transport equation for $\olinb$ we can derive
\begin{align}
\Omega \slashed{\nabla}_4 (r^2 \slashed{\Delta} \olinb) = - \Omega^2 r^2 \slashed{\Delta} \rlin + \frac{4M}{r^3}\Omega^2 \Delta_{S^2} \Olin \, .
\end{align}
Renormalising with the (commuted) Bianchi equation~\eqref{Bianchi8} we deduce
\begin{align}
\Omega \slashed{\nabla}_4 (r^2 \slashed{\Delta} \olinb- \frac{1}{r}r^3\divs  \Omega \bblin) = \Omega^2 r\divs  \Omega \bblin  -\frac{6M}{r} \Omega^2 \divs  \eblin + \frac{4M}{r^3}\Omega^2 \Delta_{S^2} \Olin  \nonumber \\
= - 2\Omega^2 r T \rlin + \Omega^2 r \divs  \Omega \blin  - \frac{3M}{r^3}\Omega^2 r \otx -\frac{3M}{r^3}\Omega^2 r \otxb  -\frac{6M}{r} \Omega^2 \divs  \eblin + \frac{4M}{r^3}\Omega^2 \Delta_{S^2} \Olin \, , 
\end{align}
where the second equality follows from inserting the Bianchi equation~\eqref{Bianchi5}. Restricting now to the solution in the gauge $\Sf^\prime$, we can integrate the right hand side in $v$ along $C_{u_f}$ using Cauchy--Schwarz and the fluxes already controlled (recall also that $\otx_{\Sf^\prime}=0$ along $C_{u_f}$) to obtain (\ref{ombprop4}) but with the term $\|\Delta_{S^2} \olinb_{\Sf^\prime}-r^2 \divs  \Omega\bblin_{\Sf^\prime}\|^2_{u_f,v_0}$ on the right-hand side. From Lemma~\ref{lem:exactsol} we have the relation
\begin{align*}
r^2 \slashed{\Delta} \olinb_{\Sf^\prime} - r^2 \divs  \Omega \bblin_{\Sf^\prime} = r^2 \slashed{\Delta} \olinb_{\Si^\prime} - r^2 \divs  \Omega \bblin_{\Si^\prime} + \frac{4M}{r} \Omega^2 \slashed{\Delta} f  \, , 
\end{align*}
whose right-hand side is easily controlled by $\mathbb{D}^2_{data}(u_f,v_0)$.
\end{proof}

\begin{corollary} \label{cor:omegabtop}
We have the following flux bound along $C_{u_f}$:
\begin{align} \label{topombflux}
\int_{v_0}^{\infty} \int_{S^2} \frac{\Omega^2}{r^2} \Big| \Delta_{S^2} \olinb _{\Sf^\prime} \Big|^2  \Big|_{u=u_f}dv  \varepsilon_{S^2}   \leq\overline{\mathbb{E}}^{1,0}_{data} (u_f)+ \mathbb{D}_{data}^2(u_f,v_0) \, .
\end{align}
\end{corollary}

\begin{remark}
Note that (\ref{topombflux}) controls all derivatives up to order two by elliptic estimates and the fact that $\olinb$ is supported on $\ell \geq 2$.
\end{remark}

\begin{proof}
The estimate follows from the bound (\ref{ombprop4}) on the spheres $S^2_{u_f,v}$ and the fact that the flux of $\Omega^2 |\divs \Omega \bblin|^2$ can be controlled by (\ref{step6m}), (\ref{step2m}) and (\ref{trxbprop}) after plugging these estimates into the identity
\[
r^2\divs  \Omega \bblin = - 2 r^2 T \rlin + r^2 \divs  \Omega \blin  - \frac{3M}{r^3} r^2 \otx +\frac{3M}{r^3} r^2 \otxb \, .
\]
\end{proof}

\subsubsection{Step 8: Closing the Twice $T$-commuted $\Gamma$-estimate}
We finally prove~\eqref{step8m}. To do this we finally apply the first conservation law with two $\mathcal{L}_T$-commutations:
\begin{align}\label{eqn:F20conslaw}
F^{2,\bf{0}}_{v_f} [\Gamma,\Si' ](u_0,u_f)+F^{2,\bf{0}}_{u_f} [\Gamma,\Si' ](v_0,v_f)=F^{2,\bf{0}}_{v_0} [\Gamma,\Si' ](u_0,u_f)+F^{2,\bf{0}}_{u_0} [\Gamma,\Si' ](v_0,v_f) \, .
\end{align}
By the gauge change formula on $C_{u_f}$ in Proposition~\ref{prop:gaugechange} and~\eqref{GammaGonetime} from Proposition~\ref{prop:Gatinfinity}, we have
\begin{align} \label{g2id}
F^{2, {\bf 0}}_{u_f} [\Gamma, \Si'] \left(v_0,v\right)&= F^{2, {\bf 0}}_{u_f} [\Gamma, \Sf'] \left(v_0,v\right)-\int_{S^2}\mathcal{G}^{2,0}_{\Gamma}(u,v,\theta,\phi)\varepsilon_{S^2}\Big|_{(u_f,v_0)} \nonumber \\
&\qquad- \frac{r^3}{2\Omega^{2}} \int_{S^2} T^2\otx_{\Si^\prime}T^2\otxb_{\Si^\prime} \varepsilon_{S^2}\Big|_{(u_f,v_f)}+\mathrm{R}(v_f),
\end{align}
with the new flux given by
\begin{align}\label{eqn:T2fluxng}
F^{2, {\bf 0}}_{u_f} [\Gamma, \Sf'] \left(v_0,v_f\right)&=\int_{v_0}^{v_f} \int_{S^2}\Bigg[  |[\ns_{T}]^2\Omega\xlin_{\Sf'}|^2+2\Omega^2 |[\ns_{T}]^2\eblin_{\Sf'}|^2 -2 [\ns_{T}]^2\olin_{\Sf'} [\ns_{T}]^2\otxb_{\Sf'} \nonumber \\
&\qquad- \frac{1}{2} |[\ns_{T}]^2\otx_{\Sf'}|^2+ \frac{4M}{r^2} [\ns_{T}]^2\Big(\frac{\Olino}{\Omega}\Big)_{\Sf'} [\ns_{T}]^2\otx_{\Sf'}  \Bigg]r^2du\varepsilon_{S^2}.
\end{align}
For the term $\mathcal{G}^{2,0}_{\Gamma}$ appearing in (\ref{g2id}) we can apply the estimate (\ref{eqn:dataG20}). 
For the ingoing cone $\underline{C}_{v_f}$ we have from~\eqref{niflux} the limiting expression
\begin{align*}
F^{2, {\bf 0}}_{v_f} \left[\Gamma,\Si' \right] \left(u_0,u_f\right)&=  \int_{u_0}^{u_f}\int_{S^2}  |[\ns_{T}]^2\Omega\xblin_{\Si'}|^2r^2du\varepsilon_{S^2}+\left[\frac{r^3}{2}\int_{S^2}T^2\otx_{\Si^\prime}T^2\otxb_{\Si^\prime}\varepsilon_{S^2}\right]_{(u_0,\infty)}^{(u_f,\infty)}+\mathrm{R}(v_f).
\end{align*}
Note the cancellation that is going to appear on the sphere $S^2_{u_f,\infty}$ when taking the limit $v_f\rightarrow \infty$ of~\eqref{eqn:F20conslaw} and inserting the above expressions.

We next evaluate the non-coercive cross-terms in~\eqref{eqn:T2fluxng}. From Proposition~\ref{prop:TTransport} we have, using the properties of the $\Sf'$ gauge (see Proposition~\ref{prop:propnewgauge}), the following identities along $C_{u_f}$:
\begin{align*}
4T^2\otx&=2\Omega^2\rho\olinb-\frac{2\Omega^2}{r}2\Omega^2\rlin+\frac{\Omega^2}{r}\frac{2\Omega^2}{r}\olinb+\frac{2\Omega^2M}{r^2}(\divs\elin+\rlin)-\Omega^2\Big[2\Omega\divs\bblin+\frac{3}{2}\rho\otxb\Big]+2\Omega^2\ds\olinb, \\
4T^2\otxb&= -\frac{2\Omega^2}{r}\Big[\partial_u\olinb-2\Omega^2\rlin\Big]-\frac{\Omega^2}{r}\Big[\frac{2M}{r^2}\otxb+\frac{2\Omega^2}{r}\olinb\Big]\\
&\quad-\frac{M\Omega^2}{r^2}\Big[2(\divs\eblin+\rlin)-\frac{2M}{r^2}\frac{\otxb}{\Omega^2}\Big]+\Omega^2\Big[2\Omega\divs\blin-\frac{3}{2}\rho\olinb\Big],\\
4T^2\olin&=-2\rho\Omega^2\olinb-2\Omega^2\omega\rlin-\Omega^2\Big[2\divs\Omega\blin-\divs\Omega\bblin-\frac{3\Omega^2}{r}\rlin-\frac{3M}{r^3}\otxb\Big] \, , \\
4T^2\big(\Olin\big)&=-2\Omega^2\rlin+\partial_u\olinb \, .
\end{align*}
We therefore have (all integrals being along $C_{u_f}$ and all quantities with respect to the solution $\Sf^\prime$)
\begin{align}
&\Bigg| \int_{v_0}^{v_f} \int_{S^2}\Big[  -2 \ns_{T}^2\olin \ns_{T}^2\otxb - \frac{1}{2} [\ns_{T}^2\otx]^2 + \frac{4M}{r^2} \ns_{T}^2\Big(\frac{\Olino}{\Omega}\Big) \ns_{T}^2\otx  \Big]r^2dv\varepsilon_{S^2} \Bigg| \nonumber \\
&\lesssim  \int_{v_0}^{v_f}\int_{S^2} \frac{\Omega^2}{r^2}\Big[|\olinb|^2+|\partial_u\olinb|^2+|\divs\elin+\rlin|^2+|\divs\eblin+\rlin|^2+|\Delta_{S^2}\olinb|^2+|r\otxb|^2\Big]dv\varepsilon_{S^2} \\
&\nonumber\qquad+\int_{v_0}^{v_f}\int_{S^2}\Omega^2\Big[|r\rlin|^2+|r^2\divs(\Omega\blin)|^2+|r^2\divs(\Omega\bblin)|^2\Big] dv\varepsilon_{S^2} \lesssim \overline{\mathbb{E}}^{1,0}_{data}(u_f)+ {\mathbb{D}}_{data}^2(u_f,v_0) \, , 
\end{align}
where the last inequality follows from Corollary~\ref{cor:omegabtop}, Propositions~\ref{prop:angularetab},~\ref{prop:ombpure},~\ref{prop:chibarpure},~\ref{prop:angularbetabflux},~\ref{prop:higheromegab}, and the estimate in equation~\eqref{step2m}.

We finally conclude from~\eqref{eqn:F20conslaw} after taking the limit $v_f\rightarrow \infty$ and inserting previous estimates
\begin{align*}
 &\int_{u_0}^{u_f} \int_{S^2} |[\ns_{T}]^2\Omega\xblin_{\Si'}|^2r^2dv\varepsilon_{S^2}+\int_{v_0}^{\infty}\int_{S^2} \Big[  |[\ns_{T}]^2\Omega\xlin_{\Sf'}|^2+2\Omega^2 |[\ns_{T}]^2\eblin_{\Sf'}|^2 \Big]r^2du\varepsilon_{S^2}\\
 &\lesssim\frac{r^3}{2}\int_{S^2}T^2\otx_{\Si^\prime}T^2\otxb_{\Si^\prime}\varepsilon_{S^2}\Big|_{(u_0,\infty)}+F^{2,\bf{0}}_{v_0} [\Gamma,\Si' ](u_0,u_f)+F^{2,\bf{0}}_{u_0} [\Gamma,\Si' ](v_0,\infty)\\
 & +C\left( \overline{\mathbb{E}}^{1,0}_{data}(u_f)+ {\mathbb{D}}_{data}^2(u_f,v_0)\right)
= \mathbb{E}^{2,0}_{data}[\Gamma](u_f) +C\left( \overline{\mathbb{E}}^{1,0}_{data}(u_f)+ {\mathbb{D}}_{data}^2(u_f,v_0)\right) \, ,
\end{align*}
where we have inserted the definition (\ref{Ei0datadef}) in the last step. This is \eqref{step8m}.

\subsection{Proof of Theorem \ref{theo:gaugeinvariantoutgoing}} \label{sec:gaugeinvariantoutgoing}
In this section we prove (\ref{finalalpha}) and (\ref{finalalphab}). Note that it suffices to prove (\ref{finalalpha}) and (\ref{finalalphab}) without the zeroth order term as the latter is controlled directly from the angular term by an elliptic estimate. 

Before we begin, we first use the fluxes already controlled through (\ref{secondorderest}) to prove another flux estimate:

\begin{proposition} \label{prop:betabflux}
We have the following estimate for $\bblin_{\Sf'}$ along $C_{u_f}$:
\begin{align}
\int_{v_0}^{\infty}\int_{S^2} r^2\Omega^2| \ns_T (\Omega \bblin)_{\Sf^\prime}|^2dv\varepsilon_{S^2}\lesssim {\mathbb{E}}^{2}_{data}(u_f).
\end{align}
\end{proposition}
\begin{proof}
Solving (\ref{id2}) for $\bblin$ and applying $\slashed{\nabla}_T$ with~\eqref{id2b} from Proposition~\ref{prop:TTransport} yields
\begin{align} \label{id4}
\slashed{\nabla}_T \Omega \bblin = 2 \slashed{\nabla}_T\slashed{\nabla}_T\eblin - \slashed{\nabla}_T (\Omega \blin)+ \frac{\Omega^2}{r} \left( \slashed{\nabla} \olin+  \slashed{\nabla} \olinb\right) -  \slashed{\nabla} \left(\Omega \slashed{\nabla}_4 \olin\right) + \Omega^2 \slashed{\nabla}  \rlin -\frac{2M \Omega^2 }{r^3} \left(\elin + \eblin\right) \, .
\end{align}
Since $\olin_{\Sf^\prime}=0$, $(\elin+ \eblin)_{\Sf^\prime}=0$ along $C_{u_f}$ by Proposition \ref{prop:propnewgauge}, we have from (\ref{id4}) in $\Sf'$-gauge along $C_{u_f}$
\begin{align}
|\slashed{\nabla}_T (\Omega \bblin)_{\Sf^\prime}|^2 
 \lesssim | \slashed{\nabla}_T (\Omega \blin)_{\Sf^\prime}|^2 + |\slashed{\nabla}_T \slashed{\nabla}_T \eblin_{\Sf^\prime}|^2  +  \frac{\Omega^4}{r^4} | [r\slashed{\nabla}] \olinb_{\Sf^\prime}|^2 +\frac{\Omega^4}{r^2}|[r \slashed{\nabla} ]\rlin_{\Sf^\prime}|^2  \, .
\end{align}
Integrating with respect to the measure $\Omega^2 r^2 dv \varepsilon_{S^2}$, the result follows from (\ref{secondorderest}) and Proposition \ref{prop:higheromegab}.
\end{proof}

In the following, all quantities are with respect to the solution $\Sf^\prime$. Recall also that $\ablin_{\Sf^\prime} = \ablin_{\Si^\prime}$. 

We first prove (\ref{finalalphab}). To control the $\ns\Ab$-part of the flux, we write (\ref{Bianchi9}) using (\ref{Bianchi8}) as
 \begin{align*}
 \divs\Ab=2\ns_T(r\Omega\bblin)-\Omega^2[r\Dso](\rlin,\slin)-\frac{6M\Omega^2}{r^2}\eblin-2\Omega^2(\Omega\bblin),
 \end{align*}
whence
\begin{align}
|\divs  \Ab|^2 &\leq 16\Big( |\slashed{\nabla}_T (r\Omega \bblin)|^2 + \Omega^2| [r\ns] \rlin|^2 +  \Omega^2| [r\ns]\slin|^2 + \frac{6M\Omega^4}{r^4} |\eblin|^2 + | \Omega \bblin|^2\Big) \, .
\end{align}
Integrating this along the cone $C_{u_f}$ with respect to the measure $r^2 \Omega^2 dv\varepsilon_{S^2}$ gives the desired control after using Proposition \ref{prop:betabflux} and (\ref{secondorderest}). For $\Omega \slashed{\nabla}_4 \Ab$-part of the flux, we write (\ref{Bianchi10}) as
\begin{align*}
\Omega \slashed{\nabla}_4\Ab&=  2 \Omega^2 r\slashed{\mathcal{D}}_2^\star \Omega \bblin + \frac{6M\Omega^2}{r^2} \Omega \xblin.
\end{align*}
Therefore,
\begin{align*}
|\Omega\slashed{\nabla}_4\Ab|^2\lesssim \Omega^4|[r\ns]\Omega\bblin|^2+\frac{\Omega^4}{r^4}|\Omega\xblin|^2.
\end{align*}
Integrating this along the cone $C_{u_f}$ with respect to the measure $r^2dv\varepsilon_{S^2}$ we obtain the desired control after using Proposition~\ref{prop:angularbetabflux} and  (\ref{trxbpropb}).

We now prove (\ref{finalalpha}). For the angular part of the flux on the left, we write~\eqref{Bianchi2} and~\eqref{Bianchi3} as
\begin{align*}
\divs A=2\ns_T(r\Omega\blin)+2\Big(1-\frac{4M}{r}\Big)\Omega\blin-\Omega^2[r\Dso](-\rlin,\slin)+\frac{6M\Omega^2}{r^2}\eta.
\end{align*}
 Therefore, we can estimate
\begin{align}
| \divs A|^2& \lesssim    | r  \slashed{\nabla}_T (\Omega \blin)|^2 + |\Omega \blin|^2 + \Omega^4 | r \slashed{\nabla} \rlin|^2 + \Omega^4 |r\slashed{\nabla} \slin|^2 + \frac{M^2}{r^4}\Omega^4 |\elin|^2 \, . 
\end{align}
Integrating over the cone  $C_{u_f}$ with respect to the measure $r^2 dv \varepsilon_{S^2}$ 
the result follows from estimating the terms on the right hand side by (\ref{secondorderest}) and 
using the elliptic divergence identity $| \slashed{\nabla} A|^2 +\frac{2}{r^2}|A|^2 =2|\divs A|^2+\divs(\cdot)$. In particular, we also control $A$ itself.

We gain control of $r^2\Omega \slashed{\nabla}_3 A$ using the Bianchi equation~\eqref{Bianchi1},
\begin{align*}
r^2\Omega\ns_3A=2r^2\Omega^2[r\Dst](\Omega\blin)-6M\Omega^2\Omega\xlin.
\end{align*}
This lets us estimate $\Omega \slashed{\nabla}_4 A$ using that $\Omega\ns_4A=2\ns_TA-\Omega\ns_3A$. So,
\begin{align} 
r^4| \Omega \slashed{\nabla}_3 A|^2&\lesssim r^4|r \Omega^2 \slashed{\mathcal{D}}_2^\star (\Omega \blin)|^2 +  \Omega^4|\Omega \xlin|^2,\label{ns3Acontrol}\\
|\Omega \slashed{\nabla}_4 A|^2 &\lesssim | \ns_T A|^2 + | \Omega \slashed{\nabla}_3 A|^2 \, .\label{auxo}
\end{align}
We now have control on $r^4| \Omega \slashed{\nabla}_3 A|^2$ since we control the right hand side of~\eqref{ns3Acontrol} via~\eqref{step2m} and~\eqref{basicuc}. Additionally, to control $|\Omega \slashed{\nabla}_4 A|^2$ we now need to estimate $| \ns_T A |^2$ and use~\eqref{auxo}. To estimate the $| \ns_T A |^2$-term we use (\ref{tchi}) and $\Omega\ns_4=2\ns_T-\Omega\ns_3$
\begin{align*}
-A&=2\ns_T(r\Omega\xlin)+2\Omega^2 r\Dst{ \elin}+\Omega^2\Omega\xblin -2\omega r\Omega\xlin+\Omega^2\Omega\xlin \, .
\end{align*}
We apply $\slashed{\nabla}_T$ to deduce an estimate for $| \ns_T A |^2$. This yields
\begin{align}
| \slashed{\nabla}_T A|^2 
&\lesssim  | \ns_T \ns_T (r\xlin \Omega) |^2 +\Omega^4  | \ns_T [r\slashed{\nabla}] {\elin} |^2 + \Omega^4 |\ns_T (\Omega \xblin)|^2 + \frac{1}{r^2} |\slashed{\nabla}_T (\Omega r \xlin) |^2  .
\end{align}
It is now manifest that integrating over the cone  $C_{u_f}$ with respect to the measure $ dv \varepsilon_{S^2}$, one can control all terms on the right hand side by (\ref{secondorderest}), except  the terms $\Omega^4 |\ns_T (\Omega \xblin)|^2$ and $\Omega^4  | \ns_T [r\slashed{\nabla}] {\elin} |^2$.

To deal with the term $|\ns_T (\Omega r\xblin)|^2$, we combine (\ref{tchi}) and (\ref{chih3b}) to write 
\begin{align*}
\ns_T (\Omega\xblin)=-\frac{1}{2r}\Ab-\frac{M}{r^2}(\Omega\xblin)+\frac{\Omega^2}{2r}\Big[\Omega\xblin+\Omega\xlin\Big]-\Omega^2\Dst\eblin.
\end{align*}
It follows that
\begin{align}
{\Omega^4}|\ns_T (\Omega \xblin)|^2 \lesssim  \frac{\Omega^4}{r^2} |\Ab|^2 + \frac{\Omega^4}{r^4} |r \Omega \xblin|^2 + \frac{\Omega^8}{r^2} |\Omega \xlin|^2 + \frac{\Omega^8}{r^2} |r \slashed{\nabla} \eblin|^2 \, .
\end{align}
After integrating with respect to the measure $ dv \varepsilon_{S^2}$ over $C_{u_f}$ one controls all terms on the right through  (\ref{finalalphab}), (\ref{trxbpropb}) and Proposition~\ref{prop:angularetab} respectively.

For the term $\Omega^4  | \ns_T [r\slashed{\nabla}] {\elin} |^2$ we recall (\ref{id1}) to estimate 
\begin{align}
\Omega^4  | \ns_T [r\slashed{\nabla}] {\elin} |^2 \lesssim \Omega^4 |r \slashed{\nabla} \blin|^2 + \Omega^4 |r\slashed{\nabla} \bblin|^2+ \frac{\Omega^4}{r^2} |r \slashed{\nabla} \elin|^2 +  \frac{\Omega^4}{r^2} |r \slashed{\nabla} \eblin|^2 + \frac{\Omega^4}{r^2}  |r^2 \slashed{\nabla}\slashed{\nabla} \olinb|^2 \, .
\end{align}
After integrating with respect to the measure $ dv \varepsilon_{S^2}$ over $C_{u_f}$ one controls all terms on the right through  (\ref{secondorderest}), Proposition~\ref{prop:angularbetabflux}, Proposition~\ref{prop:angularetab} and Corollary~\ref{cor:omegabtop}. This completes the proof of (\ref{finalalpha}).

\section{Controlling the Teukolsky Fluxes along Ingoing Cones} \label{sec:ingoing}
In this final section we obtain bounds analogous to the ones appearing in Theorem \ref{theo:gaugeinvariantoutgoing} but on the \emph{ingoing} cones. A natural approach would be to repeat the arguments of Section \ref{sec:outgoingcone} now changing the gauge on the ingoing cone to satisfy $\otxb_{\Sf^\prime}=0$. However, this argument does not seem to close as the change of gauge now produces a term of the wrong sign as opposed to the good term on null infinity (the first term in (\ref{finge})) generated in the outgoing argument. We remedy this by obtaining the ingoing flux through a standard energy estimate for the Teukolsky equation itself. The main idea is that the spacetime term produced by the first order term in the Teukolsky equation (which a priori could lead to exponential growth) can grow at most linearly in time by the estimate we already proved on the outgoing flux. 

\subsection{The Teukolsky Equations for $A$ and $\underline{A}$}
We recall the form of the Teukolsky equations for the quantities $A=r\Omega^2\alin$ and $\underline{A}=r\Omega^2\ablin$.\footnote{Of course they can also be derived directly from the equations of Section \ref{sec:lineqs}.}
\begin{align} \label{TeuA}
 \Omega\ns_4(\Omega\ns_3A)-\Omega^2\slashed{\Delta}A+\frac{4(r-3M)}{r^2}\Omega\ns_{3}A+\frac{2\Omega^2(r+3M)}{r^3}A = 0
\end{align}
\begin{align} \label{TeuAb}
\Omega \slashed{\nabla}_4 (\Omega \slashed{\nabla}_3 \underline{A}) -\Omega^2\slashed{\Delta}\underline{A}+\frac{4(r-3M)}{r^2}\Omega\ns_{4}\underline{A}+\frac{2\Omega^2(r+3M)}{r^3}\underline{A} = 0
\end{align}
While the equations take a nice ``symmetric" form when written in this way, we recall from the well-posedness theory (see (\ref{hozextendQ})) that it is $A$ and $\underline{A} \Omega^{-4}$ which extend regularly to the horizon. For $\hat{\underline{\mathrm{A}}}=\Omega^{-4}\underline{A}$ we have the equation
\begin{align}\label{TeuAbII}
\Omega\ns_4(\Omega\ns_3\hat{\underline{\mathrm{A}}})-\Omega^2\slashed{\Delta}\hat{\underline{\mathrm{A}}}-\frac{4\Omega^2}{r}\Omega\ns_4\hat{\underline{\mathrm{A}}}+\frac{4M}{r^2}\Omega\ns_3\hat{\underline{\mathrm{A}}}+\frac{2 (r-M)\Omega^2}{r^3}\hat{\underline{\mathrm{A}}}=0 \, .
\end{align}

\subsection{Energies and Fluxes}
We define for $v \geq v_0$ the fluxes
\begin{align}
\mathbb{F} [A] (v)&= \int_{u_0}^\infty \Omega^2\left[ |\Omega^{-1} \slashed{\nabla}_3 A|^2  + \frac{|[r\slashed{\nabla}] A|^2+|A|^2}{r^2}  \right] du \varepsilon_{S^2}  \, , \nonumber \\
\mathbb{F} [\underline{A}] (v)&= \int_{u_0}^\infty \Omega^2 \left[ |\Omega^{-1} \slashed{\nabla}_3 (\underline{A}\Omega^{-4})|^2  + \frac{|[r\slashed{\nabla}](\underline{A}\Omega^{-4})|^2+|(\underline{A}\Omega^{-4})|^2}{r^2}  \right]du \varepsilon_{S^2} \, .
\end{align}
Note that these fluxes are the non-degenerate (near the horizon) energy fluxes naturally associated with the Teukolsky operator when contracted with a globally timelike vectorfield which asymptotically becomes $\partial_t$. 

\subsection{The Main Result}
We obtain the following estimates along arbitrary ingoing cones:
\begin{corollary} \label{cor:main}
Under the assumptions of Theorem \ref{theo:gaugeinvariantoutgoing} we have in addition the bounds
\begin{align} \label{Aingoingest}
\sup_{v \geq v_0} \mathbb{F} [A] (v) &\leq C \cdot \mathbb{F} [A] (v_0) + C D \, , 
\\
 \label{Abingoingest}
\sup_{v \geq v_0} \mathbb{F} [\underline{A}] (v) &\leq C \cdot \mathbb{F} [\underline{A}] (v_0) + C D \, , 
\end{align}
where $D:= \sup_{u_f \in [u_0,\infty)} \mathbb{E}^2_{data} (u_f)$ with $\mathbb{E}^2_{data} (u_f)$ the quantity appearing on the right hand side of (\ref{finalalpha}).
\end{corollary}

We will prove the estimate (\ref{Aingoingest}) in Section \ref{sec:Aingoingest} and the estimate (\ref{Abingoingest}) in Section \ref{sec:Abingoingest}. 

We let $N$ be the redshift vectorfield from~\cite{DR}. In particular, for parameters $\sigma>0$, $2M<r_0<7/4M$, we let $N= \chi(r) \left(1+2M \sigma \Omega^2\right) \Omega^{-1} e_3 + \left(1+2M \sigma\Omega^2 \chi(r)\right) T$ where $\chi$ is a radial monotonically decreasing cut-off function equal to $1$ in $[2M, r_0]$ and equal to zero for $r \geq R:=5/2M$. Note that $N$ is uniformly timelike for $r\geq 2M$ and that $N= T$ for $r \geq R$. In the following, $C$ denotes a constant depending only on $\sigma$ and $r_0$, which may change from line to line.

\subsection{Proof of the Estimate on $A$} \label{sec:Aingoingest}

We start by applying the energy estimate associated to the $N$ multiplier for the Teukolsky equation (\ref{TeuA}) in a region $\mathcal{A}(v_1,v_2):=\left[u_0, \infty\right) \times [v_1,v_2] \times \mathbb{S}^2$ for arbitrary $v_0 \leq v_1 \leq v_2 < \infty$.
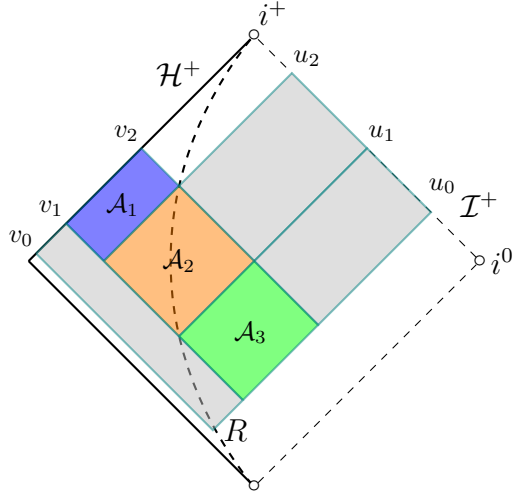
\begin{figure}[H]
\centering
\begin{tikzpicture}
            \draw [thick] (6,3) -- (8.95,5.95);
      \draw [dashed]  (11.95,3.05)--(10.5,4.5);
       \draw [dashed] (9.05,5.95) -- (9.5,5.5);
       \draw [dashed] (9.05,0.05) -- (11.95,2.95);
    \draw [thick] (6,3) -- (8.95,0.05);    
    
       \draw [thick,dashed] (8.95,5.95) to [bend right=38] (8.95,0.05);
        \node  at (8.75,0.75) {\large $R$};

    \draw [thick,teal,fill=lightgray,opacity=0.5] plot  coordinates {(8,4)(9.5,5.5)(10.5,4.5)(9,3)(8,4)};
    
     \draw [thick,teal,fill=lightgray,opacity=0.5] plot  coordinates { (9,3) (10.5,4.5) (11.35,3.65)(9.85,2.15) (9,3)};
     
       \draw [thick,teal,fill=lightgray,opacity=0.5] plot  coordinates {(6.1,3.1)(8.45,0.75)(8.85,1.15)(6.5,3.5) (6.1,3.1)};

    \draw [thick,teal,fill=blue,opacity=0.5] plot  coordinates {(7,3)(8,4)(7.5,4.5)(6.5,3.5)(7,3)};
     \node [ inner sep=-1pt] at (7.25,3.75) {$\mathcal{A}_1$};
    
  \draw [thick,teal,fill=orange,opacity=0.5] plot  coordinates {(7,3)(8,4)(9,3)(8,2)(7,3)};
   \node [ inner sep=-1pt] at (8,3) {$\mathcal{A}_2$};
  
            \draw [thick,teal,fill=green,opacity=0.5] plot  coordinates {(8.85,1.15) (8,2) (9,3)(9.85,2.15) (8.85,1.15)};
              \node [ inner sep=-1pt] at (8.95,2.05) {$\mathcal{A}_3$};

  
 		\node [fill=white, inner sep=-1pt] at (11.5,4) {$u_0$};
         \node [fill=white, inner sep=-1pt] at (10.7,4.7) {$u_1$};
          \node [fill=white, inner sep=-1pt] at (9.7,5.7) {$u_2$};
          
            \node [fill=white, inner sep=-1pt] at (6.3,3.7) {$v_1$};
          \node [fill=white, inner sep=-1pt] at (7.3,4.7) {$v_2$};
           \node [fill=white, inner sep=-1pt] at (5.9,3.3) {$v_0$};

        \node[mark size=2pt,inner sep=2pt] at (12,3) {$\circ$};
      \node[mark size=2pt] at (9,6) {$\circ$};
          \node[mark size=2pt] at (9,0) {$\circ$};
      \node (scriplus) at (12,3.75) {\large $\mathcal{I}^{+}$};
        \node (Hplus) at (8,5.5) {\large $\mathcal{H}^{+}$};    
      \node (iplus) at (9.25,6.3) {\large $i^{+}$};
      \node (inaught) at (12.3,3) {\large $i^{0}$};

      \end{tikzpicture}  
\end{figure}
\noindent Fixing $\sigma$ and $r_0$, we obtain\footnote{For the argument in this subsection any choice of $\sigma, r_0$ will work since one does \emph{not} have to exploit that the spacetime term produced by the $N$-multiplier estimate has a favourable sign in $2M \leq r \leq r_0$ for a suitable choice of $\sigma$ and $r_0$.}
\begin{align} \label{eeA}
\mathbb{F} [A](v_2) \leq C \cdot \mathbb{F} [A](v_1) + CD + C \int_{v_1}^{v_2} \int_{u_0}^\infty \int_{\mathbb{S}^2} \left[ \frac{1}{r} |\slashed{\nabla}_N A| |\slashed{\nabla}_3 A|  + \mathbf{1}_{r \leq R} \Omega^2 |DA|^2 \right]dudv \varepsilon_{S^2}  \, .
\end{align}
Here the first term in the integrand on the right comes from the first order term in the Teukolsky equation while the second term comes from the fact that the vectorfield $N$ is not Killing for $r \leq R$. In particular, $\mathbf{1}_{r \leq R}$ is the indicator function and $|DA|^2 = |\Omega^{-1} \slashed{\nabla}_3 A|^2 + |\Omega \slashed{\nabla}_4 A|^2 + |\slashed{\nabla} A|^2$ denotes the sum of all regular derivatives squared. The first two terms on the right come from controlling the fluxes on $v=v_1$ and $u=u_0$ respectively (using Theorem \ref{theo:gaugeinvariantoutgoing} for the latter). The term on the left comes from the positive future flux on $v=v_2$ (the good term appearing on the horizon we simply dropped).

Denoting $D= \sup_{u_f} \mathbb{E}^2_{data} (u_f)$, we claim that the estimate of Theorem \ref{theo:gaugeinvariantoutgoing} implies the estimate
\begin{align} \label{ims}
\int_{u_0}^\infty \int_{v_1}^{v_2}  \int_{\mathbb{S}^2} du dv  d\theta d\phi \sin \theta \Omega^2  \left[ |\Omega^{-1} \slashed{\nabla}_3 A|^2 + \frac{| \Omega \slashed{\nabla}_4 A|^2 +  |r\slashed{\nabla} A|^2 + |A|^2 }{r^2} \right]   \leq  C D\left(v_2-v_1\right) + C D \, .
\end{align}
To verify (\ref{ims}) we split the domain of integration $\mathcal{A}(v_1,v_2)=\mathcal{A}= \mathcal{A}_1 \cup \mathcal{A}_2 \cup \mathcal{A}_3$ where $\mathcal{A}_1 = \mathcal{A} \cap \{ u \geq v_2 - R^\star\}$, $\mathcal{A}_2 = \mathcal{A} \cap \{ v_1 - R^\star \leq u \leq v_2 - R^\star\}$ and $\mathcal{A}_3 = \mathcal{A} \cap \{ u \leq v_1-R^\star \}$, so that the decomposition is disjoint up to a set of measure zero, as indicated in the figure above. For $\mathcal{A}_1$ one then observes that
\begin{align} 
 & \ \ \ \ \  \int_{v_2-R^\star}^\infty du \int_{v_1}^{v_2} dv \int_{\mathbb{S}^2} \Omega^2  \left[ |\Omega^{-1} \slashed{\nabla}_3 A|^2 + \frac{| \Omega \slashed{\nabla}_4 A|^2 +  |r\slashed{\nabla} A|^2 + |A|^2 }{r^2} \right]  \nonumber \\
&\leq \int_{v_2-R^\star}^\infty du \Omega^2 (u,v_2) \sup_{u \geq v_2 -R^\star} \int_{v_1}^{v_2} dv \int_{\mathbb{S}^2}  \left[ |\Omega^{-1} \slashed{\nabla}_3 A|^2 + \frac{| \Omega \slashed{\nabla}_4 A|^2 +  |r\slashed{\nabla} A|^2 + |A|^2 }{r^2} \right] \leq C(R-2M) D \, , \nonumber
\end{align}
where we have used the monotonicty and angular independence of $\Omega^2$ in the first inequality and the fact that $\Omega^2 = -\partial_u r$ (and of course the uniform boundedness for the outgoing fluxes following from Theorem \ref{theo:gaugeinvariantoutgoing}) in the last. For $\mathcal{A}_2$ we note
\begin{align} 
& \ \ \ \ \ \int_{v_1-R^\star}^{v_2-R^\star} du \int_{v_1}^{v_2} dv \int_{\mathbb{S}^2} \Omega^2  \left[ |\Omega^{-1} \slashed{\nabla}_3 A|^2 + \frac{| \Omega \slashed{\nabla}_4 A|^2 +  |r\slashed{\nabla} A|^2 + |A|^2 }{r^2} \right]  \nonumber \\
&\leq (v_2 - v_1) \sup_{v_1-R^\star \leq u \leq v_2 -R^\star} \int_{v_1}^{v_2} dv \int_{\mathbb{S}^2}  \left[ |\Omega^{-1} \slashed{\nabla}_3 A|^2 + \frac{| \Omega \slashed{\nabla}_4 A|^2 +  |r\slashed{\nabla} A|^2 + |A|^2 }{r^2} \right]  \leq C D (v_2-v_1) \, , \nonumber
\end{align}
where we have used $\Omega^2 \leq 1$ and of course Theorem \ref{theo:gaugeinvariantoutgoing}. Finally for $\mathcal{A}_3$ we note
\begin{align} 
& \ \ \ \ \ \int_{u_0}^{v_1-R^\star} du \int_{v_1}^{v_2} dv \int_{\mathbb{S}^2} \Omega^2  \left[ |\Omega^{-1} \slashed{\nabla}_3 A|^2 + \frac{| \Omega \slashed{\nabla}_4 A|^2 +  |r\slashed{\nabla} A|^2 + |A|^2 }{r^2} \right]  \nonumber \\
&\leq \int_{u_0}^{v_1-R^\star} \frac{1}{r^2} \left(u,v_1\right)du \sup_{u_0 \leq u \leq v_1-R^\star} \int_{v_1}^{v_2} dv \int_{\mathbb{S}^2}   \left[ r^2 |\Omega^{-1} \slashed{\nabla}_3 A|^2 + | \Omega \slashed{\nabla}_4 A|^2 +  |r\slashed{\nabla} A|^2 + |A|^2 \right]  \leq \frac{CD}{R} \, , \nonumber
\end{align}
where we have used the monotonicity of $\frac{1}{r^2}$ in the $v$ direction, the boundedness of $\partial_ur=-\Omega^2$ in $\mathcal{A}_3$ and Theorem \ref{theo:gaugeinvariantoutgoing}. 

Having established (\ref{ims}), it is easy to see that one can invoke the Cauchy-Schwarz inequality to deduce from (\ref{eeA}) the following estimate, valid for any $v_0 \leq v_1 \leq v_2 <\infty$,
\begin{align}
\mathbb{F} [A](v_2) +  \int_{v_1}^{v_2} \mathbb{F} [A](v) dv  \leq C\mathbb{F} [A](v_1) + C D\left(v_2-v_1\right) + C D \, .
\end{align}
A standard argument now gives
\begin{align}
\sup_{v \geq v_0} \mathbb{F}[A] (v) \lesssim  \mathbb{F}[A] (v_0) + CD \, ,
\end{align}
which is the desired (\ref{Aingoingest}).

\subsection{Proof of the Estimate on $\underline{A}$}  \label{sec:Abingoingest}
The argument for $\underline{A}$ is slightly more involved because the fluxes obtained in Theorem \ref{theo:gaugeinvariantoutgoing} do not have optimal $\Omega^2$ weights near the horizon and because the $\slashed{\nabla}_3 \underline{A}$-derivative is not part of the outgoing flux (and cannot possibly be as there is no Bianchi equation for this derivative). However, we can exploit the redshift effect in that the estimate of the vectorfield $N$ will produce a spacetime term of the right sign and with the correct weights near the horizon. The details are as follows. 

We claim that there exists $\sigma>0$ and $r_0$ close to the horizon such that for all $v_0 \leq v_1 \leq v_2 < \infty$ contracting (\ref{TeuAbII}) with $\slashed{\nabla}_N (\Omega^{-4} \underline{A})$ produces the following estimate after integration over the region $\mathcal{A}(v_1,v_2)$ (with $d\mu = dudv d\theta d\phi \sin \theta$):
\begin{align} \label{eeA}
\mathbb{F} [\underline{A}](v_2) +  c \int_{v_1}^{v_2} dv  \mathbb{F} [\underline{A}](v) +  \int_{v_1}^{v_2} \int_{u_0}^\infty \int_{\mathbb{S}^2} d\mu  \, \mathbf{1}_{r \leq r_0} \Omega^2 |\Omega \slashed{\nabla}_4 (\Omega^{-4} \underline{A})|^2 \leq C \cdot \mathbb{F} [\underline{A}](v_1) \nonumber \\
  + C D  + C_{r_0} \int_{v_1}^{v_2} \int_{u_0}^\infty \int_{\mathbb{S}^2} d\mu \,  \, \mathbf{1}_{r \geq r_0} \Omega^2 \left[ |\Omega \slashed{\nabla}_4 (\Omega^{-4} \underline{A})|^2 + \frac{ |r \slashed{\nabla} (\Omega^{-4} \underline{A})|^2 +|\underline{A}|^2 }{r^4} \right]  \, ,
\end{align}
where we recall $D=\sup_{u_f} \mathbb{E}^2_{data} (u_f)$. 
The proof of the estimate (\ref{eeA}) proceeds along the following lines. Upon contraction with $\slashed{\nabla}_N (\Omega^{-4} \underline{A})$ we have the following identity
\begin{align}
0&=\partial_u\Big(\frac{g}{4}|\Omega\ns_4\abh|^2+\Big[\frac{\Omega^2g}{2}+f\Big]\Big[\frac{1}{2}|\ns\abh|^2+\frac{ (r-M)}{r^3}|\abh|^2\Big]\Big)\label{TeuAbIIEI}\\
&\nonumber\qquad+\partial_v\Big(\frac{\Omega^2}{2}\Big[\frac{\Omega^2g}{2}+f\Big]|\Omega^{-1}\ns_3\abh|^2+\frac{\Omega^2g}{2}\Big[\frac{1}{2}|\ns\abh|^2+\frac{ (r-M)}{r^3}|\abh|^2\Big]\Big) \\
&\nonumber\qquad+\Omega^2\Big[\frac{g'}{4}-\frac{2g}{r}\Big]|\Omega\ns_4\abh|^2+\Omega^2\Big[\frac{f'}{2}-\frac{f}{r}\Big]|\ns\abh|^2+\frac{\Omega^2}{r^3}\Big[(r-M)f'-\Big(2-\frac{3M}{r}\Big)f\Big]|\abh|^2\\
&\nonumber\qquad+\Omega^2\Big[\frac{5M}{r^2}f+\frac{2M}{r^2}\Omega^2g-\frac{\Omega^2f'}{2}-\frac{\Omega^2g'}{4}\Big]|\Omega^{-1}\ns_3\abh|^2-\frac{2\Omega^2}{r}\Big[\Big(1-\frac{3M}{r}\Big)g+2f\Big]\langle\Omega^{-1}\ns_3\abh,\Omega\ns_4\abh\rangle.
\end{align}
\begin{itemize}
\item Upon integration in $u,v$, the first two terms of \eqref{TeuAbIIEI} produce two good future boundary terms (fluxes) since $f,g>0$. The good term on the horizon we drop. The past fluxes are easily estimated by $\mathbb{F} [\underline{A}](v_1)$ for the ingoing and by $D$ for the outgoing flux using Theorem \ref{theo:gaugeinvariantoutgoing}.

\item One can check that the coefficient of $|\Omega^{-1}\slashed{\nabla}_3 \abh|^2$-term in the last line is globally positive. Moreover, it controls a spacetime term of the type $Mr^{-2} \Omega^2 |\Omega^{-1} \slashed{\nabla}_3 \abh|^2$.

\item Picking $r_0$ close enough to $2M$ allows one to establish that for $r\leq r_0$ (i.e.~near the horizon), the terms in the third line of (\ref{TeuAbIIEI}) produce three good spacetime terms. Moreover the $|\Omega\ns_4\abh|^2$-term scales with $\sigma$, which we view as a largeness factor (in fact, the $|\ns\abh|$-term also has a largeness factor of $\sigma$).

For $r\geq r_0$, the terms in the third line of (\ref{TeuAbIIEI}) are put on the right-hand side of~\eqref{eeA}. Note that the angular and zeroth order term produce have no contribution to the bulk for $r\geq R$ since $f$ is compactly supported.

\item For $r\leq r_0$, the cross-term in the final line of~\eqref{TeuAbIIEI} can be absorbed into the $|\ns_4\abh|^2$ and $|\Omega^{-1}\ns_3\abh|^2$-term because of the largeness factor of $\sigma$ (see previous step) and the favourable signs.

For $r\geq r_0$, we can apply Young's inequality with $\epsilon$ and absorb the resulting $\epsilon r^{-1}|\Omega^{-1}\ns_3\abh|^2$ using the good global $|r^{-1}\Omega^{-1}\ns_3\abh|^2$-term. The $\slashed{\nabla}_4$-term is put on the right. Note the $r$-weights here.

\end{itemize}
The spacetime term on the right hand side (\ref{eeA}) is now estimated by $C D\left(v_2-v_1\right) + C D$ exactly as in the case of $A$ in previous section using Theorem \ref{theo:gaugeinvariantoutgoing}. This time we may split the region $\mathcal{A}(v_1,v_2)$ into  $\mathcal{A}(v_1,v_2)=\mathcal{A}= \mathcal{A}^\prime_1 \cup \mathcal{A}^\prime_2 \cup \mathcal{A}^\prime_3$ where $\mathcal{A}^\prime_1 = \mathcal{A} \cap \{ u \geq v_2 - r_0^\star\}$, $\mathcal{A}^\prime_2 = \mathcal{A} \cap \{ v_1 - r_0^\star \leq u \leq v_2 - r_0^\star\}$ and $\mathcal{A}^\prime_3 = \mathcal{A} \cap \{ u \leq v_1-r_0^\star \}$, so that again the decomposition is disjoint up to a set of measure zero. The spacetime  term on the right hand side (\ref{eeA}) is still overestimated by integrating over the regions $\mathcal{A}_2^\prime \cup \mathcal{A}_3^\prime$ and those are indeed exactly handled as in the previous section. We conclude
\begin{align} \label{eeAconclude}
\mathbb{F} [\underline{A}](v_2) +  c \int_{v_1}^{v_2} dv  \mathbb{F} [\underline{A}](v)  \leq C \cdot \mathbb{F} [\underline{A}](v_1) +  C D\left(v_2-v_1\right) + C D   \, ,
\end{align}
from which the desired (\ref{Abingoingest}) follows.

\subsection{Uniform Boundedness of $\underline{A}$} \label{sec:ubab}
Recall that uniform boundedness of $A$ was a direct consequence of Theorem \ref{theo:gaugeinvariantoutgoing}, angular commutation and Sobolev embedding, see (\ref{coji}). 
For $\underline{A}$, the estimate (\ref{Abingoingest}) similarly implies 
\begin{align} 
\sup_{u \geq u_0, v \geq v_0} \int_{S^2_{u,v}} \sin \theta d\theta d\phi \cancel{\frac{1}{r}} |\Omega^{-2} \ablin r|^2 \lesssim \mathbb{F}[A] (v_0)  + \sup_{u_f} \mathbb{E}^2_{data} (u_f) \label{coji2} \, ,
\end{align}
first with the crossed factor of $\frac{1}{r}$. The factor $\frac{1}{r}$ can be removed as indicated if one also invokes the estimate (\ref{finalalphab}) along the outgoing cones. After trivial angular commutation and Sobolev embedding on spheres one finally obtains an $L^\infty$ bound for $\underline{A} \Omega^{-4} = \Omega^{-2} \ablin r$. We leave the standard details to the reader.

\pagebreak
\footnotesize
\bibliographystyle{IEEEtranS}
\bibliography{CLbib}

\begin{thebibliography}{10}
\providecommand{\url}[1]{#1}
\csname url@samestyle\endcsname
\providecommand{\newblock}{\relax}
\providecommand{\bibinfo}[2]{#2}
\providecommand{\BIBentrySTDinterwordspacing}{\spaceskip=0pt\relax}
\providecommand{\BIBentryALTinterwordstretchfactor}{4}
\providecommand{\BIBentryALTinterwordspacing}{\spaceskip=\fontdimen2\font plus
\BIBentryALTinterwordstretchfactor\fontdimen3\font minus
  \fontdimen4\font\relax}
\providecommand{\BIBforeignlanguage}[2]{{%
\expandafter\ifx\csname l@#1\endcsname\relax
\typeout{** WARNING: IEEEtranS.bst: No hyphenation pattern has been}%
\typeout{** loaded for the language `#1'. Using the pattern for}%
\typeout{** the default language instead.}%
\else
\language=\csname l@#1\endcsname
\fi
#2}}
\providecommand{\BIBdecl}{\relax}
\BIBdecl

\bibitem{ABBSM19}
L.~Andersson, T.~Bäckdahl, P.~Blue, and S.~Ma, ``{Stability for linearized
  gravity on the Kerr spacetime},'' 2019, arxiv:1903.03859.

\bibitem{AHWModeLE}
L.~Andersson, D.~Häfner, and B.~F. Whiting, ``{Mode analysis for the
  linearized Einstein equations on the Kerr metric : the large $\mathfrak{a}$
  case},'' 2022, arXiv:2207.12952.

\bibitem{AMPBReal}
L.~Andersson, S.~Ma, C.~Paganini, and B.~F. Whiting, ``{Mode stability on the
  real axis},'' \emph{J. Math. Phys.}, vol.~58, no.~7, p. 072501, 2017.

\bibitem{Gabriele1}
G.~Benomio, ``{A new gauge for gravitational perturbations of Kerr spacetimes
  I: The linearised theory},'' 2022, arXiv:2211.00602.

\bibitem{Gabriele2}
------, ``{A new gauge for gravitational perturbations of Kerr spacetimes II:
  The linear stability of Schwarzschild revisited},'' 2022, arXiv:2211.00616.

\bibitem{Chandrasekhar2}
S.~Chandrasekhar, \emph{The {Mathematical} {Theory} of {Black} {Holes}}.\hskip
  1em plus 0.5em minus 0.4em\relax Oxford Univ. Press, New York, 1992.

\bibitem{Chandrasekhar}
------, ``On the equations governing the perturbations of the {Schwarzschild}
  black hole,'' \emph{Proc. Roy. Soc. Lon. A}, vol. 343, no. 1634, pp.
  289--298, 1975.

\bibitem{Chandrasekhar3}
S.~Chandrasekhar and V.~Ferrari, ``{The Flux Integral for Axisymmetric
  Perturbations of Static Space-Times},'' \emph{Proc. Roy. Soc. Lon. A}, vol.
  428, no. 1875, pp. 325--349, 1990.

\bibitem{ChristBH}
D.~Christodoulou, \emph{{The Formation of Black Holes in General
  Relativity}}.\hskip 1em plus 0.5em minus 0.4em\relax EMS, 2009.

\bibitem{CK93}
D.~Christodoulou and S.~Klainerman, \emph{{The Global Nonlinear Stability of
  the Minkowski Space}}.\hskip 1em plus 0.5em minus 0.4em\relax Princeton Uni.
  Press, 1993.

\bibitem{scthesis}
S.~C. Collingbourne, ``{The Gregory--Laflamme Instability and Conservation Laws
  for Linearised Gravity (Doctoral thesis)},'' \emph{University of Cambridge},
  2022.

\bibitem{ColCan}
------, ``{Coercivity Properties of the Canonical Energy in Double Null
  Gauge},'' 2023, to Appear.

\bibitem{DHR2}
M.~Dafermos, G.~Holzegel, and I.~Rodnianski, ``{Boundedness and Decay for the
  Teukolsky Equation on Kerr Spacetimes {I}: The Case $|a|<<{M}$},'' \emph{Ann.
  PDE}, vol.~5, no.~1, p.~2, 2019.

\bibitem{DHR}
------, ``The linear stability of the {Schwarzschild} solution to gravitational
  perturbations,'' \emph{Acta Math.}, vol. 222, no.~1, pp. 1--214, 2019.

\bibitem{DHRT}
M.~Dafermos, G.~Holzegel, I.~Rodnianski, and M.~Taylor, ``{The non-linear
  stability of the Schwarzschild family of black holes},'' 2021,
  arXiv:2104.08222.

\bibitem{DR}
M.~Dafermos and I.~Rodnianski, ``Lectures on {Black} {Holes} and {Linear}
  {Waves},'' \emph{Clay Math. Proc.}, vol.~17, pp. 97--206, 2008.

\bibitem{Friedman78}
J.~L. Friedman, ``{Generic instability of rotating relativistic stars},''
  \emph{{Comm. Math. Phys.}}, vol.~62, no.~3, pp. 247--278, 1978.

\bibitem{Elena}
E.~Giorgi, ``{Boundedness and Decay for the Teukolsky System of Spin $\pm 2$ on
  Reissner--Nordström Spacetime: The Case $Q<< M$},'' \emph{Ann. Hen. Poin.},
  vol.~21, no.~8, pp. 2485--2580, 2020.

\bibitem{GHIW16}
S.~R. Green, S.~Hollands, A.~Ishibashi, and R.~M. Wald, ``{Superradiant
  instabilities of asymptotically anti-de Sitter black holes},'' \emph{Class.
  Quant. Grav.}, vol.~33, no.~12, p. 125022, 2016.

\bibitem{HollandsWald}
S.~Hollands and R.~M. Wald, ``Stability of {Black} {Holes} and {Black}
  {Branes},'' \emph{Comm. Math. Phys.}, vol. 321, no.~3, p. 629, 2013.

\bibitem{ghconslaw}
G.~Holzegel, ``Conservation laws and flux bounds for gravitational
  perturbations of the {Schwarzschild} metric,'' \emph{Class. Quantum Gravity},
  vol.~33, no.~20, p. 205004, 2016.

\bibitem{KS20}
S.~Klainerman and J.~Szeftel, \emph{{Global Nonlinear Stability of
  Schwarzschild Spacetime Under Polarized Perturbations}}.\hskip 1em plus 0.5em
  minus 0.4em\relax Princeton Uni. Press, 2020.

\bibitem{KSK21}
------, ``{Kerr stability for small angular momentum},'' 2021,
  arxiv:2104.11857.

\bibitem{Ma}
S.~Ma, ``{Uniform Energy Bound and Morawetz Estimate for Extreme Components of
  Spin Fields in the Exterior of a Slowly Rotating Kerr Black Hole II:
  Linearized Gravity},'' \emph{Comm. Math. Phys.}, vol. 377, no.~3, pp.
  2489--2551, 2020.

\bibitem{MasaoodI}
H.~Masaood, ``{A Scattering Theory for Linearised Gravity on the Exterior of
  the Schwarzschild Black Hole I: The Teukolsky Equations},'' \emph{{Comm.
  Math. Phys.}}, vol. 393, no.~1, pp. 477--581, 2022.

\bibitem{MasaoodII}
------, ``{A Scattering Theory for Linearised Gravity on the Exterior of the
  Schwarzschild Black Hole II: The Full System},'' 2022, arXiv:2211.07462.

\bibitem{Millet}
P.~Millet, ``{Optimal decay for solutions of the Teukolsky equation on the Kerr
  metric for the full subextremal range $|a| <M$},'' 2023, arxiv:2302.06946.

\bibitem{Yakov2}
Y.~Shlapentokh-Rothman, ``{Quantitative mode stability for the wave equation on
  the Kerr spacetime},'' vol.~16, no.~1, pp. 289--345, 2015.

\bibitem{YakovRita1}
Y.~Shlapentokh-Rothman and R.~Teixeira~da Costa, ``{Boundedness and decay for
  the {Teukolsky} equation on {Kerr} in the full subextremal range $|a|<M$:
  frequency space analysis},'' 2020, arXiv:2007.07211.

\bibitem{YakovRita2}
------, ``{Boundedness and decay for the {Teukolsky} equation on {Kerr} in the
  full subextremal range $|a|<M$: physical space analysis},'' 2023,
  arXiv:2302.08916.

\bibitem{Rita}
R.~Teixeira~da Costa, ``{Mode stability for the Teukolsky equation on extremal
  and subextremal Kerr spacetimes},'' \emph{Commun. Math. Phys.}, vol. 378,
  no.~1, pp. 705--781, 2020.

\bibitem{Teukolsky}
S.~A. {Teukolsky}, ``{Perturbations of a Rotating Black Hole. I. Fundamental
  Equations for Gravitational, Electromagnetic, and Neutrino-Field
  Perturbations},'' \emph{Astrophys. J.}, vol. 185, pp. 635--648, 1973.

\bibitem{Wald}
R.~M. Wald, ``{Construction of Solutions of Gravitational, Electromagnetic, or
  Other Perturbation Equations from Solutions of Decoupled Equations},''
  \emph{Phys. Rev. Lett.}, vol.~41, pp. 203--206, 1978.

\bibitem{Whiting}
B.~F. Whiting, ``Mode stability of the {Kerr} black hole,'' \emph{J. Math.
  Phys.}, vol.~30, no.~6, pp. 1301--1305, 1989.

\bibitem{Pham}
P.~T. Xuan, ``{Conformal scattering theories for tensorial wave equations on
  Schwarzschild spacetime},'' 2022, arXiv:2006.02888.

\end{thebibliography}
\end{document}